\documentclass[letterpaper, 10pt]{article}

\usepackage{amsthm}
\usepackage{amsfonts}
\usepackage{amsmath}
\usepackage{graphicx}
\usepackage{epstopdf}
\usepackage{lineno}
\usepackage{setspace} %\onehalfspacing %Makes 1.5-spacing.
\usepackage{verbatim}
\usepackage{hyperref}

\newtheorem{theorem}{Theorem}[section]
\newtheorem*{theorem*}{Theorem}
\newtheorem{conjecture}[theorem]{Conjecture}
\newtheorem{corollary}[theorem]{Corollary}

\newtheorem{lemma}[theorem]{Lemma}
\newtheorem{prop}{Proposition}

\newtheorem{observation}{Observation}

\newtheorem*{definition}{Definition}

\newcommand{\gbar}{\overline G}
\newcommand{\fs}{\mathcal S}
\newcommand{\fz}{\mathcal Z}

\newcommand{\cC}{\mathcal C}
\newcommand{\cZ}{\mathcal Z}

\newcommand{\TC}{\mathcal{T}\hspace{-.2ex}\mathcal{C}}

\newcommand{\TTC}{\mathcal{T}\hspace{-.2ex}\mathcal{T}\hspace{-.2ex}\mathcal{C}}

\newenvironment{enumerate*}{%reduces item spacing
\begin{enumerate}
  \setlength{\itemsep}{5pt}
  \setlength{\parskip}{0pt}
  \setlength{\parsep}{0pt}
}{\end{enumerate}}

\newenvironment{itemize*}{%reduces item spacing
\begin{itemize}
  \setlength{\itemsep}{5pt}
  \setlength{\parskip}{0pt}
  \setlength{\parsep}{0pt}
}{\end{itemize}}

\title{Claw-free graphs, skeletal graphs, and a stronger conjecture on $\omega$, $\Delta$, and $\chi$}
\author{Andrew D.\ King\thanks{Corresponding author: {\tt andrew.d.king@gmail.com}, Departments of Mathematics and Computing Science, Simon Fraser University, Burnaby, BC.  Supported by a PIMS Postdoctoral Fellowship and the NSERC Discovery Grants of Pavol Hell and Bojan Mohar.}\ \ and Bruce A.\ Reed\thanks{School of Computer Science, McGill University, Montreal. Research supported in part by a Canada Research Chair.}}

\begin{document}

\maketitle

\begin{abstract}
The second author's $\omega$, $\Delta$, $\chi$ conjecture proposes that every graph satisties $\chi \leq \lceil \frac 12 (\Delta+1+\omega)\rceil$.  In this paper we prove that the conjecture holds for all claw-free graphs.  Our approach uses the structure theorem of Chudnovsky and Seymour.

Along the way we discuss a stronger local conjecture, and prove that it holds for claw-free graphs with a three-colourable complement.  To prove our results we introduce a very useful $\chi$-preserving reduction on homogeneous pairs of cliques, and thus restrict our view to so-called {\em skeletal} graphs.
\end{abstract}

%%%%%%%%%%%%%%%%%%%%%%%%%%%%%%%%%%%%%%%%%%%%%%%%%%%%%%%%%%%%%%%%%%%%%%%%%%%%%%%%
%%%%%%%%%%%%%%%%%%%%%%%%%%%%%%%%%%%%%%%%%%%%%%%%%%%%%%%%%%%%%%%%%%%%%%%%%%%%%%%%
%%%%%%%%%%%%%%%%%%%%%%%%%%%%%%%%%%%%%%%%%%%%%%%%%%%%%%%%%%%%%%%%%%%%%%%%%%%%%%%%
%%%%%%%%%%%%%%%%%%%%%%%%%%%%%%%%%%%%%%%%%%%%%%%%%%%%%%%%%%%%%%%%%%%%%%%%%%%%%%%%
%%%%%%%%%%%%%%%%%%%%%%%%%%%%%%%%%%%%%%%%%%%%%%%%%%%%%%%%%%%%%%%%%%%%%%%%%%%%%%%%
\section{Introduction}

In this paper the graphs we consider are simple, loopless, and finite.  The multigraphs we consider are finite and may have loops.  We say that a graph $G$ is {\em claw-free} if it does not contain the complete bipartite graph $K_{1,3}$ as an induced subgraph, i.e.\ if no vertex of $G$ has three mutually nonadjacent neighbours.  Claw-free graphs are a natural generalization of line graphs and quasi-line graphs (which we define in Section \ref{sec:classes}), and have been the subject of substantial interest since Parthasarathy and Ravindra's proof of the Strong Perfect Graph Conjecture for claw-free graphs \cite{parthasarathyr76}.  Chv\'atal and Sbihi \cite{chvatals88} offered the first deep insight into the structure of claw-free graphs, proving a decomposition theorem for Berge claw-free graphs that was later refined by Maffray and Reed \cite{maffrayr99}.

Chudnovsky and Seymour recently gave a refined description of the structure of all claw-free graphs \cite{cssurvey}.  Their structure theorems for claw-free graphs have led to a wealth of recent results, for example a new algorithm for the maximum-weight stable set problem \cite{oriolops08} and new results on the stable set polytope \cite{eisenbrandosv05,gallucciogv08}.

In this paper we give a new bound on the chromatic number $\chi(G)$ when $G$ is claw-free.  The bound is in terms of the maximum degree $\Delta(G)$ and the clique number $\omega(G)$.\\

\noindent{\bf Remark: }Since we first proved these results, which appear in the first author's thesis \cite{kingthesis}, several related results have appeared, e.g.\ \cite{chudnovskykps12, edwardsk12b}.  To minimize the length of this paper we take advantage of this wherever possible.

%%%%%%%%%%%%%%%%%%%%%%%%%%%%%%%%%%%%%%%%%%%%%%%%%%%%%%%%%%%%%%%%%%%%%%%%%%%%%%%%
%%%%%%%%%%%%%%%%%%%%%%%%%%%%%%%%%%%%%%%%%%%%%%%%%%%%%%%%%%%%%%%%%%%%%%%%%%%%%%%%
\subsection{$\omega$, $\Delta$, and $\chi$}

It is easy to show that $\omega(G)\leq \chi(G)\leq \Delta(G)+1$ for any graph.  The second author conjectured that modulo a round-up, $\chi$ is closer to its trivial lower bound than its trivial upper bound \cite{reed98}.  We use $\gamma(G)$ to denote $\chi(G)\leq \lceil \frac 12 (\Delta(G)+1+\omega(G))\rceil$.

\begin{conjecture}[Reed]\label{con:main}
For any graph $G$, $\chi(G)\leq \gamma(G)$.
\end{conjecture}

In 2008 the first author proposed a local strengthening of this conjecture \cite{kingthesis}.  Before stating it we introduce some more notation.  For a vertex $v$, let $\tilde N(v)$ denote the closed neighbourhood of $v$, i.e.\ $\{v\}\cup N(v)$.  For $S\subseteq V(G)$, let $G[S]$ denote the subgraph of $G$ induced on $S$.  Let $\omega(v)$ denote the maximum size of a clique containing $v$, i.e.\ $\omega(G[\tilde N(v)])$.  Finally, let $\gamma_\ell(v)$ denote $\gamma(G[\tilde N(v)])$ and let $\gamma_\ell(G)$ denote $\max_{v\in V(G)}\gamma_\ell(v)$.

\begin{conjecture}[King]\label{con:local}
For any graph $G$, $\chi(G) \leq \gamma_\ell(G)$.
\end{conjecture}

Both conjectures hold in the fractional setting.  Reed proved that any graph satisfies $\chi_f(G)\leq \frac 12(\Delta(G)+1+\omega(G))$ \cite{molloyrbook}.  McDiarmid observed that the proof could be modified to give a stronger result:

\begin{theorem}\label{thm:fractional}
For any graph $G$, $\chi_f(G) \leq \max_{v\in V(G)} \frac 12 (d(v)+1+\omega(v))$.
\end{theorem}

The full proof appears in \cite{kingthesis}, \S 2.2.  Thus we know that for any graph, $$\chi_f(G)\leq \gamma_\ell(G)\leq \gamma(G)\leq \Delta(G)+1.$$

Conjecture \ref{con:main} was proved for line graphs by King, Reed, and Vetta \cite{kingrv07}; we extended this to all quasi-line graphs \cite{kingr08}.  Chudnovsky, King, Plumettaz and Seymour recently proved Conjecture \ref{con:local} for line graphs \cite{chudnovskykps12}; the reductions from \cite{kingr08} also extend this result to all quasi-line graphs.
\begin{theorem}\label{thm:quasiline}
Given a quasi-line graph $G$, we can colour $G$ using at most $\gamma_\ell(G)$ colours in polynomial time.
\end{theorem}
Even more recently, Edwards and King proved that a stronger local version holds in the fractional setting and for quasi-line graphs \cite{edwardsk12b}, and conjectured that it always holds:
 
\begin{conjecture}[Edwards and King]
For any graph $G$, $\chi(G) \leq \max_{uv\in E(G)}\lceil \frac 12 (\gamma_\ell(u)+\gamma_\ell(v))\rceil$.
\end{conjecture}
  
In this paper we prove that Conjecture \ref{con:main} holds for all claw-free graphs, and Conjecture \ref{con:local} holds for all claw-free graphs with a three-colourable complement, i.e.\ {\em three-cliqued} claw-free graphs:

\begin{theorem}\label{thm:main}
For any claw-free graph $G$, $\chi(G)\leq \gamma(G)$.
\end{theorem}

\begin{theorem}\label{thm:local}
For any three-cliqued claw-free graph $G$, $\chi(G)\leq \gamma_\ell(G)$.
\end{theorem}

Furthermore, both proofs yield polynomial-time algorithms.  Theorem \ref{thm:main} complements a recent result of Chudnovsky and Seymour \cite{clawfree6} for claw-free graphs with stability number $\alpha(G)$ at least three:

\begin{theorem}
For any claw-free graph $G$ with $\alpha(G)\geq 3$, $\chi(G)\leq 2\omega(G)$.
\end{theorem}

Thus our result is stronger when $\Delta(G)+1\leq 3\omega(G)$ (in fact this is always the case when $\alpha(G)\geq 4$ or $G$ is three-cliqued).

%%%%%%%%%%%%%%%%%%%%%%%%%%%%%%%%%%%%%%%%%%%%%%%%%%%%%%%%%%%%%%%%%%%%%%%%%%%%%%%%
%%%%%%%%%%%%%%%%%%%%%%%%%%%%%%%%%%%%%%%%%%%%%%%%%%%%%%%%%%%%%%%%%%%%%%%%%%%%%%%%
\subsection{Overview}

The structure theorem for claw-free graphs naturally divides our work into three types of claw-free graphs: those with a three-colourable complement, those constructed as a generalization of a line graph, and some remaining exceptional cases.  Each of the first two categories involves some basic classes and a composition operation, such that every graph in that category is either basic or can be built from the basic graphs using the composition operation.  Therefore our approach is to prove that Conjecture \ref{con:local} holds for the basic classes, then prove that Conjecture \ref{con:main} (and usually Conjecture \ref{con:local}) continues to hold when the composition operations are applied.  Finally we deal with any remaining cases.

Before we do this, we introduce some machinery that allows us to simplify the class of graphs we need to consider.  This is the notion of a {\em nonskeletal homogeneous pair of cliques}, or {\em NHPOC}.  An NHPOC can be thought of as a type of defect or ``fuzziness'', and if one exists in a claw-free graph $G$, we can reduce to a proper claw-free subgraph $G'$ without changing the chromatic number.  Since $\gamma$ and $\gamma_\ell$ are monotone graph invariants, a minimum counterexample to Theorem \ref{thm:main} or Theorem \ref{thm:local} cannot contain an NHPOC.

Nonskeletal (and other) homogeneous pairs of cliques are fundamental to the structure of claw-free graphs because of {\em thickenings}, a method of expanding vertices in claw-free graphs that generalizes the idea of {\em augmentations} introduced by Maffray and Reed \cite{maffrayr99}.  In the next section we introduce thickenings and NHPOCs, and explain how we can restrict our focus to colouring {\em skeletal graphs}.  Using skeletal graphs, we can easily prove that $\chi(G) \leq \gamma_\ell(G)$ for {\em antiprismatic thickenings}, an important class of claw-free graphs with $\alpha \leq 3$.  These include all graphs with $\alpha \leq 2$, which are trivially claw-free.  Thus we spend Section \ref{sec:skeletal} introducing our tools and showing how to apply them effectively to some straightforward classes of claw-free graphs.

In Section \ref{sec:classes} we present some important types of claw-free graphs that are fundamental to later constructions.  In Section \ref{sec:3} we describe claw-free graphs with a three-colourable complement ({\em three-cliqued} claw-free graphs).  They are built from several basic classes by a composition operation known as {\em hex-chains}.  With both three-cliqued claw-free graphs and antiprismatic thickenings, our approach is to remove a stable set $S$ for which $\gamma_\ell(G-S) < \gamma_\ell(G)$.  This is not always possible; some types of three-cliqued graphs take a little more work.  In Section \ref{sec:3} we complete the proof of Theorem \ref{thm:local}, and then move on to proving Theorem \ref{thm:main}.

To do this, we first need to deal with {\em compositions of strips}, whose structure generalizes that of line graphs and quasi-line graphs.  In Section \ref{sec:compositions} we describe their structure and generalize our approach from \cite{kingr08}.  In Section \ref{sec:icosahedral}, we deal with the remaining case: the exceptional class of {\em icosahedral thickenings} (we deal with these after compositions of strips in order to introduce a certain decomposition where it is most sensible).  This allows us to complete the proof of Theorem \ref{thm:main}.  Finally, in Section \ref{sec:algorithmic} we prove that our approach yields polynomial-time algorithms for constructing colourings that achieve our new bounds.

%%%%%%%%%%%%%%%%%%%%%%%%%%%%%%%%%%%%%%%%%%%%%%%%%%%%%%%%%%%%%%%%%%%%%%%%%%%%%%%%
%%%%%%%%%%%%%%%%%%%%%%%%%%%%%%%%%%%%%%%%%%%%%%%%%%%%%%%%%%%%%%%%%%%%%%%%%%%%%%%%
%%%%%%%%%%%%%%%%%%%%%%%%%%%%%%%%%%%%%%%%%%%%%%%%%%%%%%%%%%%%%%%%%%%%%%%%%%%%%%%%
%%%%%%%%%%%%%%%%%%%%%%%%%%%%%%%%%%%%%%%%%%%%%%%%%%%%%%%%%%%%%%%%%%%%%%%%%%%%%%%%
%%%%%%%%%%%%%%%%%%%%%%%%%%%%%%%%%%%%%%%%%%%%%%%%%%%%%%%%%%%%%%%%%%%%%%%%%%%%%%%%
%%%%%%%%%%%%%%%%%%%%%%%%%%%%%%%%%%%%%%%%%%%%%%%%%%%%%%%%%%%%%%%%%%%%%%%%%%%%%%%%
%%%%%%%%%%%%%%%%%%%%%%%%%%%%%%%%%%%%%%%%%%%%%%%%%%%%%%%%%%%%%%%%%%%%%%%%%%%%%%%%
%%%%%%%%%%%%%%%%%%%%%%%%%%%%%%%%%%%%%%%%%%%%%%%%%%%%%%%%%%%%%%%%%%%%%%%%%%%%%%%%
%%%%%%%%%%%%%%%%%%%%%%%%%%%%%%%%%%%%%%%%%%%%%%%%%%%%%%%%%%%%%%%%%%%%%%%%%%%%%%%%
%%%%%%%%%%%%%%%%%%%%%%%%%%%%%%%%%%%%%%%%%%%%%%%%%%%%%%%%%%%%%%%%%%%%%%%%%%%%%%%%
%%%%%%%%%%%%%%%%%%%%%%%%%%%%%%%%%%%%%%%%%%%%%%%%%%%%%%%%%%%%%%%%%%%%%%%%%%%%%%%%
%%%%%%%%%%%%%%%%%%%%%%%%%%%%%%%%%%%%%%%%%%%%%%%%%%%%%%%%%%%%%%%%%%%%%%%%%%%%%%%%
%%%%%%%%%%%%%%%%%%%%%%%%%%%%%%%%%%%%%%%%%%%%%%%%%%%%%%%%%%%%%%%%%%%%%%%%%%%%%%%%
%%%%%%%%%%%%%%%%%%%%%%%%%%%%%%%%%%%%%%%%%%%%%%%%%%%%%%%%%%%%%%%%%%%%%%%%%%%%%%%%
%%%%%%%%%%%%%%%%%%%%%%%%%%%%%%%%%%%%%%%%%%%%%%%%%%%%%%%%%%%%%%%%%%%%%%%%%%%%%%%%
\section{Skeletal graphs and thickenings}\label{sec:skeletal}

Chudnovsky and Seymour introduced thickenings, which generalize the operations of augmentation and multiplication, as a way to distill the essential structure of a graph or trigraph \cite{clawfree5}.  Here we describe thickenings and discuss how to reduce non-minimal structure that arises as a result of the thickening operation.

We {\em multiply} a vertex $v$ by taking the disjoint union of $G-v$ and a nonempty clique $I(v)$, then making each vertex of $N(v)$ adjacent to each vertex of $I(v)$.  In this case any two vertices of $I(v)$ are {\em twins}, i.e.\ they have the same closed neighbourhood.  A clique $C$ is a {\em homogeneous clique} if it has size between $2$ and $n-1$, and every vertex outside $C$ sees either none or all of $C$.  So as long as $I(v)$ is not a singleton or the entire graph, it is a homogeneous clique.  Note that vertex multiplication will never introduce a claw when applied to a claw-free graph.

To generalize this operation, we consider edges whose deletion does not introduce a claw.  We say that an edge $e$ in a claw-free graph $G$ is {\em claw-neutral} if $G-e$ is claw-free.  A matching $M$ is {\em claw-neutral} if every edge of $M$ is claw-neutral.  Observe that if $M$ is claw-neutral, then $G-M$ is claw-free.  

Let $M$ be a claw-neutral matching in a claw-free graph $G$.  We say that $G'$ is a {\em thickening of $G$ under $M$} (or sometimes just a thickening of $G$) if we can construct it from $G$ in the following way.  First we multiply each vertex.  Then for every $uv\in M$, we remove from $G'$ a nonempty proper subset of the edges between $I(u)$ and $I(v)$.  If $M$ is empty we say that $G'$ is a {\em proper thickening} of $G$; in this case $G'$ simply arises from $G$ by vertex multiplication.  For a set $S\subseteq V(G)$ we use $I(S)$ to denote $\cup_{v\in S}I(v)$.%  Chudnovsky and Seymour introduced thickenings and actually characterized thickenings of {\em claw-free trigraphs}, in which two vertices connected by an edge of $M$ are declared to be ``semi-adjacent'' \cite{clawfree5}.

Just as proper thickenings give rise to homogeneous cliques, thickenings give rise to {\em homogeneous pairs of cliques}.  A pair $(A,B)$ of disjoint nonempty cliques is a homogeneous pair of cliques if $|A\cup B|\geq 3$ and every vertex outside $A\cup B$ sees all or none of $A$, and all or none of $B$.  So for $u,v\in V(G)$, if $|I(u)|+|I(v)|\geq 3$ then $(I(u),I(v))$ is a homogeneous pair of cliques regardless of whether or not $uv\in E(G)$ or $uv\in M$.

It turns out that in a minimum counterexample to Theorem \ref{thm:main} or \ref{thm:local}, we can guarantee that every homogeneous pair of cliques has a very simple structure.  We address this issue now.

%%%%%%%%%%%%%%%%%%%%%%%%%%%%%%%%%%%%%%%%%%%%%%%%%%%%%%%%%%%%%%%%%%%%%%%%%%%%%%%%
%%%%%%%%%%%%%%%%%%%%%%%%%%%%%%%%%%%%%%%%%%%%%%%%%%%%%%%%%%%%%%%%%%%%%%%%%%%%%%%%
%%%%%%%%%%%%%%%%%%%%%%%%%%%%%%%%%%%%%%%%%%%%%%%%%%%%%%%%%%%%%%%%%%%%%%%%%%%%%%%%
\subsection{Skeletal graphs and skeletal homogeneous pairs}

Given a homogeneous pair of cliques $(A,B)$ in a graph $G$, we want to remove edges between $A$ and $B$ in $G$ to reach a subgraph $G'$ such that:
\begin{itemize*}
\item $G'$ is easier to describe and colour than $G$
\item given a $k$-colouring of $G'$ we can easily find a $k$-colouring of $G$.
\end{itemize*}

In this paper we use two such reductions.  A homogeneous pair of cliques $(A,B)$ is {\em linear}\footnote{These were originally called {\em nontrivial} homogeneous pairs of cliques by Chudnovsky and Seymour, who used them in their description of quasi-line graphs \cite{cssurvey}.  We prefer the more descriptive term {\em nonlinear} in part because they are less trivial than {\em skeletal} homogeneous pairs of cliques.} precisely if $G[A,B]$ contains no induced $C_4$ (equivalently, $G[A\cup B]$ is a {\em linear interval graph}, which we define later).  Chudnovsky and Seymour used these to describe quasi-line graphs \cite{cssurvey}, and Chudnovsky and Fradkin used them to colour quasi-line graphs \cite{chudnovskyo07}, as did we \cite{kingr08}.

For claw-free graphs we need a stronger reduction.  Observe that if we remove an edge between $A$ and $B$ without changing the chromatic number of the subgraph induced on $A\cup B$, the chromatic number of the graph will not change.  Furthermore, since $G[A\cup B]$ is cobipartite and therefore perfect, $\chi(G[A\cup B])= \omega(G[A\cup B])$.  We say that $(A,B)$ is {\em skeletal} if we cannot remove an edge between $A$ and $B$ without changing the clique number of $G[A\cup B]$.  We say that $G$ is {\em skeletal} if it contains no nonskeletal homogeneous pair of cliques.  Observe that every skeletal homogeneous pair of cliques is linear.

Now for the reduction result.  The following theorem immediately implies that a minimum counterexample to Theorem \ref{thm:main} or Theorem \ref{thm:local} must be skeletal.

\begin{theorem}\label{thm:skelhp}
Let $G$ be a nonskeletal graph.  Then there is a skeletal subgraph $G'$ of $G$ such that:
\begin{enumerate*}
\item If $G$ is quasi-line (resp.\ claw-free) then $G'$ is also quasi-line (resp.\ claw-free).
\item $\chi(G')=\chi(G)$ and $\chi_f(G')=\chi_f(G)$.
\item If $\chi(\overline{G})=3$ then $\chi(\overline{G'})=3$.
\end{enumerate*}
Furthermore we can find $G'$ in $O(m(m^2+n^{5/2}))$ time, and given a $k$-colouring of $G'$ we can construct a $k$-colouring of $G$ in $O(mn^{5/2})$ time.
\end{theorem}

This theorem follows immediately from at most $m$ applications of the following two lemmas.

\begin{lemma}\label{lem:findhp}
For any graph $G$, we can find a nonskeletal homogeneous pair of cliques, or determine that none exists, in $O(m^2)$ time.
\end{lemma}

\begin{lemma}\label{lem:reduction}
Given a graph $G$ and a nonskeletal homogeneous pair of cliques $(A,B)$, in $O(n^{5/2})$ time we can remove edges between $A$ and $B$ to reach a proper subgraph $G'$ such that:
\begin{enumerate*}
\item $(A,B)$ is a skeletal homogeneous pair of cliques in $G'$.
\item If $G$ is quasi-line (resp.\ claw-free) then $G'$ is also quasi-line (resp.\ claw-free).
\item $\chi(G')=\chi(G)$ and $\chi_f(G')=\chi_f(G)$.
\item If $\chi(\overline{G})=3$ then $\chi(\overline{G'})=3$.
\end{enumerate*}
Furthermore given a $k$-colouring of $G'$ we can construct a $k$-colouring of $G$ in $O(n^{5/2})$ time.
\end{lemma}

Theorem \ref{thm:skelhp} strengthens Lemma 9 from \cite{kingr08}, which itself expands on Lemma 5.1 from \cite{chudnovskyo07}.  We defer the proofs of Lemmas \ref{lem:findhp} and \ref{lem:reduction} to Section \ref{sec:lemmas}.  If we only wanted to reduce nonlinear homogeneous pairs of cliques, we could use the faster and more sophisticated algorithm from \cite{chudnovskyk11}.

%%%%%%%%%%%%%%%%%%%%%%%%%%%%%%%%%%%%%%%%%%%%%%%%%%%%%%%%%%%%%%%%%%%%%%%%%%%%%%%%
%%%%%%%%%%%%%%%%%%%%%%%%%%%%%%%%%%%%%%%%%%%%%%%%%%%%%%%%%%%%%%%%%%%%%%%%%%%%%%%%
%%%%%%%%%%%%%%%%%%%%%%%%%%%%%%%%%%%%%%%%%%%%%%%%%%%%%%%%%%%%%%%%%%%%%%%%%%%%%%%%
\subsubsection{The importance of being skeletal}\label{sec:skelsub}

If $(A,B)$ is skeletal then the edges between $A$ and $B$ are contained in a single clique $\Omega(A,B)$, which we consider to be empty if there are no edges between $A$ and $B$ (see Figure \ref{fig:hpoc}).  Thus $A\cup B$ can be partitioned into the four sets $A\cap \Omega(A,B)$, $B\cap \Omega(A,B)$, $A\setminus \Omega(A,B)$, $B\setminus \Omega(A,B)$, each of which is a homogeneous clique, a singleton, or empty.  For convenience, when talking about a thickening we often use $\Omega(v_i,v_j)$ to denote $\Omega(I(v_i),I(v_j))$.  We now explain why the structure of a skeletal homogeneous pair of cliques is so useful.

\begin{figure}
\begin{center}
\includegraphics[scale=0.6]{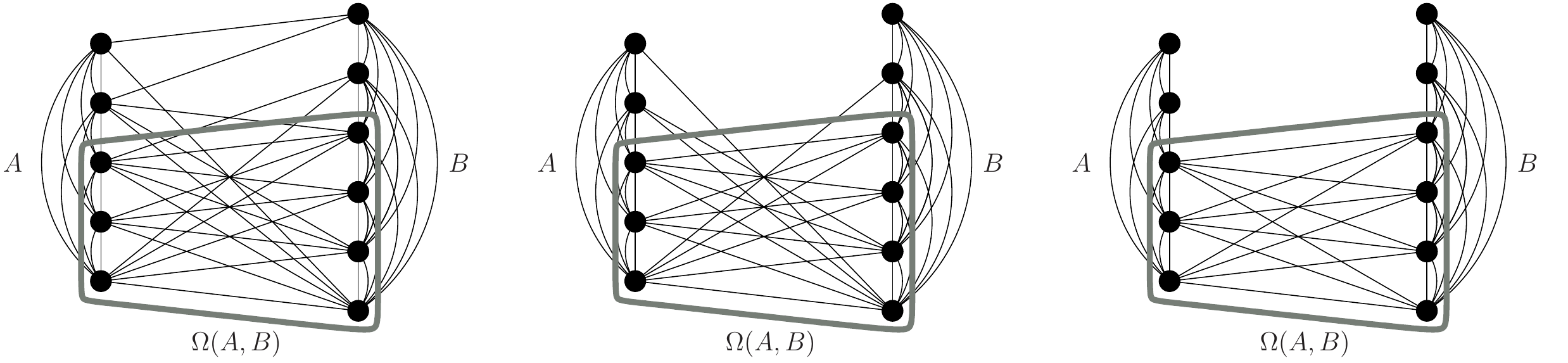}
\end{center}
\caption{\small{Three homogeneous pairs of cliques:  one nonlinear (left),  one nonskeletal linear (middle), and one skeletal (right).  We reduce a nonskeletal homogeneous pair of cliques $(A,B)$ by removing edges without changing the size of a maximum clique in $G[A\cup B]$.}}
\label{fig:hpoc}
\end{figure}

Our approach to colouring often involves removing a stable set $S$ from a supposedly minimum counterexample $G$ and confirming that for a given vertex set $C$, the removal of $S$ causes $\max_{v\in C}(d(v)+\omega(v))$ to drop by two.  We can easily insist that $S$ be a maximal stable set, so $d(v)+\omega(v)$ drops by at least one for every vertex in $G-S$.  In this case, removing $S$ lowers $\max_{v\in C}\gamma_\ell(v)$.  Thus we only need to worry about vertices in $C$ maximizing $d(v)+\omega(v)$.  In particular, if there is a vertex $v$ in $C$ whose closed neighbourhood properly contains the closed neighbourhood of another vertex $v'$, we can safely disregard $v'$ in our analysis.  In this case we say that $v$ {\em trumps}\index{trump} $v'$.  

Now consider the vertices in a skeletal homogeneous pair of cliques $(A,B)$.  We can make several simple observations, all of which are symmetric with respect to $A$ and $B$:
\begin{enumerate*}
\item Every vertex in $A\setminus \Omega(A,B)$ is trumped by every vertex in $A\cap \Omega(A,B)$.
\item Removing a vertex from $A\cap \Omega(A,B)$ lowers $d(v)$ for any $v\in A\cup \Omega(A,B)$.
\item Removing a vertex from $A\cap \Omega(A,B)$ lowers $\omega(v)$ for any $v\in A$.
\item Removing a vertex from $A\cap \Omega(A,B)$ and a vertex from $B\setminus \Omega(A,B)$ lowers $d(v)$ by two for any $v\in B\cap \Omega(A,B)$, and lowers $\omega(v)$ for any $v\in B\setminus \Omega(A,B)$.  In particular, it lowers $\max_{v\in A\cup B}(d(v)+\omega(v))$ by two.
\end{enumerate*}

We now prove that Theorem \ref{thm:local} holds for {\em antiprismatic thickenings} by exploiting the simplicity of skeletal homogeneous pairs of cliques.\\

%%%%%%%%%%%%%%%%%%%%%%%%%%%%%%%%%%%%%%%%%%%%%%%%%%%%%%%%%%%%%%%%%%%%%%%%%%%%%%%%
%%%%%%%%%%%%%%%%%%%%%%%%%%%%%%%%%%%%%%%%%%%%%%%%%%%%%%%%%%%%%%%%%%%%%%%%%%%%%%%%
%%%%%%%%%%%%%%%%%%%%%%%%%%%%%%%%%%%%%%%%%%%%%%%%%%%%%%%%%%%%%%%%%%%%%%%%%%%%%%%%
\subsection{Antiprismatic thickenings}

A {\em triad} is a stable set of size three.  A graph $G$ is {\em antiprismatic} if every triad $T$ contains exactly two neighbours of every vertex in $G-T$.  Such graphs are clearly claw-free, and they were described in detail by Chudnovsky and Seymour \cite{clawfree1, clawfree2}.  We say that an edge $e=uv$ in an antiprismatic graph $G$ is {\em changeable} if $G-e$ is also antiprismatic.  If this is the case, then (i) in $G$, neither $u$ nor $v$ is in a triad, and (ii) in $G-e$, $u$ and $v$ are in at most one triad (see \cite{clawfree1}, \S 16).

Given a matching $M$, we say that $M$ is a {\em changeable matching} in $G$ if for every $M'\subseteq M$, $G-M'$ is antiprismatic.  If $M$ is a changeable matching in $G$, then $M$ is claw-neutral in $G$.  If $G'$ is a thickening of an antiprismatic graph $G$ under a changeable matching $M$, then we say that $G'$ is an {\em antiprismatic thickening}.  In this section we prove that $\chi \leq \gamma_\ell$ for antiprismatic thickenings.

%%%%%%%%%%%%%%%%%%%%%%%%%%%%%%%%%%%%%%%%%%%%%%%%%%%%%%%%%%%%%%%%%%%%%%%%%%%%%%%%
\subsubsection{The case $\alpha\leq 2$}

We begin with trivially antiprismatic graphs, i.e.\ graphs containing no triad.  In these graphs, a colouring corresponds to a matching in the complement, and we can therefore appeal to well-known results in matching theory.

\begin{theorem}\label{thm:2local}
Let $G$ be any graph with $\alpha(G)\leq 2$.  Then $\chi(G)\leq \gamma_\ell(G)$.
\end{theorem}

Our proof relies on the observation that an optimal colouring of a graph with $\alpha\leq 2$ corresponds to a maximum matching in the complement $\gbar$.  Rabern \cite{rabern08} independently proved that $\chi \leq \gamma$ for such graphs using a similar approach.

\begin{proof}[Proof of Theorem \ref{thm:2local}]
Let $G$ be a minimum counterexample to the theorem.  Applying the Edmonds-Gallai structure theorem (\cite{edmonds65b, gallai59}, see also \cite{kingthesis} \S 2.5) for maximum matchings tells us that either (i) there is a vertex $v\in G$ such that $\chi(G)=\chi(G-v)$, (ii) $\gbar$ is not connected, or (iii) $\gbar$ has a matching of size $\lfloor \frac n2 \rfloor$ and consequently $\chi(G) = \lceil \frac n2 \rceil$.  Minimality of $G$ tells us that (i) is impossible.

Suppose $\gbar$ is not connected.  Then $V(G)$ can be partitioned into nonempty $V_1$ and $V_2$ such that $V_1$ is joined to $V_2$, i.e.\ every possible edge between $V_1$ and $V_2$ exists.  It is easy to confirm that $\chi(G)=\chi(G[V_1])+\chi(G[V_2]) \leq \gamma_\ell(G[V_1])+\gamma_\ell(G[V_2]) \leq \gamma_\ell(G)$, the middle inequality following from the minimality of $G$.

Therefore (iii) must be the case, so $\chi(G)=\lceil \frac n2 \rceil$.  Since $\chi_f(G) \geq \frac{n}{\alpha(G)}$, we have $\chi(G)= \lceil \chi_f(G)\rceil$.  By Theorem \ref{thm:fractional},
$$\chi(G)\leq \lceil \chi_f(G) \rceil \leq  \max_{v\in V(G)}\left\lceil \tfrac 12 (d(v)+1+\omega(v))\right\rceil.
$$
This proves the theorem.
\end{proof}

It is not hard to prove the case $\chi(G)=\lceil \frac n 2 \rceil$ without using Theorem \ref{thm:fractional}.  However, this application of Theorem \ref{thm:fractional} is a useful trick and we will use it again later in the paper.

%%%%%%%%%%%%%%%%%%%%%%%%%%%%%%%%%%%%%%%%%%%%%%%%%%%%%%%%%%%%%%%%%%%%%%%%%%%%%%%%
\subsubsection{The case $\alpha=3$}

It remains to show that $\chi(G) \leq \gamma_\ell(G)$ for any antiprismatic thickening $G$ containing a triad.  This case is fairly easy, and is a perfect example of a method we will use repeatedly:  Given a supposed minimum counterexample $G$, we remove a stable set $T$ (in this case a triad) such that $\gamma_\ell(G-T)< \gamma_\ell(G)$.  This immediately contradicts the minimality of our supposed counterexample, since we can make the triad $T$ a colour class in a $\chi(G-T)+1$ colouring of $G$.  We first define the type of triad we seek; we will use them repeatedly.  Recall from Section \ref{sec:skelsub} that a vertex $u$ {\em trumps} a vertex $v$ if $\tilde N(v) \subset \tilde N(u)$.

\begin{definition}
Let $T$ be a triad in a graph $G$.  If every vertex $v$ in $G-T$ has two neighbours in $T$ or a twin in $T$ or is trumped by a vertex in $T$, then we say that $T$ is a {\em good triad}.
\end{definition}

Observe that any good triad $T$ has the property that $\gamma_\ell(G-T)\leq \gamma_\ell(G)-1$.

\begin{theorem}\label{thm:antiprismaticlocal}
Let $G$ be an antiprismatic thickening.  Then $\chi(G) \leq \gamma_\ell(G)$.
\end{theorem}

\begin{proof}
Let $G$ be a minimum counterexample to the theorem.  We already know that $\alpha(G)=3$.  If $G$ contains a good triad $T$, then since $\chi(G-T)\leq \gamma_\ell(G-T)$ and $\chi(G-T)\geq \chi(G)-1$, we know that $\chi(G)\leq \gamma_\ell(G)$.  Therefore to reach a contradiction it suffices to prove the existence of a good triad.  Suppose $G$ is a thickening of an antiprismatic graph $H$ under a changeable matching $M$.

Suppose there is a triad $\{u,v,w\}$ in $H$.  Then note that by the properties of a changeable edge, none of $u,v,w$ is an endpoint of any edge $e$ in $M$: the other endpoint $y$ would either form a claw with $T$, or $y$ would have only one neighbour in $T$ in $G-e$, contradicting the fact that $M$ is changeable.  Let $T$ be a triad in $I(u)\cup I(v)\cup I(w)$.  Every vertex in $(I(u)\cup I(v)\cup I(w))\setminus T$ has a twin in $T$, and every vertex in $G-(I(u)\cup I(v)\cup I(w))$ has two neighbours in $T$.  Therefore $T$ is a good triad and we are done.

So there is no triad in $H$.  Since $\alpha(G)=3$, there are vertices $u,v,w$ in $H$ such that $e=uv\in M$ and $\{u,v,w\}$ is a triad in $H-e$.  By the definition of a thickening, $I(u)\cup I(v)$ is not a clique but there is at least one edge between $I(u)$ and $I(v)$.

We claim that $(I(u),I(v))$ is a skeletal homogeneous pair of cliques in $G$.  For if this is not the case, Lemma \ref{lem:reduction} tells us that we can remove edges between $I(u)$ and $I(v)$ to reach a proper subgraph $G'$ of $G$ with $\chi(G')=\chi(G)$; one can easily confirm that $G'$ is either a thickening of $H$ under $M$, or a thickening of $H-e$ under $M-e$.  Either way, $G'$ is an antiprismatic thickening and contradicts the minimality of $G$.  Therefore $(I(u),I(v))$ is skeletal, $\Omega(u,v)$ is nonempty, and at least one of $I(u)\setminus \Omega(u,v)$ and $I(v)\setminus \Omega(u,v)$ is nonempty.  Assume $I(u)\setminus \Omega(u,v)$ is nonempty.  Let $a,b,c\in V(G)$ be vertices in $I(u)\setminus \Omega(u,v)$, $I(v)\cap \Omega(u,v)$, and $I(w)$ respectively, and note that $T=\{a,b,c\}$ is a triad.  It suffices to show that it is a good triad, which we do now.

Observe that $w$ is not in $V(M)$, for if there were an edge $wx\in M$ then since $H-e$ is antiprismatic, $x$ would have two neighbours in $\{u,v,w\}$ in $H-e$, contradicting the fact that $H-e-wx$ must also be antiprismatic since $M$ is a changeable matching in $H$.  Since $H-e$ is antiprismatic, any vertex of $G$ without two neighbours in $T$ must be in $I(u)\cup I(v)\cup I(w)$.  Therefore a vertex in $I(w)\setminus T$ has a twin in $T$, a vertex in $I(u)\setminus T$ has two neighbours or a twin in $T$ (depending on whether or not it is in $\Omega(u,v)$), and a vertex in $I(v)\setminus T$ has a twin in $T$ or is trumped by a vertex in $T$ (again depending on whether or not it is in $\Omega(u,v)$).  Therefore $T$ is a good triad and we are done.
\end{proof}

The proof actually implies a slightly different result, which is worth stating separately:

\begin{corollary}
Let $G$ be a skeletal antiprismatic thickening with $\alpha(G)\geq 3$.  Then $G$ contains a good triad.
\end{corollary}

In Section \ref{sec:algorithmic} we will show that given an antiprismatic thickening $G$ of an antiprismatic graph $H$ under a changeable matching $M$, we can find $H$ and $M$ in polynomial time.

%%%%%%%%%%%%%%%%%%%%%%%%%%%%%%%%%%%%%%%%%%%%%%%%%%%%%%%%%%%%%%%%%%%%%%%%%%%%%%%%
%%%%%%%%%%%%%%%%%%%%%%%%%%%%%%%%%%%%%%%%%%%%%%%%%%%%%%%%%%%%%%%%%%%%%%%%%%%%%%%%
%%%%%%%%%%%%%%%%%%%%%%%%%%%%%%%%%%%%%%%%%%%%%%%%%%%%%%%%%%%%%%%%%%%%%%%%%%%%%%%%
%%%%%%%%%%%%%%%%%%%%%%%%%%%%%%%%%%%%%%%%%%%%%%%%%%%%%%%%%%%%%%%%%%%%%%%%%%%%%%%%
%%%%%%%%%%%%%%%%%%%%%%%%%%%%%%%%%%%%%%%%%%%%%%%%%%%%%%%%%%%%%%%%%%%%%%%%%%%%%%%%
%%%%%%%%%%%%%%%%%%%%%%%%%%%%%%%%%%%%%%%%%%%%%%%%%%%%%%%%%%%%%%%%%%%%%%%%%%%%%%%%
%%%%%%%%%%%%%%%%%%%%%%%%%%%%%%%%%%%%%%%%%%%%%%%%%%%%%%%%%%%%%%%%%%%%%%%%%%%%%%%%
%%%%%%%%%%%%%%%%%%%%%%%%%%%%%%%%%%%%%%%%%%%%%%%%%%%%%%%%%%%%%%%%%%%%%%%%%%%%%%%%
%%%%%%%%%%%%%%%%%%%%%%%%%%%%%%%%%%%%%%%%%%%%%%%%%%%%%%%%%%%%%%%%%%%%%%%%%%%%%%%%
%%%%%%%%%%%%%%%%%%%%%%%%%%%%%%%%%%%%%%%%%%%%%%%%%%%%%%%%%%%%%%%%%%%%%%%%%%%%%%%%
%%%%%%%%%%%%%%%%%%%%%%%%%%%%%%%%%%%%%%%%%%%%%%%%%%%%%%%%%%%%%%%%%%%%%%%%%%%%%%%%
%%%%%%%%%%%%%%%%%%%%%%%%%%%%%%%%%%%%%%%%%%%%%%%%%%%%%%%%%%%%%%%%%%%%%%%%%%%%%%%%
%%%%%%%%%%%%%%%%%%%%%%%%%%%%%%%%%%%%%%%%%%%%%%%%%%%%%%%%%%%%%%%%%%%%%%%%%%%%%%%%
%%%%%%%%%%%%%%%%%%%%%%%%%%%%%%%%%%%%%%%%%%%%%%%%%%%%%%%%%%%%%%%%%%%%%%%%%%%%%%%%
\section{Some important types of claw-free graphs}\label{sec:classes}

To fully describe skeletal claw-free graphs we must first define some fundamental subclasses, the first of which was antiprismatic thickenings.  Here we describe line graphs, linear and circular interval graphs, and antihat thickenings.

%%%%%%%%%%%%%%%%%%%%%%%%%%%%%%%%%%%%%%%%%%%%%%%%%%%%%%%%%%%%%%%%%%%%%%%%%%%%%%%%
\subsection{Line graphs}

Given a multigraph $H$, its {\em line graph} $L(H)$ is the graph with one vertex for each edge of $H$, in which two vertices are adjacent precisely if their corresponding edges in $H$ share at least one endpoint.  We say that $G$ is a line graph if $G=L(H)$ for some multigraph $H$.  Thus the neighbours of any vertex $v$ in a line graph $L(H)$ are covered by two cliques, one for each endpoint of the edge in $H$ corresponding to $v$.  Observe that every line graph is claw-free.  When considering the line graph of $H$ we may assume that $H$ is loopless, since replacing a loop with a pendant edge in $H$ will not change $L(H)$.

Suppose $G$ is the line graph of $H$, and that $G$ contains a matching $M$ in which each edge corresponds to the two edges in $H$ incident to some vertex of degree $2$.  Then $M$ is a claw-neutral matching, and any thickening of $G$ under $M$ is a thickening what Chudnovsky and Seymour call a {\em thickening of a line trigraph} \cite{clawfree5}.  Now suppose $G'$ is a skeletal thickening of $G$ under $M$.  We claim that $G'$ is actually a line graph as well:

\begin{prop}\label{prop:line}
If a graph $G'$ is a thickening of a line trigraph and is skeletal, then $G$ is a line graph.
\end{prop}
\begin{proof}
Let $G'$ be a skeletal thickening of a line graph $G$ under a matching $M$ as described in the paragraph above.
Consider an edge $uv\in M$ and the corresponding homogeneous pair of cliques $(I(u),I(v))$ in $G'$.  Every vertex in $(I(u)\cup I(v)) \setminus \Omega(u,v)$ is simplicial.  Therefore $G'$ is a thickening of a line graph $L(H')$ under a matching $M\setminus \{uv\}$, where $H'$ is constructed from $H$ looking at the unshared endpoints of $u$ and $v$ and adding a pendant edge to each.  Repeating this process for each edge in $M$ proves the claim.
\end{proof}

It is useful to bear this fact in mind when we define the class $\TTC_1$ in Section \ref{sec:3}.

%%%%%%%%%%%%%%%%%%%%%%%%%%%%%%%%%%%%%%%%%%%%%%%%%%%%%%%%%%%%%%%%%%%%%%%%%%%%%%%%
\subsection{Linear interval graphs, circular interval graphs, and quasi-line graphs}

One class of graphs lying between line graphs and claw-free graphs is the class of {\em quasi-line graphs}.  A graph is quasi-line if the neighbourhood of every vertex induces the complement of a bipartite graph.  We now present two fundamental types of quasi-line graphs.

A {\em linear interval graph} is a graph $G=(V,E)$ with a {\em linear interval representation}, which is a point on the real line for each vertex and a set of intervals such that vertices $u$ and $v$ are adjacent in $G$ precisely if there is an interval containing both corresponding points on the real line.  Linear interval graphs are chordal and therefore perfect.

In the same vein, a {\em circular interval graph} is a graph with a {\em circular interval representation}, which consists of $|V|$ points on the unit circle and a set of intervals (arcs) on the unit circle such that two vertices of $G$ are adjacent precisely if some arc contains both corresponding points.  This class contains all linear interval graphs.  Deng, Hell, and Huang proved that we can identify and find a representation of a circular or linear interval graph in linear time \cite{denghh96}.

A circular interval graph is a {\em long circular interval graph} if it has a circular interval representation in which no three intervals cover the entire circle.  Note that it is still possible for three intervals to cover all vertices.

Theorem \ref{thm:quasiline} tells us that every quasi-line graph satisfies $\chi(G) \leq \gamma_\ell(G)$.  For circular interval graphs, this bound follows easily from known results.  First, Niessen and Kind \cite{niessenk00} proved that circular interval graphs have the {\em round-up property}:

\begin{lemma}\label{lem:nk}
For any circular interval graph $G$, $\chi(G) = \lceil \chi_f(G)\rceil$.
\end{lemma}

A result of Shih and Hsu \cite{shihh89} tells us that we can optimally colour circular interval graphs efficiently:

\begin{lemma}\label{lem:shihhsu}
Given a circular interval graph $G$, we can find an optimal colouring of $G$ in $O(n^{3/2})$ time.
\end{lemma}

These results, along with Theorem \ref{thm:fractional}, immediately imply that Theorem \ref{thm:local} holds for circular interval graphs.

\begin{lemma}
If $G$ is a circular interval graph, we can find a $\gamma_\ell(G)$-colouring of $G$ in polynomial time.
\end{lemma}

%%%%%%%%%%%%%%%%%%%%%%%%%%%%%%%%%%%%%%%%%%%%%%%%%%%%%%%%%%%%%%%%%%%%%%%%%%%%%%%%
\subsection{Antihat thickenings}\label{sec:structureantihat}

We need to consider certain thickenings of graphs that are nearly antiprismatic.  Let $k\geq 2$.  We first define a graph $H$ with vertex set $A\cup B\cup C$ as follows.  Let $A = \{ a_0, a_1, \ldots, a_k\}$, $B = \{ b_0, b_1, \ldots, b_k\}$, and $C = \{c_1, \ldots, c_k\}$ be disjoint cliques.  Adjacency between the cliques is as follows:

\begin{itemize*}
\item $a_0$ has no neighbour outside $A\cup \{b_0\}$, and $b_0$ has no neighbour outside $B\cup\{a_0\}$.
\item For $1 \leq i,j \leq k$, $a_i$ and $b_j$ are nonadjacent if $i \neq j$ and adjacent if $i=j$.
\item For $1 \leq i,j \leq k$, $a_i$ and $b_i$ are adjacent to $c_j$ if $i \neq j$, and nonadjacent to $c_j$ if $i=0$ or if $i = j$.
\end{itemize*}
Let $X \subset A \cup B \cup C \setminus \{a_0,b_0\}$ such that $|C \setminus X| \geq 2$, and let $G=H-X$.  We say that $G$ is an {\em antihat graph}.  To define antihat thickenings, we first define a set $M\in V(G)^2$ as follows:

\begin{itemize*}
\item $M$ is a matching in $G\cup M$ containing no edge of $G[A]$, $G[B]$, or $G[C]$.
\item $a_0b_0$ is in $M$ if $a_0$ and $b_0$ are adjacent in $G$.
\item If $1\leq i,j$ and $a_ib_j\in M$ then $i=j$ and $c_i\in X$.
\item If $1\leq i,j$ and $b_ic_j\in M$ then $i=j$ and $a_i\in X$.
\item If $1\leq i,j$ and $a_ic_j\in M$ then $i=j$ and $b_i\in X$.
\end{itemize*}

In this case $G\cup M$ is claw-free and $M$ is a claw-neutral matching in $G\cup M$.  If $G'$ is a thickening of $G\cup M$ under $M$ then we say that it is an {\em antihat thickening}.  Observe that given an antihat graph $G$, adding an edge between $a_0$ and $b_0$ gives us an antiprismatic graph, as does deleting one or both of $a_0$ and $b_0$.

Having presented these graph classes, we can move on to the next step: describing and colouring three-cliqued claw-free graphs.

%%%%%%%%%%%%%%%%%%%%%%%%%%%%%%%%%%%%%%%%%%%%%%%%%%%%%%%%%%%%%%%%%%%%%%%%%%%%%%%%
%%%%%%%%%%%%%%%%%%%%%%%%%%%%%%%%%%%%%%%%%%%%%%%%%%%%%%%%%%%%%%%%%%%%%%%%%%%%%%%%
%%%%%%%%%%%%%%%%%%%%%%%%%%%%%%%%%%%%%%%%%%%%%%%%%%%%%%%%%%%%%%%%%%%%%%%%%%%%%%%%
%%%%%%%%%%%%%%%%%%%%%%%%%%%%%%%%%%%%%%%%%%%%%%%%%%%%%%%%%%%%%%%%%%%%%%%%%%%%%%%%
%%%%%%%%%%%%%%%%%%%%%%%%%%%%%%%%%%%%%%%%%%%%%%%%%%%%%%%%%%%%%%%%%%%%%%%%%%%%%%%%
%%%%%%%%%%%%%%%%%%%%%%%%%%%%%%%%%%%%%%%%%%%%%%%%%%%%%%%%%%%%%%%%%%%%%%%%%%%%%%%%
%%%%%%%%%%%%%%%%%%%%%%%%%%%%%%%%%%%%%%%%%%%%%%%%%%%%%%%%%%%%%%%%%%%%%%%%%%%%%%%%
%%%%%%%%%%%%%%%%%%%%%%%%%%%%%%%%%%%%%%%%%%%%%%%%%%%%%%%%%%%%%%%%%%%%%%%%%%%%%%%%
%%%%%%%%%%%%%%%%%%%%%%%%%%%%%%%%%%%%%%%%%%%%%%%%%%%%%%%%%%%%%%%%%%%%%%%%%%%%%%%%
%%%%%%%%%%%%%%%%%%%%%%%%%%%%%%%%%%%%%%%%%%%%%%%%%%%%%%%%%%%%%%%%%%%%%%%%%%%%%%%%
%%%%%%%%%%%%%%%%%%%%%%%%%%%%%%%%%%%%%%%%%%%%%%%%%%%%%%%%%%%%%%%%%%%%%%%%%%%%%%%%
%%%%%%%%%%%%%%%%%%%%%%%%%%%%%%%%%%%%%%%%%%%%%%%%%%%%%%%%%%%%%%%%%%%%%%%%%%%%%%%%
%%%%%%%%%%%%%%%%%%%%%%%%%%%%%%%%%%%%%%%%%%%%%%%%%%%%%%%%%%%%%%%%%%%%%%%%%%%%%%%%
%%%%%%%%%%%%%%%%%%%%%%%%%%%%%%%%%%%%%%%%%%%%%%%%%%%%%%%%%%%%%%%%%%%%%%%%%%%%%%%%
%%%%%%%%%%%%%%%%%%%%%%%%%%%%%%%%%%%%%%%%%%%%%%%%%%%%%%%%%%%%%%%%%%%%%%%%%%%%%%%%
\section{Three-cliqued claw-free graphs}\label{sec:3}

We now consider claw-free graphs with a three-colourable complement.  Given cliques $A$, $B$, and $C$ that partition the vertices of a claw-free graph $G$, we say that $(G,A,B,C)$ is a {\em three-cliqued claw-free graph}.  We also sometimes just call $G$ a three-cliqued claw-free graph without specifying a 3-colouring of $\gbar$.  As we will state formally in Theorem \ref{thm:3decomp}, any skeletal three-cliqued claw-free graph either admits a {\em hex-join}, which we describe shortly, or belongs to one of several base classes.

%%%%%%%%%%%%%%%%%%%%%%%%%%%%%%%%%%%%%%%%%%%%%%%%%%%%%%%%%%%%%%%%%%%%%%%%%%%%%%%%
%%%%%%%%%%%%%%%%%%%%%%%%%%%%%%%%%%%%%%%%%%%%%%%%%%%%%%%%%%%%%%%%%%%%%%%%%%%%%%%%
\subsection{Base classes of three-cliqued claw-free graphs}\label{sec:3cbase}

Since we restrict our attention to skeletal claw-free graphs, we can restrict the base classes of hex-joins that we describe.  However, it is possible to compose two nonskeletal three-cliqued claw-free graphs with a hex-join and reach a skeletal graph, so we cannot assume the base graphs are skeletal.  We therefore consider {\em weakly skeletal} base graphs, i.e.\ those in which every nonskeletal homogeneous pair of cliques has one clique intersecting at least two of $A$, $B$, and $C$:

\begin{definition}
Let $(X,Y)$ be a homogeneous pair of cliques in a three-cliqued graph $(G,A,B,C)$.  Then $(X,Y)$ is {\em straddling} if at least one of $X$ or $Y$ intersects more than one of $A$, $B$, and $C$.  We say that $(G,A,B,C)$ is {\em weakly skeletal} if every nonskeletal homogeneous pair of cliques is straddling.
\end{definition}

The first four classes we define contain weakly skeletal thickenings of members of the classes $\TC_1,\ldots,\TC_4$ as defined by Chudnovsky and Seymour \cite{clawfree5}.

\begin{itemize*}
\item {\bf A type of line graph}.  Let $H$ be a multigraph with pairwise nonadjacent vertices $a,b,c$ such that each of $a,b,c$ has at least three neighbours, and such that every edge of $H$ has an endpoint in $\{a,b,c\}$.  We further insist that for each $S \subset \{a,b,c\}$ there is at most one vertex $u$ outside $\{a,b,c\}$ whose neighbourhood is $S$.  Let $G = L(H)$, and let cliques $A$, $B$, and $C$ in $G$ correspond to the edges incident to $a$, $b$, and $c$ respectively in $H$.  Then $(G,A,B,C)$ is a three-cliqued claw-free graph.  Let $\TTC_1$ be the set of all such three-cliqued graphs such that every vertex is in a triad, with the added condition of being weakly skeletal.\footnote{To see that these graphs correspond to weakly skeletal thickenings of trigraphs in $\TC_1$ from \cite{clawfree5}, recall Proposition~\ref{prop:line} and its proof.}
\item {\bf Long circular interval graphs.}  Let $(G,A,B,C)$ be a three-cliqued long circular interval graph with a circular interval representation such that each of $A$, $B$, $C$ is a set of consecutive vertices in circular order.  Let $\TTC_2$ be the set of all such graphs that are weakly skeletal, such that every vertex is in a triad.
\item {\bf Antihat thickenings.}  Let $G$ be an antihat thickening, and let $A,B,C$, and $X$ be as they are in the definition of $G$.  Let $A'=A\setminus X$ and define $B'$ and $C'$ accordingly.  Then $(G-I(X),I(A'),I(B'),I(C'))$ is a three-cliqued claw-free graph.  Let $\TTC_3$ be the class of all such three-cliqued graphs with the added condition of being weakly skeletal.
\item {\bf Antiprismatic thickenings.}  Let $(G,A,B,C)$ be a three-cliqued antiprismatic graph, and let $(G',I(A),I(B),I(C))$ be a thickening of $G$ under a changeable matching $M$.  Let $\TTC_4$ be the class of all such graphs $(G',I(A),I(B),I(C))$ that are weakly skeletal.
\end{itemize*}
The final two exceptional cases correspond to thickenings of graphs in Chudnovsky and Seymour's class $\TC_5$ \cite{clawfree5}.
\begin{itemize*}
\item {\bf Exception I.}  Let $G$ be a graph on vertices $v_1,\ldots, v_8$ with adjacency as follows: $v_1$ is adjacent to $v_2,v_3,v_6,v_7$; $v_2$ is adjacent to $v_3,v_4$; $v_3$ is adjacent to $v_4,v_5$; $v_4$ is adjacent to $v_5,v_6$; $v_5$ is adjacent to $v_6$; $v_6$ and $v_8$ are adjacent to $v_7$; $v_2$ may or may not be adjacent to $v_5$.  There are no other edges.  Now let $M$ be a matching containing $v_1v_4$, $v_3v_6$, and also $v_2v_5$ if $v_2v_5\in E(G)$.  Let $X\subseteq \{v_3,v_4\}$.  Let $G'$ be a thickening of $(G\cup M)-X$ under $M$ (see Figure \ref{fig:exception}).  Then $(G', I(\{v_1,v_2,v_3\}),I(\{v_4,v_5,v_6\}),I(\{v_7,v_8\}))$ is a three-cliqued claw-free graph.  Let $\TTC_5$ be the set of all such graphs with the added condition of being weakly skeletal.

\begin{figure}
\begin{center}
\includegraphics[scale=.7]{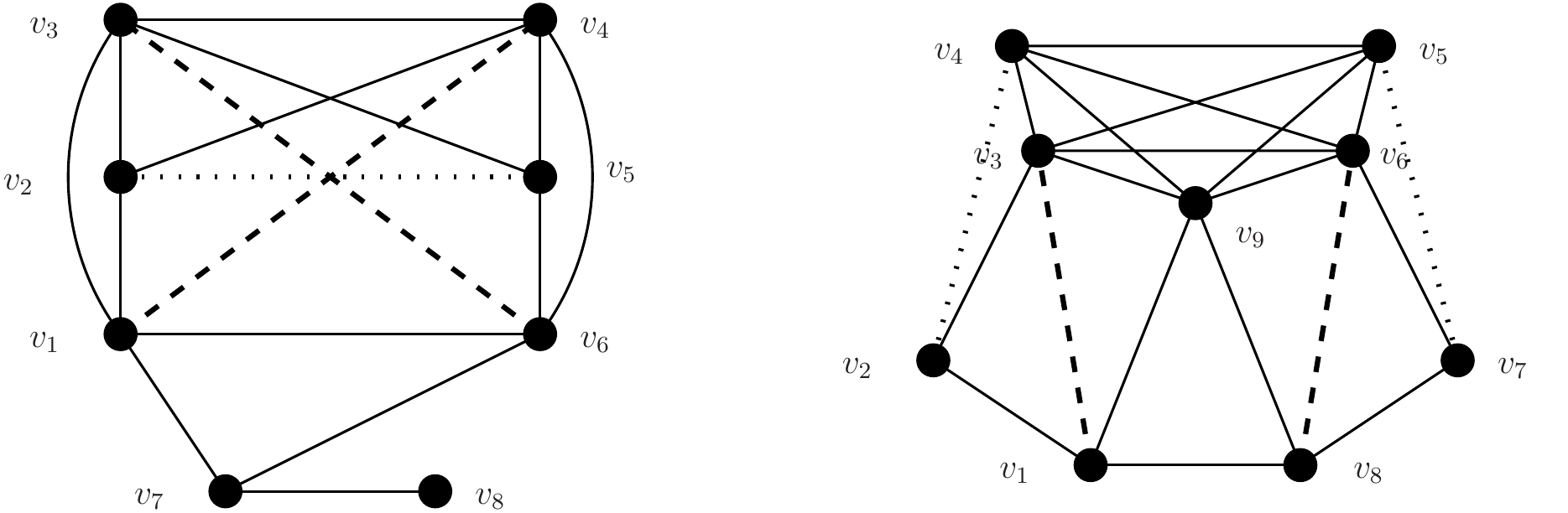}
\end{center}
\caption{\small{The graphs underlying exceptional thickenings in $\TC_5'$ (left) and $\TC_5''$ (right).  Solid, dashed, and dotted lines represent adjacent vertices, edges in $M$, and unspecified adjacency respectively.  All other pairs are nonadjacent.}}
\label{fig:exception}
\end{figure}

\item {\bf Exception II.}  Let $G$ be a graph on vertices $v_1,\ldots, v_9$ with the following structure.  Let $A= \{v_1,v_2\}$, $B=\{v_7,v_8\}$, and $C=\{v_3,v_4,v_5,v_6,v_9\}$ be cliques.  Let $v_1$ be adjacent to $v_3$, $v_8$, and $v_9$.  Let $v_8$ be adjacent to $v_6$ and $v_9$.  Let $v_2$ be adjacent to $v_3$ and possibly $v_4$.  Let $v_7$ be adjacent to $v_6$ and possibly $v_5$.  Now let $M$ be a matching in $G$ containing $v_1v_3$ and $v_6v_8$, as well as possibly $v_2v_4$ and $v_5v_7$ (see Figure \ref{fig:exception}).  Let $X$ be a subset of $\{v_3,v_4,v_5,v_6\}$ such that:
\begin{itemize*}
\item $v_2$ and $v_7$ each have a neighbour in $C\setminus X$.
\item If $X$ contains neither $v_4$ nor $v_5$ then $v_2$ is adjacent to $v_4$ and $v_7$ is adjacent to $v_5$.
\end{itemize*}
We insist that every vertex of $(G-M)-X$ is in a triad.  Let $G'$ be a thickening of $G-X$ under $M$.  Then $(G', I(A),I(B),I(C\setminus X))$ is a three-cliqued claw-free graph.  Let $\TTC_6$ be the set of all such graphs with the added condition of being weakly skeletal.
\end{itemize*}

We allow permutations of the sets $A,B,C$ for any of these classes, so if $(G,A,B,C)$ is in $\TTC_i$ for some $1\leq i\leq 6$ and $\{A',B',C'\} = \{A,B,C\}$, then $(G,A',B',C')$ is also in $\TTC_i$.  Having described the building blocks for three-cliqued claw-free graphs, we now move on to how they are combined (or from our perspective, decomposed).

%%%%%%%%%%%%%%%%%%%%%%%%%%%%%%%%%%%%%%%%%%%%%%%%%%%%%%%%%%%%%%%%%%%%%%%%%%%%%%%%
%%%%%%%%%%%%%%%%%%%%%%%%%%%%%%%%%%%%%%%%%%%%%%%%%%%%%%%%%%%%%%%%%%%%%%%%%%%%%%%%
%%%%%%%%%%%%%%%%%%%%%%%%%%%%%%%%%%%%%%%%%%%%%%%%%%%%%%%%%%%%%%%%%%%%%%%%%%%%%%%%
%%%%%%%%%%%%%%%%%%%%%%%%%%%%%%%%%%%%%%%%%%%%%%%%%%%%%%%%%%%%%%%%%%%%%%%%%%%%%%%%
\subsection{Decomposition: hex-joins}\label{sec:hexjoin}

We can decompose skeletal three-cliqued claw-free graphs into the base classes we just defined using a single decomposition operation: hex-joins.  Let $(G,A,B,C)$ be a three-cliqued graph, and suppose we partition $A$ into $A_1, A_2$, $B$ into $B_1, B_2$, $C$ into $C_1, C_2$.  Let $G_1 = G[A_1\cup B_1\cup C_1]$ and let $G_2=G[A_2\cup B_2\cup C_2]$.  Suppose we can construct $G$ from the disjoint union of $G_1$ and $G_2$ by adding every possible edge between $A_1$ and $A_2$, $A_2$ and $B_1$, $B_1$ and $B_2$, $B_2$ and $C_1$, $C_1$ and $C_2$, and $C_2$ and $A_1$.  Then we say that $(G,A,B,C)$ admits a {\em hex-join} into $(G_1, A_1,B_1,C_1)$ and $(G_2,A_2,B_2,C_2)$.

A simple observation explains our focus on weakly skeletal base graphs:
\begin{observation}
Let $(X,Y)$ be a nonskeletal, non-straddling homogeneous pair of cliques in a three-cliqued graph $(G_1,A_1,B_1,C_1)$.  If $(G,A,B,C)$ admits a hex-join into $(G_1,A_1,B_1,C_1)$ and any three-cliqued graph $(G_2,A_2,B_2,C_2)$, then $(X,Y)$ is a nonskeletal homogeneous pair of cliques in $(G,A,B,C)$.  In particular, $(G,A,B,C)$ is not skeletal.
\end{observation}

We use the following decomposition theorem for skeletal three-cliqued claw-free graphs.  It is a straightforward weakening of Chudnovsky and Seymour's structure theorem for three-cliqued claw-free trigraphs (4.1 in \cite{clawfree5}), as discussed in Chapter 9 of \cite{kingthesis}.

\begin{theorem}\label{thm:3decomp}
Any skeletal three-cliqued claw-free graph $(G,A,B,C)$ not in $\TTC_4$ admits a hex-join into terms $(G_1,A_1,B_1,C_1)$ and $(G_2,A_2,B_2,C_2)$, where $(G_1,A_1,B_1,C_1)$ is in one of $\TTC_1$, $\TTC_2$, $\TTC_3$, $\TTC_5$, or $\TTC_6$.
\end{theorem}

The omission of $\TTC_4$ from the list of possibilities comes from the easy fact that a graph admitting a hex-join into two terms, both of which are in $\TTC_4$, will itself be in $\TTC_4$.\\

\noindent{\bf Remark: }The reader familiar with the structure of {\em claw-free trigraphs} may object to our omission of {\em worn hex-joins}, described in \cite{clawfree5}.  This omission is possible because if $(G,A,B,C)$ admits a worn hex-join into $(G_1,A_1,B_1,C_1)$ and $(G_2,A_2,B_2,C_2)$, where $(G_1,A_1,B_1,C_1)$ is in one of $\TTC_1$, $\TTC_2$, $\TTC_3$, $\TTC_5$, or $\TTC_6$, then that worn hex-join is actually a hex-join, since every vertex in one of these classes arises as the image, in a thickening, of a vertex that is in a triad in the {\em trigraph} sense.

%%%%%%%%%%%%%%%%%%%%%%%%%%%%%%%%%%%%%%%%%%%%%%%%%%%%%%%%%%%%%%%%%%%%%%%%%%%%%%%%
%%%%%%%%%%%%%%%%%%%%%%%%%%%%%%%%%%%%%%%%%%%%%%%%%%%%%%%%%%%%%%%%%%%%%%%%%%%%%%%%
%%%%%%%%%%%%%%%%%%%%%%%%%%%%%%%%%%%%%%%%%%%%%%%%%%%%%%%%%%%%%%%%%%%%%%%%%%%%%%%%
%%%%%%%%%%%%%%%%%%%%%%%%%%%%%%%%%%%%%%%%%%%%%%%%%%%%%%%%%%%%%%%%%%%%%%%%%%%%%%%%
%%%%%%%%%%%%%%%%%%%%%%%%%%%%%%%%%%%%%%%%%%%%%%%%%%%%%%%%%%%%%%%%%%%%%%%%%%%%%%%%
%%%%%%%%%%%%%%%%%%%%%%%%%%%%%%%%%%%%%%%%%%%%%%%%%%%%%%%%%%%%%%%%%%%%%%%%%%%%%%%%
%%%%%%%%%%%%%%%%%%%%%%%%%%%%%%%%%%%%%%%%%%%%%%%%%%%%%%%%%%%%%%%%%%%%%%%%%%%%%%%%
%%%%%%%%%%%%%%%%%%%%%%%%%%%%%%%%%%%%%%%%%%%%%%%%%%%%%%%%%%%%%%%%%%%%%%%%%%%%%%%%
%%%%%%%%%%%%%%%%%%%%%%%%%%%%%%%%%%%%%%%%%%%%%%%%%%%%%%%%%%%%%%%%%%%%%%%%%%%%%%%%
%%%%%%%%%%%%%%%%%%%%%%%%%%%%%%%%%%%%%%%%%%%%%%%%%%%%%%%%%%%%%%%%%%%%%%%%%%%%%%%%
\subsection{Colouring three-cliqued claw-free graphs}

We now prove our first main result, Theorem \ref{thm:local}, which states that every three-cliqued claw-free graph $G$ satisfies $\chi(G)\leq \gamma_\ell(G)$.

To bound the chromatic number of antiprismatic thickenings, we removed a good triad whenever possible.  We will do the same for the remaining types of three-cliqued claw-free graphs.  A claw-free graph containing no triad is necessarily antiprismatic, but not all three-cliqued claw-free graphs containing a triad contain a good triad.  Observe that no minimum counterexample to Theorem \ref{thm:local} contains a good triad.  Furthermore, good triads behave nicely with respect to hex-joins:

\begin{observation}
Suppose that a three-cliqued claw-free graph $(G,A,B,C)$ admits a hex-join into $(G_1,A_1,B_1,C_1)$ and $(G_2,A_2,B_2,C_2)$.  If $T$ is a good triad in $G_1$, then it is also a good triad in $G$.
\end{observation}

Let $G$ be a minimum counterexample to Theorem \ref{thm:local}.  Then $G$ is skeletal and is not an antiprismatic thickening.  So Theorem \ref{thm:3decomp} implies that $G$ admits a hex-join into $(G_1,A_1,B_1,C_1)$ and $(G_2,A_2,B_2,C_2)$, where $(G_1,A_1,B_1,C_1)$ is in $\TTC_1$, $\TTC_2$, $\TTC_3$, $\TTC_5$, or $\TTC_6$.  We deal with these five possibilities individually.  Note that $G_2$ may be empty, but this does not affect our approach.

%%%%%%%%%%%%%%%%%%%%%%%%%%%%%%%%%%%%%%%%%%%%%%%%%%%%%%%%%%%%%%%%%%%%%%%%%%%%%%%%
%%%%%%%%%%%%%%%%%%%%%%%%%%%%%%%%%%%%%%%%%%%%%%%%%%%%%%%%%%%%%%%%%%%%%%%%%%%%%%%%
%%%%%%%%%%%%%%%%%%%%%%%%%%%%%%%%%%%%%%%%%%%%%%%%%%%%%%%%%%%%%%%%%%%%%%%%%%%%%%%%
%%%%%%%%%%%%%%%%%%%%%%%%%%%%%%%%%%%%%%%%%%%%%%%%%%%%%%%%%%%%%%%%%%%%%%%%%%%%%%%%
%%%%%%%%%%%%%%%%%%%%%%%%%%%%%%%%%%%%%%%%%%%%%%%%%%%%%%%%%%%%%%%%%%%%%%%%%%%%%%%%
\subsubsection{Five classes to consider}

We now prove a set of lemmas that together imply Theorem \ref{thm:local}, dealing with the easier cases first.

%%%%%%%%%%%%%%%%%%%%%%%%%%%%%%%%%%%%%%%%%%%%%%%%%%%%%%%%%%%%%%%%%%%%%%%%%%%%%%%%
%%%%%%%%%%%%%%%%%%%%%%%%%%%%%%%%%%%%%%%%%%%%%%%%%%%%%%%%%%%%%%%%%%%%%%%%%%%%%%%%
%%%%%%%%%%%%%%%%%%%%%%%%%%%%%%%%%%%%%%%%%%%%%%%%%%%%%%%%%%%%%%%%%%%%%%%%%%%%%%%%
\subsubsection*{Long circular interval graphs ($\TTC_2$)}

\begin{lemma}\label{lem:min32}
Any three-cliqued graph $(G_1,A_1,B_1,C_1)$ in $\TTC_2$ contains a good triad.
\end{lemma}

\begin{proof}
Suppose that $(G_1,A_1,B_1,C_1)$ is in $\TTC_2$, and call the vertices of $G$ $\{a_1,\ldots, a_i,\allowbreak b_1,\ldots, b_j,\allowbreak c_1, \ldots, c_k \}$ in circular order.

We can find a triad $T$ containing $a_1$ by adding $b_p$ for the minimum $p$ such that $b_p$ does not see $a_1$, then adding $c_q$ for the minimum $q$ such that $b_p$ does not see $c_q$.  The triad $T$ exists since $a_1$ is in a triad, and it follows from the structure of circular interval graphs that $a_1$ and $b_p$ are in a triad together.  If some vertex in $(A \setminus \{a_1\}) \cup\{b_x \mid x < p\}$ does not see both $a_1$ and $b_p$, then we are in a degenerate case where $G_1$ is a linear interval graph, and the vertex in question is a twin of $a_1$ or $b_p$, or it is trumped by $a_1$ or $b_p$.  The same applies to every vertex in $\{b_x \mid x > p\} \cup \{c_y \mid y < q\}$:  each vertex has two neighbours in $T$ or a twin in $T$ or is trumped by a vertex in $T$.  Similarly, if some vertex $v$ in $\{c_l \mid l > q\}$ has only one neighbour in $T$ then it has no neighbours in $A$, hence it is trumped by or is a twin of $c_q$.  Thus $T$ is a good triad.
\end{proof}

%%%%%%%%%%%%%%%%%%%%%%%%%%%%%%%%%%%%%%%%%%%%%%%%%%%%%%%%%%%%%%%%%%%%%%%%%%%%%%%%
%%%%%%%%%%%%%%%%%%%%%%%%%%%%%%%%%%%%%%%%%%%%%%%%%%%%%%%%%%%%%%%%%%%%%%%%%%%%%%%%
%%%%%%%%%%%%%%%%%%%%%%%%%%%%%%%%%%%%%%%%%%%%%%%%%%%%%%%%%%%%%%%%%%%%%%%%%%%%%%%%
\subsubsection*{Antihat thickenings ($\TTC_3$)}

\begin{lemma}\label{lem:min33}
Any three-cliqued graph $(G_1,A_1,B_1,C_1)$ in $\TTC_3$ contains a good triad.
\end{lemma}

\begin{proof}
Let $T$ be a triad consisting of a vertex $a$ of $I(a_0)$ and vertices $b$ in $I(B\setminus \{b_0\})$ and $c$ in $I(C)$ respectively, following the definition of an antihat thickening.  If $b$ and $c$ are in $I(b_i)$ and $I(c_i)$ respectively, we insist that $T$ intersects $\Omega(b_i,c_i)$ if it is not empty.  We also insist that if $I(a_0)\cap \Omega(a_0b_0)$ is nonempty, then $T$ intersects it.  It is easy to confirm from the structure of an antihat thickening that $T$ exists and is a good triad.
\end{proof}

%%%%%%%%%%%%%%%%%%%%%%%%%%%%%%%%%%%%%%%%%%%%%%%%%%%%%%%%%%%%%%%%%%%%%%%%%%%%%%%%
%%%%%%%%%%%%%%%%%%%%%%%%%%%%%%%%%%%%%%%%%%%%%%%%%%%%%%%%%%%%%%%%%%%%%%%%%%%%%%%%
%%%%%%%%%%%%%%%%%%%%%%%%%%%%%%%%%%%%%%%%%%%%%%%%%%%%%%%%%%%%%%%%%%%%%%%%%%%%%%%%
\subsubsection*{Exception I ($\TTC_5$)}

\begin{lemma}\label{lem:min35}
Any three-cliqued graph $(G_1,A_1,B_1,C_1)$ in $\TTC_5$ contains a good triad.
\end{lemma}

\begin{proof}
Let $T$ be a triad including one vertex in each of $I(v_7)$, $I(v_2)$, and $I(v_5)$, such that $T$ intersects $\Omega(v_2v_5)$ if it is not empty.  It is easy to confirm that $T$ is a good triad:  Vertices in $I(\{v_1,v_3,v_4,v_6\})$ have two neighbours in $T$, and vertices in $I(\{v_7,v_8\})$ have a neighbour or a twin in $T$.  If $\Omega(v_2v_5)$ is empty then vertices in $I(\{v_2,v_5\})$ have a twin in $T$.  If not, then assume without loss of generality that $T$ intersects $I(v_2)\cap \Omega(v_2v_5)$ and $I(v_5)\setminus \Omega(v_2v_5)$.  Then vertices in $I(v_2)\cap \Omega(v_2v_5)$ have a twin in $T$, vertices in $I(v_2)\setminus \Omega(v_2v_5)$ are trumped by a vertex in $T$, vertices in $I(v_5)\cap \Omega(v_2v_5)$ have two neighbours in $T$, and vertices in $I(v_5)\setminus \Omega(v_2v_5)$ have a twin in $T$.  Therefore $T$ is a good triad.
\end{proof}

%%%%%%%%%%%%%%%%%%%%%%%%%%%%%%%%%%%%%%%%%%%%%%%%%%%%%%%%%%%%%%%%%%%%%%%%%%%%%%%%
%%%%%%%%%%%%%%%%%%%%%%%%%%%%%%%%%%%%%%%%%%%%%%%%%%%%%%%%%%%%%%%%%%%%%%%%%%%%%%%%
%%%%%%%%%%%%%%%%%%%%%%%%%%%%%%%%%%%%%%%%%%%%%%%%%%%%%%%%%%%%%%%%%%%%%%%%%%%%%%%%
\subsubsection*{Exception II ($\TTC_6$)}

\begin{lemma}\label{lem:min36}
Any three-cliqued graph $(G_1,A_1,B_1,C_1)$ in $\TTC_6$ contains a good triad.
\end{lemma}

\begin{proof}
Let $T$ be a triad including one vertex in each of $I(v_2)$, $I(v_7)$, and $I(v_9)$, such that $T$ intersects $\Omega(v_2v_4)$ if it is not empty, and intersects $\Omega(v_5v_7)$ if it is not empty.  It is easy to confirm that $T$ is a good triad (see Figure \ref{fig:exception}).
\end{proof}

%%%%%%%%%%%%%%%%%%%%%%%%%%%%%%%%%%%%%%%%%%%%%%%%%%%%%%%%%%%%%%%%%%%%%%%%%%%%%%%%
%%%%%%%%%%%%%%%%%%%%%%%%%%%%%%%%%%%%%%%%%%%%%%%%%%%%%%%%%%%%%%%%%%%%%%%%%%%%%%%%
%%%%%%%%%%%%%%%%%%%%%%%%%%%%%%%%%%%%%%%%%%%%%%%%%%%%%%%%%%%%%%%%%%%%%%%%%%%%%%%%
\subsubsection*{A type of line graph ($\TTC_1$)}

We now prove the necessary lemma for $G_1$ in $\TTC_1$.  This is by far the most difficult case.  We make extensive use of the fact that line graphs of bipartite multigraphs are perfect.

\begin{lemma}\label{lem:min31}
Let $(G,A,B,C)$ be a minimum counterexample to Theorem \ref{thm:local} and suppose it admits a hex-join into $(G_1,A_1,B_1,C_1)$ and $(G_2,A_2,B_2,C_2)$.  Then $(G_1,A_1,B_1,C_1)$ is not in $\TTC_1$.
\end{lemma}

\begin{proof}
Suppose $(G_1,A_1,B_1,C_1)$ is in $\TTC_1$.  Then $G_1$ is the line graph of some bipartite multigraph $H$ which has a stable set $\{a,b,c\}$ corresponding to $A_1$, $B_1$, and $C_1$.  Assume without loss of generality that $|C_1|\leq |B_1|\leq |A_1|$.  We call the other vertices of $H$ {\em centres}.  Depending on the structure of $G_1$ we will take one of two actions:

\begin{enumerate*}
\item Remove a triad from $G_1$, lowering $\gamma_\ell(G)$.
\item Remove edges from $G_1$ without changing $\chi(G)$ or changing the fact that $(G_1,A_1,B_1,C_1)$ is three-cliqued or claw-free.
\end{enumerate*}

Every vertex of $G_1$ is in a triad.  If there are only three centres then removing any triad $T$ will lower $\gamma_\ell(G)$ since every vertex in $G_1-T$ will have two neighbours or a twin in $T$ -- this can be confirmed easily since the graph underlying $H$ will be a subgraph of $K_{3,3}$.  So there are at least four centres.  Call the four centres of highest degree $w$, $x$, $y$, and $z$ such that $d(w)\geq d(x)\geq d(y)\geq d(z)$.

For any centre $s$, denote by $A_s$ the clique corresponding to the edges of $H$ between $a$ and $s$.  Define $B_s$ and $C_s$ accordingly.  Denote $A_s\cup B_s\cup C_s$ by $X_s$.  We now consider, for some vertex $v\in A_s$, what cliques of size $\omega(v)$ can contain $v$.  By the structure of a hex-join, observe that such a clique must be one of:
\begin{itemize*}
\item A clique in $G_1$ intersecting all of $A_1,B_1,C_1$.  Specifically, $A_s \cup B_s \cup C_s = X_s$.
\item A clique in $A_1\cup B_1 \cup A_2$ containing all of $A_2$.  Specifically, $A_2 \cup A_s \cup B_s$.
\item A clique in $A_1 \cup C_1 \cup C_2$ containing all of $C_2$.  Specifically, $C_2 \cup A_s \cup C_s$.
\item A clique in $A_1 \cup A_2 \cup C_2$ containing all of $A_1$.  Such a clique has size at least $|A_1|+\max\{|A_2|,|C_2|\}$.
\end{itemize*}
Note also that the closed neighbourhood of $v$ is $A_1\cup X_s\cup A_2 \cup C_2$.  We can make similar observations about the cliques of size $\omega(v)$ when $v$ is in $A_1 \setminus A_s$ or $B_1$ or $C_1$.  These observations help us characterize the situations in which removing a triad $T$ lowers $\omega(v)$ and therefore $\gamma_\ell(v)$.

Note that at most two centres have degree $\geq d(a)$, since there are at least four centres and the sum of their degrees is $d(a)+d(b)+d(c)$.  Suppose there are at most three centres with degree $\geq d(c)$.  Then the structure of $\TTC_1$ tells us that we can find a matching of size 3 in $H$ hitting each of these centres having degree greater than 1.  We will now show that removing the corresponding triad $T$ from $G$ will lower $\gamma_\ell(v)$ for all $v\in G_1$.  This triad $T$ will hit $A_1$, $B_1$, and $C_1$.  Any vertex $v$ in $G_1$ without two neighbours or a twin in $T$, that is not trumped by a vertex in $T$, will correspond to an edge in $H$ incident to some centre $s$, where $d(s)<d(c)$.  By our above observations about cliques of size $\omega(v)$, we can see that since $|X_s|<|C_1|\leq |B_1|\leq |A_1|$, any clique of size $\omega(v)$ containing $v$ must contain one of $C_1$, $B_1$, or $A_1$.  Therefore such a clique intersects $T$, so removing $T$ lowers $\omega(v)$ and also $\gamma_\ell(v)$.  This contradicts the minimality of $G$, so we can assume that there are at least four centres of degree $\geq |C|$, i.e.\ $d(z)\geq d(c)$.  Given this restriction we now consider several cases.
\\

\noindent{\bf Case 1:} $d(w) \geq d(a)$ and $c$ sees $w$.

Since $d(w)\geq d(a)$ it follows that $d(x)+d(y)+d(z) \leq |B_1|+|C_1|$, and so $d(x)+d(y) \leq |B_1|$ and $|C_1|\leq d(z)\leq \frac 13(|B_1|+|C_1|)$.  Therefore $2|C_1| \leq 2d(z)\leq d(x)+d(y) \leq |B_1|$.  Take a triad $T$ that hits $X_w$, $X_x$, and $X_y$, and consider a vertex $v$ for which $\omega(v)$ does not drop when $T$ is removed.  Clearly $v$ is not in $X_w\cup X_x\cup X_y$, so it is in $X_s$ for some centre $s$ with $d(s)\leq d(z)\leq \frac 12|B_1|$.  Since $|X_s|<|B_1|$ and $\omega(v)$ does not drop, $v$ must be in $C_s$.  Take some $u\in C_w$.  We will show that $d(u)+\omega(u)> d(v)+\omega(v)$, which implies that $\gamma_\ell(G-T)<\gamma_\ell(T)$.

Clearly $u$ has at least $|A_1|-|C_1|$ neighbours in $G_1-C_1$.  But $v$ has at most $\frac 12|B_1|-1$ neighbours in $G_1-C_1$.  Therefore $d(u)> d(v)+\frac 12|A_1|-|C_1|$.  Recall the structure of maximal cliques containing $u$ and $v$.  If $\omega(v)> \omega(u)$ then either $|C_s|+|A_s|>\max\{|C_w|+|A_w|,|C_1|\}$ or $|C_s|+|B_s|>\max\{|C_w|+|B_w|,|C_1|\}$.  But in this case $\omega(v) \leq \omega(u)+d(s)-|C_1|\leq \omega(u)+\frac 12|A_1|-|C_1|$.  It follows that $d(v)+\omega(v)<d(u)+\omega(u)$, completing the case.
\\

\noindent{\bf Case 2:} $d(w) \geq d(a)$ and $c$ does not see $w$.

Make the subgraph $G'$ of $G$ by removing all edges between $C_1$ and $G_1\setminus C_1$ -- observe that $(G',A,B,C)$ is claw-free and three-cliqued.  Further observe that because $H$ has at least four centres, if $G'=G$ then $(A_1,B_1)$ is a nonskeletal homogeneous pair of cliques in both $G_1$ and $G$, a contradiction.  Thus $G'$ is a proper subgraph of $G$.  We claim that $\chi(G')=\chi(G)$, contradicting the minimality of $G$.  Denote by $G_1'$ the subgraph of $G'$ induced on the vertices of $G_1$.

Take a $\chi(G')$-colouring $\cC'$ of $G'$.  We will rearrange the colour classes of $\cC'$ on $G_1'$ to reach a proper colouring of $G_1$.  Denote by $t$ the number of triad colour classes in $\cC'$ restricted to $G_1'$.  Denote by $d_{AB}$, $d_{AC}$, and $d_{BC}$ the number of diads (i.e.\ colour classes of size two) in $\cC'$ restricted to $G_1'$ intersecting $A_1$ and $B_1$, $A_1$ and $C_1$, and $B_1$ and $C_1$ respectively.  It suffices to show that we can pack the appropriate disjoint stable sets into $G_1$.  That is, we want to find $t$ triads in $G_1$, $d_{AB}$ diads intersecting $A_1$ and $B_1$, $d_{AC}$ diads intersecting $A_1$ and $C_1$, and $d_{BC}$ diads intersecting $B_1$ and $C_1$, such that all of these stable sets are disjoint.

We begin with $|A_1|+|B_1|-d(w) = |A_1|+|B_1|-\omega(G[A_1\cup B_1])$ diads intersecting $A_1$ and $B_1$.  Since $G[A_1\cup B_1]$ is cobipartite, these diads hit every vertex of $(A_1\cup B_1)\setminus X_w$.  Observe that $|A_1|+|B_1|-d(w) \geq t+d_{AB}$.  So we want to extend some of the diads to triads.  We can actually extend $|C_1|$ of them.  To see this, note that there are at least three centres of degree $\geq |C_1|$ other than $w$, so every vertex in $C_1$ has at least $C_1$ non-neighbours in $(A_1\cup B_1)\setminus X_w$.  So we have $|A_1|+|B_1|-d(w)-|C_1|$ disjoint diads intersecting $A$ and $B$ and a further $|C_1|$ disjoint triads.

Thus it is clear that we can find the desired disjoint stable sets, beginning with the diads intersecting $A$ and $B$.  When picking our $d_{AC}+d_{BC}$ remaining diads we take a vertex in $X_w$ not intersecting an $AB$ diad whenever possible.  Once we have found the necessary diads, we have enough $AB$ diads remaining so that we can extend them to triads.  These stable sets give us a $\chi(G')$-colouring of $G$, contradicting the minimality of $G$.
\\

\noindent{\bf Case 3:} $d(w) < d(a)$.

As in the previous case, we remove edges from $G_1$ without introducing a claw or changing the chromatic number of $G$.  There is at most one clique $X$ in $G[B_1\cup C_1]$ of size greater than $|B_1|$, by the structure of $G_1$.  If $X$ exists, construct $G'$ from $G$ by removing all edges from $G_1$ except those within $A_1$, $B_1$, $C_1$, and $X$.  If such an $X$ does not exist, set $X$ as $B_1$ and construct $G'$ from $G$ by removing all edges from $G_1$ except those within $A_1$, $B_1$, and $C_1$.  It is easy to confirm that $G'$ is claw-free and a proper subgraph of $G$.  We will show that $\chi(G')=\chi(G)$, contradicting the minimality of $G$.

We claim that there is an $\omega(G_1)$-colouring of $G_1$ using $|B_1|+|C_1|-|X|$ triads.  To see this, we remove $|A_1|-|X|$ vertices from $A_1$ one at a time, always taking one from the largest clique $X_s$ that still has a vertex in $A_1$.  If after removing $k$ vertices we have disjoint $X_s$ and $X_{s'}$ of size $|A_1|-k$, then we have $|A_1|+|B_1|+|C_1| \geq k+ 2(|A_1|-k)+2|C_1|$, contradicting the facts that there are at least four centres of degree $\geq |C_1|$ in $H$ and that $|B_1|\leq |A_1|-k$.  Thus we can see that we reach a perfect graph on $|X|+|B_1|+|C_1|$ vertices with clique number $|X|$.  In an $|X|$-colouring of this graph every colour class intersects both $A_1$ and $X$, thus the colouring uses exactly $|B_1|+|C_1|-|X|$ triads.  The other colour classes are diads intersecting $X$ and $A_1$.  Thus as in the previous case, we can rearrange the colour classes of a $\chi(G')$-colouring $\cC'$ of $G'$ to construct a $\chi(G')$-colouring of $G$.
\end{proof}

%%%%%%%%%%%%%%%%%%%%%%%%%%%%%%%%%%%%%%%%%%%%%%%%%%%%%%%%%%%%%%%%%%%%%%%%%%%%%%%%
%%%%%%%%%%%%%%%%%%%%%%%%%%%%%%%%%%%%%%%%%%%%%%%%%%%%%%%%%%%%%%%%%%%%%%%%%%%%%%%%
%%%%%%%%%%%%%%%%%%%%%%%%%%%%%%%%%%%%%%%%%%%%%%%%%%%%%%%%%%%%%%%%%%%%%%%%%%%%%%%%
\subsubsection{Completing the proof}

We now combine our lemmas to prove Theorem \ref{thm:local}.

\begin{proof}[Proof of Theorem \ref{thm:local}]
Let $(G,A,B,C)$ be a minimum counterexample to the theorem.  Then $G$ is skeletal and is not an antiprismatic thickening.  Theorem \ref{thm:3decomp} tells us that $(G,A,B,C)$ admits a hex-join into $(G_1,A_1,B_1,C_1)$ and $(G_2,A_2,B_2,C_2)$ such that $(G_1,A_1,B_1,C_1)$ is in one of $\TTC_1$, $\TTC_2$, $\TTC_3$, $\TTC_5$ or $\TTC_6$.  Lemmas \ref{lem:min31}, \ref{lem:min32}, \ref{lem:min33}, \ref{lem:min35}, and \ref{lem:min36} tell us that $(G_1,A_1,B_1,C_1)$ cannot be in $\TTC_1$, $\TTC_2$, $\TTC_3$, $\TTC_5$ or $\TTC_6$ respectively.  Thus $G$ cannot exist, proving the theorem.
\end{proof}

%%%%%%%%%%%%%%%%%%%%%%%%%%%%%%%%%%%%%%%%%%%%%%%%%%%%%%%%%%%%%%%%%%%%%%%%%%%%%%%%
%%%%%%%%%%%%%%%%%%%%%%%%%%%%%%%%%%%%%%%%%%%%%%%%%%%%%%%%%%%%%%%%%%%%%%%%%%%%%%%%
%%%%%%%%%%%%%%%%%%%%%%%%%%%%%%%%%%%%%%%%%%%%%%%%%%%%%%%%%%%%%%%%%%%%%%%%%%%%%%%%
%%%%%%%%%%%%%%%%%%%%%%%%%%%%%%%%%%%%%%%%%%%%%%%%%%%%%%%%%%%%%%%%%%%%%%%%%%%%%%%%
%%%%%%%%%%%%%%%%%%%%%%%%%%%%%%%%%%%%%%%%%%%%%%%%%%%%%%%%%%%%%%%%%%%%%%%%%%%%%%%%
%%%%%%%%%%%%%%%%%%%%%%%%%%%%%%%%%%%%%%%%%%%%%%%%%%%%%%%%%%%%%%%%%%%%%%%%%%%%%%%%
%%%%%%%%%%%%%%%%%%%%%%%%%%%%%%%%%%%%%%%%%%%%%%%%%%%%%%%%%%%%%%%%%%%%%%%%%%%%%%%%
%%%%%%%%%%%%%%%%%%%%%%%%%%%%%%%%%%%%%%%%%%%%%%%%%%%%%%%%%%%%%%%%%%%%%%%%%%%%%%%%
%%%%%%%%%%%%%%%%%%%%%%%%%%%%%%%%%%%%%%%%%%%%%%%%%%%%%%%%%%%%%%%%%%%%%%%%%%%%%%%%
%%%%%%%%%%%%%%%%%%%%%%%%%%%%%%%%%%%%%%%%%%%%%%%%%%%%%%%%%%%%%%%%%%%%%%%%%%%%%%%%
%%%%%%%%%%%%%%%%%%%%%%%%%%%%%%%%%%%%%%%%%%%%%%%%%%%%%%%%%%%%%%%%%%%%%%%%%%%%%%%%
%%%%%%%%%%%%%%%%%%%%%%%%%%%%%%%%%%%%%%%%%%%%%%%%%%%%%%%%%%%%%%%%%%%%%%%%%%%%%%%%
%%%%%%%%%%%%%%%%%%%%%%%%%%%%%%%%%%%%%%%%%%%%%%%%%%%%%%%%%%%%%%%%%%%%%%%%%%%%%%%%
%%%%%%%%%%%%%%%%%%%%%%%%%%%%%%%%%%%%%%%%%%%%%%%%%%%%%%%%%%%%%%%%%%%%%%%%%%%%%%%%
%%%%%%%%%%%%%%%%%%%%%%%%%%%%%%%%%%%%%%%%%%%%%%%%%%%%%%%%%%%%%%%%%%%%%%%%%%%%%%%%
\section{Compositions of strips and generalized 2-joins}\label{sec:compositions}

In this section we describe how to colour claw-free {\em compositions of strips}, an important class of graphs built as a generalization of line graphs.  For a discussion of this composition operation we refer the reader to \cite{kingthesis} \S 5.2 or \cite{cssurvey}.  Rather than concerning ourselves with the global structure of these graphs, we instead focus on decompositions that arise in these graphs, and how we might exploit these decompositions in order to extend partial colourings.  These decompositions are {\em generalized 2-joins}:

\begin{definition}
Suppose vertex sets $V_1$ and $V_2$ partition $V(G)$ and there are cliques $X_i$ and $Y_i$ in $V_i$ such that $X_1\cup X_2$ and $Y_1\cup Y_2$ are cliques, and there are no other edges between $V_1$ and $V_2$.  Then we say that $((X_1,Y_1),(X_2,Y_2))$ is a {\em generalized 2-join}.
\end{definition}

Let $G_1$ and $G_2$ denote $G[V_1]$ and $G[V_2]$, respectively.  In order to extend a $\gamma(G)$-colouring of $G_1$ to a $\gamma(G)$-colouring of $G$, merely having our generalized 2-join is not enough.  Rather, we need to know the structure of $G_2$ and exploit properties of restricted colourings based on that structure.  The structure of $G_2$ can be described in terms of {\em strips}.

\begin{definition}A {\em strip} $(G,A,B)$ is a claw-free graph $G$ with two cliques $A$ and $B$ such that for any vertex $v \in A$ (resp.\ $B$), the neighbourhood of $v$ outside $A$ (resp.\ $B$) is a clique.
\end{definition}

The strip $(G,A,B)$ will actually be $(G_2,X_2,Y_2)$.  We now examine five types of strips that will give us the five types of generalized 2-join that we need to deal with.

%%%%%%%%%%%%%%%%%%%%%%%%%%%%%%%%%%%%%%%%%%%%%%%%%%%%%%%%%%%%%%%%%%%%%%%%%%%%%%%%
%%%%%%%%%%%%%%%%%%%%%%%%%%%%%%%%%%%%%%%%%%%%%%%%%%%%%%%%%%%%%%%%%%%%%%%%%%%%%%%%
%%%%%%%%%%%%%%%%%%%%%%%%%%%%%%%%%%%%%%%%%%%%%%%%%%%%%%%%%%%%%%%%%%%%%%%%%%%%%%%%
%%%%%%%%%%%%%%%%%%%%%%%%%%%%%%%%%%%%%%%%%%%%%%%%%%%%%%%%%%%%%%%%%%%%%%%%%%%%%%%%
%%%%%%%%%%%%%%%%%%%%%%%%%%%%%%%%%%%%%%%%%%%%%%%%%%%%%%%%%%%%%%%%%%%%%%%%%%%%%%%%
%%%%%%%%%%%%%%%%%%%%%%%%%%%%%%%%%%%%%%%%%%%%%%%%%%%%%%%%%%%%%%%%%%%%%%%%%%%%%%%%
%%%%%%%%%%%%%%%%%%%%%%%%%%%%%%%%%%%%%%%%%%%%%%%%%%%%%%%%%%%%%%%%%%%%%%%%%%%%%%%%
%%%%%%%%%%%%%%%%%%%%%%%%%%%%%%%%%%%%%%%%%%%%%%%%%%%%%%%%%%%%%%%%%%%%%%%%%%%%%%%%
\subsection{Five types of strips}

The first strips we consider are linear interval strips, which are essential to the structure of quasi-line graphs.  The other four types contain a $W_5$, i.e.\ an induced $C_5$ with a universal vertex, and therefore can only appear in graphs that are not quasi-line.

%%%%%%%%%%%%%%%%%%%%%%%%%%%%%%%%%%%%%%%%%%%%%%%%%%%%%%%%%%%%%%%%%%%%%%%%%%%%%%%%
%%%%%%%%%%%%%%%%%%%%%%%%%%%%%%%%%%%%%%%%%%%%%%%%%%%%%%%%%%%%%%%%%%%%%%%%%%%%%%%%
%%%%%%%%%%%%%%%%%%%%%%%%%%%%%%%%%%%%%%%%%%%%%%%%%%%%%%%%%%%%%%%%%%%%%%%%%%%%%%%%
\subsubsection{Linear interval strips}

Let $G$ be a linear interval graph, and let cliques $A$ and $B$ be the $|A|$ leftmost and $|B|$ rightmost vertices of $G$ in some linear interval representation of $G$.  Then $(G,A,B)$ is a {\em linear interval strip}.

%%%%%%%%%%%%%%%%%%%%%%%%%%%%%%%%%%%%%%%%%%%%%%%%%%%%%%%%%%%%%%%%%%%%%%%%%%%%%%%%
%%%%%%%%%%%%%%%%%%%%%%%%%%%%%%%%%%%%%%%%%%%%%%%%%%%%%%%%%%%%%%%%%%%%%%%%%%%%%%%%
%%%%%%%%%%%%%%%%%%%%%%%%%%%%%%%%%%%%%%%%%%%%%%%%%%%%%%%%%%%%%%%%%%%%%%%%%%%%%%%%
\subsubsection{Antihat strips}

Let $G$ be an antihat graph, and let $G'$ be an antihat thickening, i.e.\ a thickening of $G\cup M$ under $M$ as defined in Section \ref{sec:structureantihat}.  We specify two cliques of $G'$:  $A' = I(A\setminus (X\cup \{a_0\}))$ and $B' = I(B\setminus (X\cup  \{b_0\}))$.  Then $(G'-I(a_0)-I(b_0), A', B')$ is a strip and if it contains a $W_5$ we say that it is an {\em antihat strip}.  These antihat strips are a slight generalization of the antihat strips used in Chudnovsky and Seymour's survey \cite{cssurvey}.

%%%%%%%%%%%%%%%%%%%%%%%%%%%%%%%%%%%%%%%%%%%%%%%%%%%%%%%%%%%%%%%%%%%%%%%%%%%%%%%%
%%%%%%%%%%%%%%%%%%%%%%%%%%%%%%%%%%%%%%%%%%%%%%%%%%%%%%%%%%%%%%%%%%%%%%%%%%%%%%%%
%%%%%%%%%%%%%%%%%%%%%%%%%%%%%%%%%%%%%%%%%%%%%%%%%%%%%%%%%%%%%%%%%%%%%%%%%%%%%%%%
\subsubsection{Strange strips}

Let $H$ be a claw-free graph on cliques $A=\{a_1, a_2 \}$, $B=\{ b_1, b_2, b_3 \}$, and $C=\{c_1, c_2 \}$ with adjacency as follows: $a_1, b_1$ are adjacent; $c_1$ is adjacent to $a_2$ and $b_2$ and $b_3$; $c_2$ is adjacent to $a_1$, $a_2$, $b_1$, and $b_2$.  All other pairs are nonadjacent.  Let $G$ be a thickening of $H$ under $M=\{b_3c_1,b_2c_2 \}$ (see Figure \ref{fig:strange}).  Then $(G,I(A),I(B))$ is a strip and we say that it is a {\em strange strip}.

\begin{figure}
\begin{center}
\includegraphics[scale=.6]{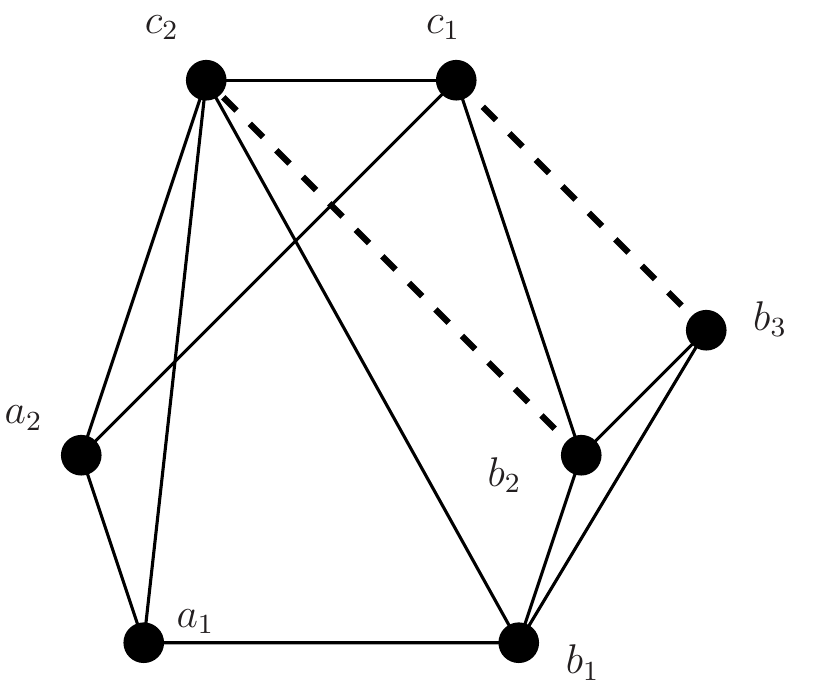}
\end{center}
\caption{\small{The graph underlying strange strips.  Dashed lines represent edges in $M$.}}
\label{fig:strange}
\end{figure}

%%%%%%%%%%%%%%%%%%%%%%%%%%%%%%%%%%%%%%%%%%%%%%%%%%%%%%%%%%%%%%%%%%%%%%%%%%%%%%%%
%%%%%%%%%%%%%%%%%%%%%%%%%%%%%%%%%%%%%%%%%%%%%%%%%%%%%%%%%%%%%%%%%%%%%%%%%%%%%%%%
%%%%%%%%%%%%%%%%%%%%%%%%%%%%%%%%%%%%%%%%%%%%%%%%%%%%%%%%%%%%%%%%%%%%%%%%%%%%%%%%
\subsubsection{Pseudo-line strips}

We will define a type of line graph and modify it slightly.  Let $J$ be a graph containing a path on vertices $j_1,j_2,j_3$ in order such that every edge of $J$ is incident to at least one of $j_1,j_2,j_3$.  Let $H=L(J)$, and for $i\in \{1,3\}$ let $X_i$ be the set of vertices of $H$ corresponding to edges incident to $j_i$ in $J$.  Both $X_1$ and $X_3$ are cliques.  Let $v_1$ and $v_2$ be the vertices of $H$ corresponding to the edges $j_1j_2$ and $j_2j_3$ respectively.  Let $G$ be a thickening of $H$ under $M=\{v_1v_2 \}$.  Then $(G,X_1,X_3)$ is a strip and if it contains a $W_5$ we say it is a {\em pseudo-line strip}\index{strip!pseudo-line}.

These strips correspond to thickenings of the class $\cZ_3$ defined in \cite{clawfree5}.  We call the vertices of $J$ other than $\{j_i \mid 1\leq i\leq 3\}$ {\em centres} of $J$.

%%%%%%%%%%%%%%%%%%%%%%%%%%%%%%%%%%%%%%%%%%%%%%%%%%%%%%%%%%%%%%%%%%%%%%%%%%%%%%%%
%%%%%%%%%%%%%%%%%%%%%%%%%%%%%%%%%%%%%%%%%%%%%%%%%%%%%%%%%%%%%%%%%%%%%%%%%%%%%%%%
%%%%%%%%%%%%%%%%%%%%%%%%%%%%%%%%%%%%%%%%%%%%%%%%%%%%%%%%%%%%%%%%%%%%%%%%%%%%%%%%
\subsubsection{Gear strips}

Take a graph $H$ on vertices $\{v_1, \ldots, v_{10}\}$ with adjacency as follows.  The vertices $v_1, \ldots, v_6$ are a 6-hole in order.  Next, $v_7$ is adjacent to $v_1,v_2,v_3,v_6$; $v_{8}$ is adjacent to $v_3, v_4, v_5, v_6, v_7$; $v_{9}$ is adjacent to $v_3,v_4,v_6,v_1,v_7,v_{8}$; $v_{10}$ is adjacent to $v_2,v_3,v_5,v_6,v_7,v_{8}$.  There are no other edges in $H$.  Let $X\subseteq\{v_{9},v_{10}\}$.  See Figure \ref{fig:xxx}.

\begin{figure}
\begin{center}
\includegraphics[scale=.6]{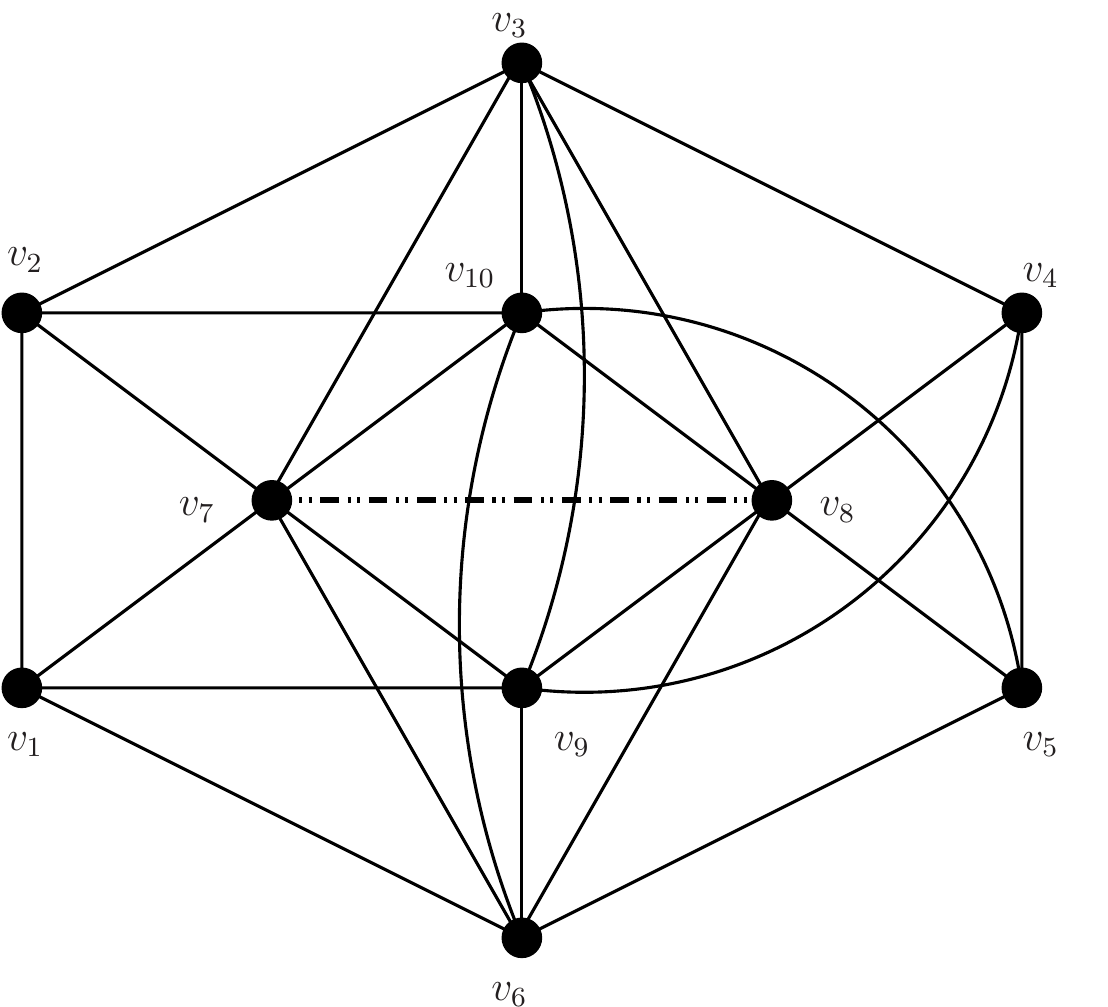}
\end{center}
\caption{\small{The graph underlying gear strips.}}
\label{fig:xxx}
\end{figure}

If $G$ is a thickening of $H\setminus X$ under a matching $M\subseteq \{v_7v_{8}\}$, then $(G,I(v_1)\cup I(v_2), I(v_4)\cup I(v_5))$ is a strip, and we say that it is a {\em gear strip}, following the terminology of Galluccio, Gentile, and Ventura \cite{gallucciogv08}.  These correspond to thickenings of the class $\cZ_5$ in \cite{clawfree5} and are a slight generalization of thickenings of XX-strips as defined in \cite{cssurvey}.

%%%%%%%%%%%%%%%%%%%%%%%%%%%%%%%%%%%%%%%%%%%%%%%%%%%%%%%%%%%%%%%%%%%%%%%%%%%%%%%%
%%%%%%%%%%%%%%%%%%%%%%%%%%%%%%%%%%%%%%%%%%%%%%%%%%%%%%%%%%%%%%%%%%%%%%%%%%%%%%%%
%%%%%%%%%%%%%%%%%%%%%%%%%%%%%%%%%%%%%%%%%%%%%%%%%%%%%%%%%%%%%%%%%%%%%%%%%%%%%%%%
%%%%%%%%%%%%%%%%%%%%%%%%%%%%%%%%%%%%%%%%%%%%%%%%%%%%%%%%%%%%%%%%%%%%%%%%%%%%%%%%
\subsection{Five types of generalized 2-joins}

We now define generalized 2-joins corresponding to these five types of strips.  Suppose in our claw-free graph $G$ there is a generalized 2-join $((X_1,Y_1),(X_2,Y_2))$ separating $G_1$ and $G_2$, such that $X_1$, $X_2$, $Y_1$, and $Y_2$ are cliques and are pairwise disjoint except for possibly $X_1$ and $Y_1$.

\begin{itemize*}
\item {\bf Canonical interval 2-joins.}  If $(G_2,X_2,Y_2)$ is a linear interval strip with $X_2$ and $Y_2$ disjoint such that $G_2$ is not a clique, then we say that $((X_1,Y_1),(X_2,Y_2))$ is a {\em canonical interval 2-join}.
\item {\bf Antihat 2-joins.}  If $(G_2,X_2,Y_2)$ is an antihat strip then we say that $((X_1,Y_1),(X_2,Y_2))$ is an {\em antihat 2-join}\index{2-join!antihat}.
\item {\bf Strange 2-joins.}  If $(G_2,X_2,Y_2)$ is a strange strip then we say that $((X_1,Y_1),(X_2,Y_2))$ is a {\em strange 2-join}\index{2-join!strange}.
\item {\bf Pseudo-line 2-joins.} If $(G_2,X_2,Y_2)$ is a pseudo-line strip then we say that $((X_1,Y_1),(X_2,Y_2))$ is a {\em pseudo-line 2-join}.
\item {\bf Gear 2-joins.}  If $(G_2,X_2,Y_2)$ is a gear strip then we say that $((X_1,Y_1),(X_2,Y_2))$ is a {\em gear 2-join}\index{2-join!gear}.
\end{itemize*}

Aside from a limited set of exceptions, every claw-free graph that is not quasi-line or three-cliqued or an antiprismatic thickening admits one of these five types of generalized 2-join; we will state this more formally in the next section.  We must now prove that no minimum counterexample to Theorem \ref{thm:main} admits any of these five types of generalized 2-join.

%%%%%%%%%%%%%%%%%%%%%%%%%%%%%%%%%%%%%%%%%%%%%%%%%%%%%%%%%%%%%%%%%%%%%%%%%%%%%%%%
%%%%%%%%%%%%%%%%%%%%%%%%%%%%%%%%%%%%%%%%%%%%%%%%%%%%%%%%%%%%%%%%%%%%%%%%%%%%%%%%
%%%%%%%%%%%%%%%%%%%%%%%%%%%%%%%%%%%%%%%%%%%%%%%%%%%%%%%%%%%%%%%%%%%%%%%%%%%%%%%%
%%%%%%%%%%%%%%%%%%%%%%%%%%%%%%%%%%%%%%%%%%%%%%%%%%%%%%%%%%%%%%%%%%%%%%%%%%%%%%%%
\subsection{Dealing with the first four types}

We begin by proving a lemma that implies that a minimum counterexample to Theorem \ref{thm:main} or Conjecture \ref{con:local} cannot admit a canonical interval 2-join, an antihat 2-join, a strange 2-join, or a gear 2-join.  Then we prove a lemma that implies that a that a minimum counterexample to Theorem \ref{thm:main} or Conjecture \ref{con:main} cannot admit a pseudo-line 2-join.  For each of these two tasks we need to define a special colouring invariant.

Given $G$ admitting a generalized 2-join $((X_1, Y_1),(X_2, Y_2))$, let $H_2$ denote $G[V_2 \cup X_1 \cup Y_1]$ and denote $G_2-X_2-Y_2$ by $Z_2$.  For the first four types of generalized 2-join, we will use a local invariant $\gamma_\ell^j(H_2)\leq \gamma_\ell(G)$ which is easier to control when extending a partial colouring across a generalized 2-join.  For pseudo-line 2-joins we will use an analogous global invariant $\gamma_g^j(H_2)\leq\gamma(G)$.

For a set of vertices $S$ we define $\Delta_G(S)$ as $\max_{v\in S}d_G(v)$.  Likewise we define $\omega(S)$ as $\max_{v\in S}\omega(v)$ and $\gamma_\ell(S)$ as $\max_{v\in S}\gamma_\ell(v)$.  For $v \in H_2$ we define $\omega'(v)$ as the size of the largest clique in $H_2$ containing $v$ and not intersecting both $X_1 \setminus Y_1$ and $Y_1 \setminus X_1$ (basically, $\omega'(v)$ is the largest clique containing $v$ that we can easily manage).  Let $\omega'(H_2)$ denote $\max_{v\in V(H_2)}\omega'(v)$.  Now we define:

$$\gamma_\ell^j(H_2) = \max_{v\in V_2\cup X_1\cup Y_1}\left\lceil\tfrac 12( d_G(v)+1+\omega'(v))\right\rceil.$$

$$\gamma_g^j(H_2) = \left\lceil\tfrac 12( \Delta_G(V(H_2))+1+\omega'(H_2))\right\rceil.$$

Observe that $\gamma_\ell^j(H_2) \leq \gamma_\ell(G)$ and $\gamma_g^j(H_2)\leq \gamma(G)$.  Note that if $v \in X_1 \cup Y_1$, then $\omega'(v)$ is $|X_1|+|X_2|$, $|Y_1|+|Y_2|$, or $|X_1\cap Y_1|+ \omega(G[X_2 \cup Y_2])$.  In \cite{kingthesis} (and \cite{chudnovskykps12}) we proved:

\begin{lemma}\label{lem:quasilinemce2}
Let $G$ be a graph and suppose $G$ admits a canonical interval 2-join $((X_1, Y_1),(X_2, Y_2))$.  Then given a proper $l$-colouring of $G_1$ for any $l \geq \gamma_\ell^j(H_2)$, we can find a proper $l$-colouring of $G$ in $O(nm)$ time.
\end{lemma}

Since $\gamma_\ell^j(H_2) \leq \gamma_\ell(G) \leq \gamma(G)$ this lemma implies that no minimum counterexample to Theorem \ref{thm:main} or Conjecture \ref{con:local} admits a canonical interval 2-join.  We now prove a corresponding lemma for antihat, strange, and gear 2-joins in a skeletal claw-free graph.

\begin{lemma}\label{lem:compositionlocal}
Suppose a skeletal claw-free graph $G$ admits a canonical interval 2-join or an antihat 2-join or a strange 2-join or a gear 2-join $((X_1, Y_1),(X_2, Y_2))$.  Then given a proper $l$-colouring of $G_1$ for any $l \geq \gamma_\ell^j(H_2)$, we can find a proper $l$-colouring of $G$.
\end{lemma}

As Lemma \ref{lem:quasilinemce2} deals with canonical interval 2-joins, we can split the remainder of the proof up into three lemmas corresponding to antihat 2-joins, strange 2-joins, and gear 2-joins.  Our approach in each case is to set up the colouring of $G_1$ so that we can do one of two things.  When possible, we colour $G_2$ directly by constructing an auxiliary graph from $G_2$ and appealing to perfection or Theorem \ref{thm:local}.  If that is not possible then we remove stable sets, reducing $\gamma_\ell^j(H_2)$ each time, until $G_2$ becomes degenerate and we can appeal to a previous result.

Observe that it suffices to prove the case $l=\gamma_\ell^j(H_2)$.  For if $l>\gamma_\ell^j(H_2)$, we can simply remove $l-\gamma_\ell^j(H_2)$ arbitrarily chosen colour classes and deal with what remains.

%%%%%%%%%%%%%%%%%%%%%%%%%%%%%%%%%%%%%%%%%%%%%%%%%%%%%%%%%%%%%%%%%%%%%%%%%%%%%%%%
\subsubsection{Antihat 2-joins}

\begin{lemma}\label{lem:compantihat}
Suppose a skeletal claw-free graph $G$ admits an antihat 2-join $((X_1, Y_1),(X_2, Y_2))$.  Then given a proper $l$-colouring of $G_1$ for any $l \geq \gamma_\ell^j(H_2)$, we can find a proper $l$-colouring of $G$.
\end{lemma}

\begin{proof}
Consider a minimum counterexample for some fixed $l$.  If $G_2$ contains a skeletal homogeneous pair of cliques $(A,B)$ then one of $A$ and $B$ is partially but not completely contained in one of $X_2$ or $Y_2$.

Let $k$ be the number of colours appearing in both $X_1$ and $Y_1$.  We begin by making $k$ minimal, as we did in Case 6 of the proof of Lemma \ref{lem:quasilinemce2}.  To do this, we simply find a vertex $v$ in $X_1\cup Y_1$ with a colour appearing in both $X_1$ and $Y_1$, such that some colour $i$ does not appear in $X_1\cup Y_1\cup N(v)$, and recolour $v$ with $i$.  This minimality of $k$ ensures a bound on $l$, as long as $k\geq 1$:  Let vertices $u\in X_1$ and $v\in Y_1$ have the same colour.  Then $d(u)+1 \geq |X_2|+ (l-|Y_1|+k)$, since minimality ensures that $u$ has a neighbour in $G_1$ of every colour except possibly those in $Y_1$ not appearing in $X_1$.  Similarly, $d(v)+1\geq |Y_2|+(l-|X_1|+k)$.  Therefore since $\omega'(u)$ and $\omega'(v)$ are at least $|X_1|+|X_2|$ and $|Y_1|+|Y_2|$ respectively, $l\geq |X_2|+\frac 12 (l+k+|X_1|-|Y_1|)$ and $l \geq |Y_2|+\frac 12 (l+k+|Y_1|-|X_1|)$.  Consequently $l \geq |X_2|+|Y_2|+k$ if $k>0$.

Suppose there is a colour class $S$ in $G_1$ hitting $X_1$ but not $Y_1$.  Then add to this colour class a stable set $S'$ of size two intersecting $Y_2$ and $Z_2$.  By the structure of antihat thickenings, we can assume that $S'$ intersects $I(b_1)$ and $I(c_1)$ without loss of generality.  If $\Omega(b_1c_1)$ is nonempty, we insist that $S'$ intersect it.

Note first that every vertex in $I(b_1)\cup I(c_1)$ is trumped or has a twin in $S'$ or has two neighbours in $S'$.  Every vertex in $Z_2 \setminus I(c_1)$ has two neighbours in $S'$, as does every vertex in $Y_2 \setminus I(b_1)$.  Every vertex in $X_2$ has a neighbour in $S$ and a neighbour in $S'$ -- this neighbour will be in $Y_2$ for vertices in $I(a_1)$, and in $Z_2$ for all other vertices in $X_2$ (recall from the definition of an antihat 2-join that since $I(b_1)$ and $I(c_1)$ are both nonempty, $I(a_1)$ is complete to $I(b_1)$ and anticomplete to $I(c_1)$).  Thus $\gamma_\ell^j(v)$ drops for any $v\in G_2$.  Since $S\cup S'$ intersects both $X_1\cup X_2$ and $Y_1\cup Y_2$, $\gamma_\ell^j(v)$ drops for any $v\in X_1\cup Y_1$.  Therefore we remove $S\cup S'$ and lower $\gamma_\ell^j(H_2)$.

We repeat this approach until either $Y_2\cup Z_2$ is a clique, or all colours in $X_1$ appear in $Y_1$.  Suppose we remove $t_1$ stable sets in this way.  We then take colour classes of $G_1$ hitting $Y_1$ but not $X_1$, and remove them along with stable sets of size two in $X_2\cup Z_2$, using the symmetric argument to show that $\gamma_\ell^j(H_2)$ drops each time.  We do this until either all colours appearing in $Y_1$ are in $X_1$, or until $X_2\cup Z_2$ is a clique.  Let $t_2$ be the number of stable sets we remove in this way, let $S_1$ be the set of all vertices we have removed from $G$, and let $t=t_1+t_2$.  Notice that $\gamma_\ell^j(H_2-S_1)\leq \gamma_\ell^j(H_2)-t$.

Suppose $X_1\setminus S_1$ is empty.  Then we can colour $G_2-S_1$ using $l-t$ colours by Theorem \ref{thm:local}, since $G_2$ is three-cliqued.  Since $Y_1\setminus S_1$ is a clique cutset in $G-S_1$, this immediately gives us an $(l-t)$-colouring of $G-S_1$ and therefore an $l$-colouring of $G$.  So we can assume $X_1\setminus S_1$ and symmetrically $Y_1\setminus S_1$ are nonempty.

Now suppose every colour hitting $Y_1\setminus S_1$ also hits $X_1 \setminus S_1$.  Again we $(l-t)$-colour $G_2-S_1$, noting that at most $|X_2|+|Y_2|-t$ colours appear on $(X_2\cup Y_2)\setminus S_2$ because $|(X_2\cup Y_2)\setminus S_2|=|X_2|+|Y_2|-t$.  We ensure that no colour hits both $X_1$ and $X_2$, and that no colour hits both $Y_1$ and $Y_2$.  This is possible because $l-t> |(X_1\cup X_2)\setminus S_1|$ and $l-t \geq |X_2|+|Y_2|+k-t$, as we proved above.  This gives us a proper $(l-t)$-colouring of $G-S_1$, and therefore an $l$-colouring of $G$.

By symmetry, this covers the case in which every colour hitting $X_1\setminus S_1$ also hits $Y_1\setminus S_1$.  Thus there is a colour in $X_1$ but not $Y_1$, and one in $Y_1$ but not $X_1$.  So our method stopped because both $(Y_2\cup Z_2)\setminus S_1$ and $(X_2\cup Z_2)\setminus S_1$ are cliques.

In this final case, we $(l-t)$-colour $G_2-S_1$ by applying Lemma \ref{lem:quasilinemce2} as follows.  Notice that $(X_2\setminus S_1, Y_2\setminus S_1)$ is a homogeneous pair of cliques in $G-S_1$.  We reduce it to a skeletal homogeneous pair of cliques without changing the chromatic number using Theorem \ref{thm:skelhp}; the result is a graph $G'$ in which $((X_1\setminus S_1, Y_1\setminus S_1),(X_2\setminus S_1, Y_2\setminus S_1))$ is a canonical interval 2-join.  We can therefore apply Lemma \ref{lem:quasilinemce2} to find an $(l-t)$-colouring of $G'$.  Again using Lemma \ref{lem:reduction}, we can construct an $(l-t)$-colouring of $G-S_1$.  This immediately gives us an $l$-colouring of $G$, proving the lemma.
\end{proof}

%%%%%%%%%%%%%%%%%%%%%%%%%%%%%%%%%%%%%%%%%%%%%%%%%%%%%%%%%%%%%%%%%%%%%%%%%%%%%%%%
\subsubsection{Strange 2-joins}

The next case is strange 2-joins; we use a similar approach.

\begin{lemma}\label{lem:compstrange}
Suppose a skeletal claw-free graph $G$ admits a strange 2-join $((X_1, Y_1),(X_2, Y_2))$.  Then given a proper $l$-colouring of $G_1$ for any $l \geq \gamma_\ell^j(H_2)$, we can find a proper $l$-colouring of $G$.
\end{lemma}

\begin{proof}
As in the proof of the previous lemma, assume $\gamma_\ell^j(H_2)=l$ and let $k$ denote the number of colours appearing in both $X_1$ and $Y_1$.  We begin by modifying the colouring of $G_1$ so that $k$ is minimal, so again we can assume that either $k=0$ or $l\geq |X_2|+|Y_2|+k$.  Denote $G_2-X_2-Y_2$ by $Z_2$.

Let $t = \min\{|I(a_1)|, |I(c_1) \cap \Omega(c_1,b_3)|, |Y_2|-k  \}$.  We remove $t$ colours hitting $Y_1$ but not $X_1$.  With each colour class we remove a vertex of $I(a_1)$ and a vertex of $I(c_1) \cap \Omega(c_1,b_3)$.  Together these vertices form $t$ stable sets; call their union $S_1$.  As in the proof of the previous lemma, we now consider our situation depending on the value of $t$.  Note that each time we remove a stable set, every vertex in $G_2$ is either trumped or loses two neighbours or loses a twin.  It is therefore easy to see that $\gamma_\ell^j(H_2-S_1)\leq \gamma_\ell^j(H_2)-t$.

Suppose $I(a_1)$ is empty.  We apply Lemma \ref{lem:quasilinemce2} to $(l-t)$-colour $G-S_1$ as follows.  First observe that removing $S_1$ turns $((X_1\setminus S_1, Y_1\setminus S_1),(X_2\setminus S_1, Y_2\setminus S_1))$ into a {\em fuzzy linear interval 2-join}, meaning that we can turn it into a canonical interval 2-join by reducing nonskeletal homogeneous pairs of cliques:  $(Z_2\setminus S_1, Y_2 \setminus S_1)$ is a homogeneous pair of cliques in $G-S_1$, so we can reduce it to a skeletal homogeneous pair of cliques using Theorem \ref{thm:skelhp}, at which point $((X_1\setminus S_1, Y_1\setminus S_1),(X_2\setminus S_1, Y_2\setminus S_1))$ becomes a canonical linear interval 2-join in a proper claw-free subgraph $G'$ of $G-S_1$.  We can therefore apply Lemma \ref{lem:quasilinemce2} to $G'$, since we already have an $(l-t)$-colouring of $G_1-S_1$, to find an $(l-t)$-colouring of $G'$.  Theorem \ref{thm:skelhp} tells us that we can use this colouring to construct an $(l-t)$-colouring of $G-S_1$.  Combining this with a $t$-colouring of $G[S_1]$ gives us an $l$-colouring of $G$.

Now suppose $I(c_1) \cap \Omega(c_1,b_3)$ is empty but $I(c_1)$ is not empty.  To $(l-t)$-colour $G-S_1$, we first remove the vertices of $I(b_3)$, which have become simplicial.  Now observe that $((X_1\setminus S_1, Y_1\setminus S_1),(X_2\setminus S_1, Y_2\setminus (I(b_3)\cup S_1))$ is an antihat 2-join.  The remaining sets of $G_2$ are $I(a_1)$, $I(a_2)$, $I(b_1)$, $I(b_2)$, $I(c_1)$, and $I(c_2)$.  To see the antihat 2-join, we relabel these sets as in the definition of an antihat thickening as follows:  $(I(a_1),I(a_2))\rightarrow (I(a_1),I(a_2))$, $(I(b_1),I(b_2))\rightarrow (I(b_1),I(b_3))$, and $(I(c_1),I(c_2))\rightarrow (I(c_3),I(c_1))$.  %Respectively, these sets become $I(a_1)$, $I(a_2)$, $I(b_1)$, $I(b_3)$, $I(c_3)$, and $I(c_1)$ (see Section \ref{sec:structureantihat}).
  We can therefore apply Lemma \ref{lem:compantihat} to find an $(l-t)$-colouring of $G-(S_1\cup I(b_0))$, then replace and colour the simplicial vertices in $I(b_0)$ to get an $(l-t)$-colouring of $G-S_1$.  This gives us an $l$-colouring of $G$, completing the case of strange 2-joins.
\end{proof}

%%%%%%%%%%%%%%%%%%%%%%%%%%%%%%%%%%%%%%%%%%%%%%%%%%%%%%%%%%%%%%%%%%%%%%%%%%%%%%%%
\subsubsection{Gear 2-joins}

The final and most difficult case is that of gear 2-joins.

\begin{lemma}\label{lem:compgear}
Suppose a skeletal claw-free graph $G$ admits a gear 2-join $((X_1, Y_1),(X_2, Y_2))$.  Then given a proper $l$-colouring of $G_1$ for any $l \geq \gamma_\ell^j(H_2)$, we can find a proper $l$-colouring of $G$.
\end{lemma}

\begin{proof}
We proceed by induction on $|G|$, taking as our basis the trivial case in which $\min\{|X_1|,|Y_1|\}=0$; in this case we have a 1-join and the result follows from Theorem \ref{thm:local} since gear strips are three-cliqued.  So assume both $X_1$ and $Y_1$ are nonempty.  Let $Z_2$ denote $G_2\setminus (X_2\cup Y_2)$.  Again we can let $G$ be a minimum counterexample and assume that $l= \gamma_\ell^j(H_2)$.

In this case we make $k$, the overlap between $X_1$ and $Y_1$ in the colouring of $G_1$, maximal.\\

\noindent{\bf Case 1: } $k>0$.

If $k>0$, we remove a colour class hitting both $X_1$ and $Y_1$, along with one vertex each of $I(v_9)$ and $I(v_{10})$, if they are both nonempty.  In this case every vertex of $G_2$ loses a twin or two neighbours.  Since we remove a vertex in both $X_1$ and $Y_1$, it is easy to see that $\gamma_\ell^j(H_2)$ drops.  Since removing vertices from $I(v_9)$ and $I(v_{10})$ will not change the fact that we have a gear 2-join, we can proceed by induction, having reduced both $\gamma_\ell^j(H_2)$ and $l$.

So assume that $I(v_9)\cup I(v_{10})$ is a clique, i.e.\ one of $I(v_9)$ and $I(v_{10})$ is empty.  We do the same thing, but instead we remove a colour class hitting both $X_1$ and $Y_1$, along with a vertex of $I(v_3)$ and a vertex of $I(v_6)$.  Clearly $\gamma_\ell^j(H_2)$ drops as before and we can proceed by induction, since as long as neither $I(v_3)$ nor $I(v_6)$ becomes empty we will still have a gear 2-join.

Suppose $I(v_6)$ becomes empty, and one of $I(v_9)$ and $I(v_{10})$ is empty.  By symmetry we can assume that $I(v_{9})$ is empty.  We are now left with a fuzzy linear interval 2-join: Reducing (if necessary) the possibly nonlinear homogeneous pairs of cliques $(I(v_7),I(v_8))$ and $(I(v_3)\cup I(v_{10}), I(v_4)\cup I(v_5))$ leaves us with a canonical interval 2-join.  The vertices, in linear order, are $I(v_1)$, $I(v_2)$, $I(v_7)$, $I(v_3)\cup I(v_{10})$, $I(v_8)$, $I(v_4) \cup (v_5)$.  The reader can confirm this, along with symmetry between $v_9$ and $v_{10}$, by consulting Figure \ref{fig:xxx}.  So, as in the proof of the previous lemma, we can find our $l$-colouring of $G$ by reducing on these two homogeneous pairs of cliques and invoking Lemma \ref{lem:quasilinemce2}.

This completes the proof of the lemma when $k>0$.\\

\noindent{\bf Case 2:} $k=0$; $l > |X_1|+|Y_1|$.

In this case we remove a colour class hitting neither $X_1$ nor $Y_1$, along with a stable set of size three in $G_2$.  Call their union $S$.  If $I(v_{10})$ is nonempty, we remove a vertex of $I(v_{10})$ along with on vertex each of $I(v_1)$ and $I(v_4)$.  Every vertex in $G_2$ loses a twin or two neighbours, so it is easy to confirm that $\gamma_\ell^j(H_2)$ drops.  Thus we can proceed by induction, provided that both $I(v_1)$ and $I(v_4)$ are still nonempty.

If $I(v_1)$ and $I(v_4)$ are both empty, then we extend the colouring of $G_1$ to an $l$-colouring of $G_1\cup I(v_2)\cup I(v_5)$.  We then note that $((I(v_2)\cup I(v_{10})\setminus S, I(v_5)\cup I(v_{10})\setminus S), (I(v_3)\cup I(v_7), I(v_6)\cup I(v_8)  ))$ is a fuzzy linear interval 2-join, in which $(I(v_3)\cup I(v_7), I(v_6)\cup I(v_8))$ is the only possible nonlinear homogeneous pair of cliques.  So we can construct an $(l-1)$-colouring of $G-S$ by Lemma \ref{lem:quasilinemce2} as in the previous two proofs.  This gives us an $l$-colouring of $S$.

So assume $I(v_1)$ is now empty but $I(v_4)$ is not.  Clearly we can extend the $(l-1)$-colouring of $G_1-S$ to a proper $(l-1)$-colouring of $(G_1-S)\cup I(v_2)$.  We claim that we now have an antihat 2-join and we can find an $(l-1)$-colouring of $G-S$ using Lemma \ref{lem:compantihat}.

The 2-join in $G-S$ is $((I(v_2), Y_1\setminus S), ((I(v_3)\cup I(v_7))\setminus S, Y_2\setminus S))$.  To see that $(G_2-S)-(I(v_1)\cup I(v_2))$ is an antihat strip, we will relabel the vertices to conform with the definition of an antihat thickening.  We relabel the sets $I(v_3)$, $I(v_{10})$, and $I(v_7)$ as $I(a_1)$, $I(a_2)$, and $I(a_3)$ respectively.  We relabel $I(v_{4})$ and $I(v_{5})$ as $I(b_{1})$ and $I(b_{2})$ respectively.  Finally, we relabel $I(v_{6})$, $I(v_{9})$, and $I(v_8)$ as $I(c_1)$, $I(c_2)$, and $I(c_3)$ (or $I(c_4)$ if $I(v_7)\cup I(v_8)$ is a clique) respectively.  It is straightforward to confirm that this is an antihat strip.  We therefore have an antihat 2-join in $G-S$, so by Lemma \ref{lem:compantihat} we can find an $(l-1)$-colouring of $G-S$ and an $l$-colouring of $G$.

If $I(v_{10})$ is empty, then instead of taking vertices from $I(v_{10})$, $I(v_{1})$ and $I(v_4)$, we take vertices from $I(v_1)$, $I(v_3)$ and $I(v_5)$, and proceed symmetrically.  This time, we may worry that $I(v_3)$ will become empty, but in this case, since $I(v_{10})$ is also empty, we get a fuzzy linear interval 2-join exactly as in Case 1.\\

\noindent{\bf Case 3: }$k=0$; $l=|X_1|+|Y_1|$.

In this final case, every colour appears in $X_1\cup Y_1$, and no colour appears twice.  Therefore $X_2$ and $Y_2$ must receive colours appearing in $Y_1$ and $X_1$ respectively.  Since $k$ is maximal, $l\geq |X_2|+|X_1|+\frac 12 |Y_1|$ (from a vertex in $X_1$), and $l\geq |Y_2|+|Y_1|+\frac 12 |X_1|$ (from a vertex in $Y_1$).  It follows that $2l \geq \frac 32(|X_1|+|Y_1|)+|X_2|+|Y_2|$, so $|X_2|+|Y_2|\leq \frac 12 l$.

Notice that $Z_2$ is cobipartite, and that the only non-edges in $Z_2$ are in $I(v_3)\cup I(v_6)$, $I(v_7)\cup I(v_8)$, and $I(v_9)\cup I(v_{10})$.  We begin with an optimal colouring of $Z_2$, removing the colour classes of size two.  Let $t_1$ be the number of such colour classes in $I(v_3)\cup I(v_6)$, and let $t$ be the total number of such colour classes.  Denote these $2t$ vertices by $S$, noting that $Z_2-S$ is a clique.

We construct an auxiliary graph $G'$ from $G_2-S$ by adding all possible edges between $X_2$ and $Y_2$.  Now $G'$ is cobipartite and perfect, and since a proper colouring of $G'$ will give vertices in $X_2$ and $Y_2$ distinct colours, it suffices to prove that $\omega(G') \leq l-t$.  This gives us an $l$-colouring of $G_2$ in which no colour appears twice on $X_2\cup Y_2$, so we can use it to extend the $l$-colouring of $G_1$ to an $l$-colouring of $G$.

Suppose there is a clique $W$ of size greater than $l-t$ in $G'$.  We will now prove that $l-|X_2|-|Y_2|\geq \frac 12 |Z_2|\geq t$, which implies that $W$ cannot be $X_2\cup Y_2$.  Consider vertices $u, v, x, y$ in $I(v_1)$, $I(v_2)$, $I(v_4)$, and $I(v_5)$ respectively.  Since every vertex in $Z_2$ has two neighbours in this set, the sum of the four degrees is at least $2(|X_1|+|X_2|+|Y_1|+|Y_2|+|Z_2|)-4$.  Therefore the sum $\gamma_\ell^j(u)+ \gamma_\ell^j(v)+ \gamma_\ell^j(x)+ \gamma_\ell^j(y)$ is at least $4l \geq 2(|X_1|+|X_2|+|Y_1|+|Y_2|)+|Z_2|$.  Thus $2l \geq |Z_2| + 2(|X_2|+|Y_2|)$, so $\frac 12 |Z_2|+|X_2|+|Y_2|\leq l$.

A maximal clique $W$ in $G'$ intersecting both $I(v_1)$ and $I(v_2)$ as well as $Z_2$ must be $(I(v_1)\cup I(v_2)\cup I(v_7))\setminus S$.  But a vertex $v$ in $(I(v_7)\cap \Omega(v_7,v_8))\setminus S$ (this set is nonempty because $(I(v_7),I(v_8))$ is a skeletal homogeneous pair) has either two neighbours or a twin in each stable set of size two in $S$.  This means that if $|W|>l-t$, then $\gamma_\ell^j(v)>l$, a contradiction.  So $W$ is not such a clique, and by symmetry $W$ does not intersect all three of $I(v_4)$, $I(v_5)$, and $I(v_8)$.  A similar argument implies that $W$ cannot intersect only one of $I(v_1)$, $I(v_2)$, $I(v_4)$, and $I(v_5)$.  Since $|X_2|+|Y_2|\leq l-t$ we can see that $W$ cannot intersect three of these sets.  Furthermore $|Z_2-S| = \omega(Z_2)-t\leq l-t$, so $W$ cannot be contained in $Z_2-S$.  Therefore $W$ intersects all three of $X_2$, $Y_2$, and $Z_2$, and we can assume by symmetry that $W$ is $I(v_4)\setminus S$ and its neighbourhood in $X_2\cup Y_2$, i.e.\ $(I(v_2)\cup I(v_3) \cup I(v_4))\setminus S$.  

Suppose that $|W|>l-t$.  This inequality will provide us with new bounds on $l$, giving us a contradiction and completing the proof of the lemma.  Let $u$ and $v$ be vertices in $I(v_2)$ and $I(v_4)$ respectively.  Observe that $d(u)+1 \geq |X_1|+|X_2|+|I(v_3)\setminus S|+t$, since $u$ sees one vertex in every stable set in $S$.  Thus $d(u)+1 \geq |X_1|+|X_2|+|I(v_3)|+(t-t_1)$, and likewise $d(v)+1 \geq |Y_1|+|Y_2|+|I(v_3)|+(t-t_1)$.  Since $I(v_2)\cup I(v_3)\cup I(v_{10})\cup I(v_7)$ is a clique, it follows that $\omega'(u) \geq |I(v_2)|+|I(v_3)|+(t-t_1)$, because every stable set of $S$ hits $I(v_3)\cup I(v_7)\cup I(v_{10})$ exactly once.  Likewise, $\omega'(v) \geq |I(v_4)|+|I(v_3)|+(t-t_1)$.  The sum of these figures is at most $2\gamma_\ell^j(u)+2\gamma_\ell^j(v)$, which is at most $4l$.  This implies:
$$4l \geq  (|X_1|+|Y_1|) + (|X_2|+|Y_2|) + 4(t-t_1) + 4|I(v_3)| + |I(v_2)|+ I(v_4)|.$$

We know that $|X_1|+|Y_1| = l$, $|X_2|+|Y_2|>|I(v_2)|+|I(v_4)|$, and by assumption, $|I(v_2)|+|I(v_4)|+|I(v_3)|-t_1 > l-t$.  Therefore,
\begin{eqnarray*}
3l &\geq& 2( |I(v_2)|+|I(v_4)|+|I(v_3)| ) + 2|I(v_3)| + 4(t-t_1) \\
 &\geq& 2l   + 2|I(v_3)| + 2(t-t_1)
\end{eqnarray*}
Thus $|I(v_3)|-t_1 \leq \frac l2 -t$.  And since $|X_2|+|Y_2|\leq \frac l2$, we get $|W| = |I(v_2)|+|I(v_3)|+|I(v_4)| - t_1 \leq l-t$, contrary to our assumption.

It follows that $\omega(G')\leq l-t$, so we can indeed complete the $l$-colouring of $G_2$ that is compatible with the colouring of $G_1$.  This proves the lemma.
\end{proof}

Lemmas \ref{lem:compantihat}, \ref{lem:compstrange}, and \ref{lem:compgear} together immediately imply Lemma \ref{lem:compositionlocal}.\\

%%%%%%%%%%%%%%%%%%%%%%%%%%%%%%%%%%%%%%%%%%%%%%%%%%%%%%%%%%%%%%%%%%%%%%%%%%%%%%%%
%%%%%%%%%%%%%%%%%%%%%%%%%%%%%%%%%%%%%%%%%%%%%%%%%%%%%%%%%%%%%%%%%%%%%%%%%%%%%%%%
%%%%%%%%%%%%%%%%%%%%%%%%%%%%%%%%%%%%%%%%%%%%%%%%%%%%%%%%%%%%%%%%%%%%%%%%%%%%%%%%
%%%%%%%%%%%%%%%%%%%%%%%%%%%%%%%%%%%%%%%%%%%%%%%%%%%%%%%%%%%%%%%%%%%%%%%%%%%%%%%%
\subsection{Dealing with pseudo-line 2-joins}

To deal with pseudo-line 2-joins we use $\gamma_g^j(H_2)$ rather than $\gamma_\ell^j(H_2)$.

\begin{lemma}\label{lem:composition}
Suppose a skeletal claw-free graph $G$ admits a canonical interval 2-join or an antihat 2-join or a strange 2-join or a gear 2-join or a pseudo-line 2-join $((X_1, Y_1),(X_2, Y_2))$.  Then given a proper $l$-colouring of $G_1$ for any $l \geq \gamma_g^j(H_2)$, we can find a proper $l$-colouring of $G$.
\end{lemma}

\begin{proof}
We prove the lemma by induction on $l$.  We let $G$ be a minimum counterexample, noting that $l=\gamma_g^j(H_2)$.  Assume that $|X_1|\geq |Y_1|$.

If $((X_1, Y_1),(X_2, Y_2))$ is a canonical interval 2-join or an antihat 2-join or a strange 2-join or a gear 2-join, then the lemma is immediately implied by Lemma \ref{lem:compositionlocal} given the observation that $\gamma_\ell^j(H_2)\leq \gamma_g^j(H_2)$.  So we can assume that we have a pseudo-line 2-join.

Recall that $G_2$ is based on the line graph of a graph $J$, and the vertices of $J$ other than $j_1$, $j_2$, and $j_3$ are called {\em centres}.  For a centre $t$ in $J$, we call the corresponding clique $C_t$.  That is, $C_t = \cup I(e)$ over all vertices $e$ of $H$ whose corresponding edge in $J$ is incident to $t$.  Let the edges $j_1j_2$ and $j_2j_3$ be $e_1$ and $e_2$ respectively.  Note that $Z_2$ is a clique and so is $Z_2 \cup \Omega(e_1,e_2)$.

We begin by making the number $k$ of colours in $G_1$ that hit both $X_1$ and $Y_1$ maximal.  First suppose that there is no colour class appearing in neither $X_1$ nor $Y_1$.  As in the previous proofs, $l>|X_1|$.  Since $k$ is maximal, there is a vertex $v\in X_1$ with a colour not appearing in $Y_1$, and it must have at least $l-1$ neighbours in $G_1$.  This vertex is in $X_1 \cup X_2$, so $l = \gamma_g^j(H_2)\geq \frac 12l+\frac 12|X_1|+|X_2|$.  Hence $l \geq |X_1|+2|X_2|$.  Since $l=|X_1|+|Y_1|-k$, we have $|X_2|\leq \frac 12 |Y_1|-\frac 12 k$.  Now since $|X_2|$ is nonempty, $|Y_1|> k$ and there is a vertex in $Y_1$ with a colour not appearing in $X_1$.  We can therefore apply the symmetric argument to prove that $l \geq \frac 12l+\frac 12|Y_1|+|Y_2|$, and consequently $|Y_2|\leq \frac 12 |X_1|-\frac 12 k$.

Observe that if $|Z_2| \leq \frac 12 (|X_1|+|Y_1|)$ we can easily finish the colouring by giving $X_2$ colours appearing in $Y_1$ but not $X_1$, $Y_2$ colours appearing in $X_1$ but not $Y_1$, and $Z_2$ colours appearing in both $X_1$ and $Y_1$, and any leftover colours.  In fact we can do this whenever $|Z_2| \leq l - |X_2|-|Y_2|$.  So assume $|Z_2|> l- |X_2|-|Y_2|$.  Let $A$ be a maximum clique in $G[X_2\cup Z_2]$.  Since $G[X_2\cup Z_2]$ is cobipartite, we can colour it with $|A|$ colours, $|X_2|$ of which intersect $X_2$.  Therefore if $|A|\leq l-|Y_2|$ we can colour $Y_2$ using colours that appear in $X_1$ but not in $Y_1$, then colour $X_2$ and $Z_2$ using $|A|$ colours such that those colours appearing in $X_2$ do not appear in $X_1$.

To see that $|A|\leq l-|Y_2|$, note that $\omega'(H_2)\geq |A|$ and since the degree of any vertex in $I(e_1)$ is at least $|X_1|+|X_2|+|Z_2|-1$, $l = \gamma_g^j(H_2)\geq \frac 12(|A|+|Z_2|+|X_2|+|X_1|)$.  Since $|Z_2| > l - |X_2|-|Y_2|$, this implies that $l> |A|+|X_1|-|Y_2|\geq |A|+\frac 12 |Y_2|$.  Therefore $|A|\leq l-|Y_2|$ and we can complete the $\gamma_g^j(G)$-colouring of $G$.
\\

We can now assume that there is a colour class $S$ in $G_1$ that appears in neither $X_1$ nor $Y_1$.  We will find a stable set $S_2$ in $G_2$ such that removing $S\cup S_2$ lowers $\gamma_g^j(H_2)$; this will imply that $\chi(G)\leq l$ by induction.

First note that if there are at most two centres then we actually have an antihat 2-join -- this is straightforward to confirm as there are only five vertices in $J$.  So we can assume that there are at least three centres.

Suppose we set $S_2$ to be a diad (i.e.\ a stable set of size two) in $G[I(e_1)\cup I(e_2)]$ such that $S_2$ intersects $\Omega(e_1,e_2)$ if it is nonempty.  $S_2$ exists because $G[I(e_1)\cup I(e_2)]$ is not a clique.  If removing $S\cup S_2$ does not lower $\omega_j(G)$, then there must be a maximal clique in $G_2$ disjoint from $S_2$.  Such a clique must be $C_t$ for some centre $t$ that sees $j_1$, $j_2$, and $j_3$ in $J$.

The size of $C_t$ must be at least $\max\{|X_1\cup X_2|,|Y_1\cup Y_2|, |Z_2|\} > \frac 13 |V(G_2)|$, so by the number of vertices in $G_2$ there can be at most two such ``centre cliques'' of size $\omega'(H_2)$, since they must be disjoint --  call the other one $C_{t'}$ if it exists.  If we let $S_2$ be a stable set corresponding to a matching in $J$ that hits three centres and in particular hits $t$ and (if it exists) $t'$, we can see that removing $S\cup S_2$ lowers $\omega_j(G)$ so we are done.  This $S_2$ must exist because $C_t$ intersects all of $X_2$, $Y_2$, and $Z_2$, so we can find $S_2$ unless every other centre has neighbourhood $j_2$ in $J$.  If this is the case we can again easily confirm that we have an antihat 2-join, so we are done.
\end{proof}

To prove Theorem \ref{thm:main}, it only remains to deal with icosahedral thickenings.

%%%%%%%%%%%%%%%%%%%%%%%%%%%%%%%%%%%%%%%%%%%%%%%%%%%%%%%%%%%%%%%%%%%%%%%%%%%%%%%%
%%%%%%%%%%%%%%%%%%%%%%%%%%%%%%%%%%%%%%%%%%%%%%%%%%%%%%%%%%%%%%%%%%%%%%%%%%%%%%%%
%%%%%%%%%%%%%%%%%%%%%%%%%%%%%%%%%%%%%%%%%%%%%%%%%%%%%%%%%%%%%%%%%%%%%%%%%%%%%%%%
%%%%%%%%%%%%%%%%%%%%%%%%%%%%%%%%%%%%%%%%%%%%%%%%%%%%%%%%%%%%%%%%%%%%%%%%%%%%%%%%
%%%%%%%%%%%%%%%%%%%%%%%%%%%%%%%%%%%%%%%%%%%%%%%%%%%%%%%%%%%%%%%%%%%%%%%%%%%%%%%%
%%%%%%%%%%%%%%%%%%%%%%%%%%%%%%%%%%%%%%%%%%%%%%%%%%%%%%%%%%%%%%%%%%%%%%%%%%%%%%%%
%%%%%%%%%%%%%%%%%%%%%%%%%%%%%%%%%%%%%%%%%%%%%%%%%%%%%%%%%%%%%%%%%%%%%%%%%%%%%%%%
%%%%%%%%%%%%%%%%%%%%%%%%%%%%%%%%%%%%%%%%%%%%%%%%%%%%%%%%%%%%%%%%%%%%%%%%%%%%%%%%
%%%%%%%%%%%%%%%%%%%%%%%%%%%%%%%%%%%%%%%%%%%%%%%%%%%%%%%%%%%%%%%%%%%%%%%%%%%%%%%%
%%%%%%%%%%%%%%%%%%%%%%%%%%%%%%%%%%%%%%%%%%%%%%%%%%%%%%%%%%%%%%%%%%%%%%%%%%%%%%%%
%%%%%%%%%%%%%%%%%%%%%%%%%%%%%%%%%%%%%%%%%%%%%%%%%%%%%%%%%%%%%%%%%%%%%%%%%%%%%%%%
%%%%%%%%%%%%%%%%%%%%%%%%%%%%%%%%%%%%%%%%%%%%%%%%%%%%%%%%%%%%%%%%%%%%%%%%%%%%%%%%
%%%%%%%%%%%%%%%%%%%%%%%%%%%%%%%%%%%%%%%%%%%%%%%%%%%%%%%%%%%%%%%%%%%%%%%%%%%%%%%%
%%%%%%%%%%%%%%%%%%%%%%%%%%%%%%%%%%%%%%%%%%%%%%%%%%%%%%%%%%%%%%%%%%%%%%%%%%%%%%%%
%%%%%%%%%%%%%%%%%%%%%%%%%%%%%%%%%%%%%%%%%%%%%%%%%%%%%%%%%%%%%%%%%%%%%%%%%%%%%%%%
\section{Icosahedral thickenings}\label{sec:icosahedral}

The icosahedron is the unique vertex-transitive graph on twelve vertices in which the neighbourhood of every vertex induces a $C_5$.  A result of Fouquet \cite{fouquet93} tells us that a claw-free graph with $\alpha \geq 3$ is quasi-line precisely if no neighbourhood contains an induced $C_5$, so the icosahedron is the epitome of a claw-free graph that is not quasi-line.

There are several graphs related to the icosahedron that we must treat as a structural exception, as they are not three-cliqued or antiprismatic, and they do not arise as a composition of strips, which we will define shortly.  The first is the icosahedron itself, which we define explicitly.  Let the graph $G_0$ have vertices $v_0, v_1, \ldots, v_{11}$.  For $i = 1, \ldots, 10$, $v_i$ is adjacent to $v_{i+1}$ and $v_{i+2}$ with indices modulo 10.  The neighbourhood of $v_0$ is $\{v_i: 1\leq i \leq 10$, $i$ is odd$\}$, and the neighbourhood of $v_{11}$ is $\{v_i: 1\leq i \leq 10$, $i$ is even$\}$.  $G_0$ is the icosahedron (see Figure \ref{fig:icosahedral}).

\begin{figure}
\begin{center}
\includegraphics[scale=.6]{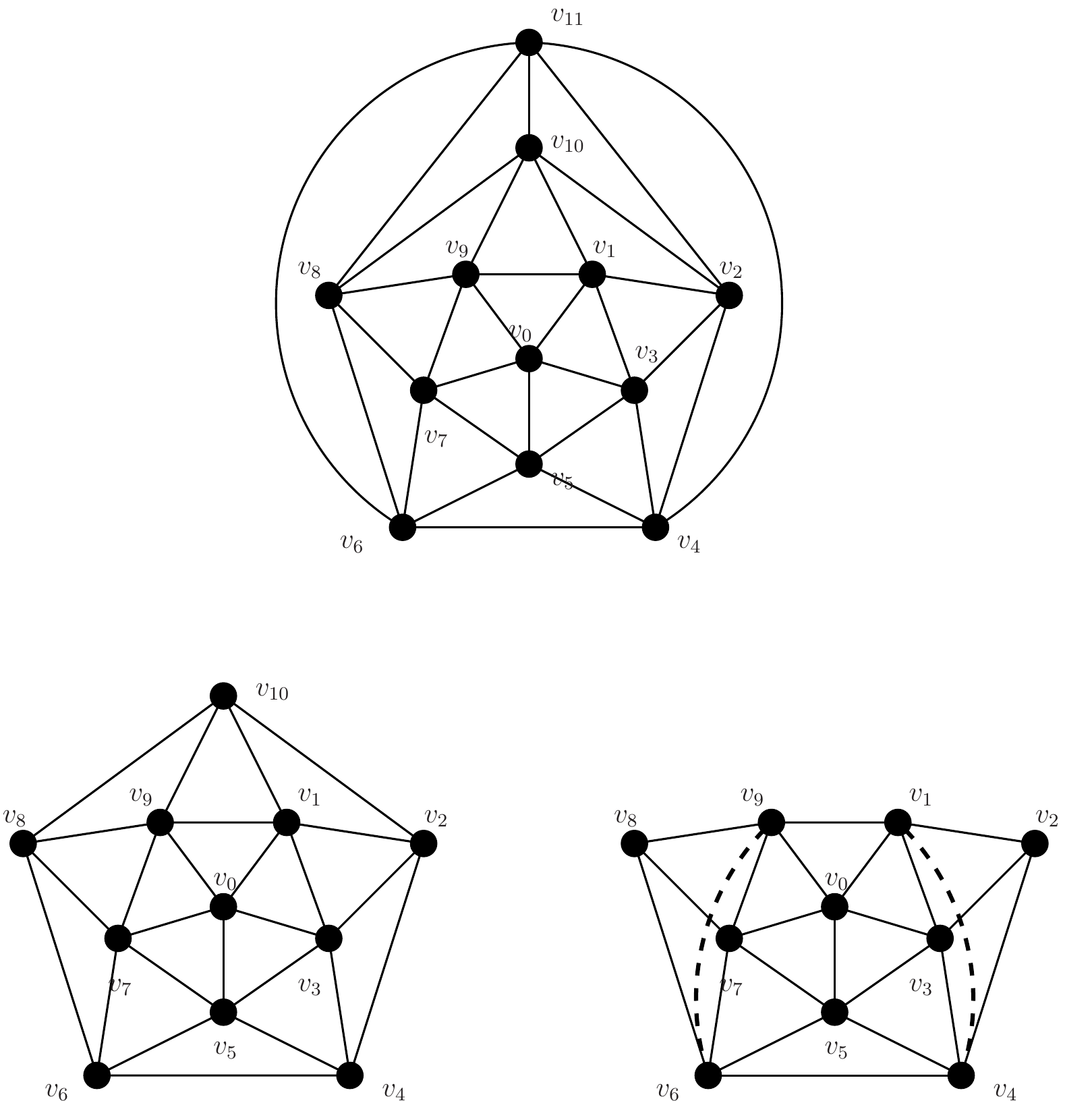}
\end{center}
\caption{\small{The icosahedron $G_0$ (top), with its derivative graphs $G_1$ (left) and $G_2\cup M$ (right).  In $G_2\cup M$, each of $\{v_1,v_4\}$ and $\{v_6,v_9\}$ is a nonadjacent pair or is in $M$.}}
\label{fig:icosahedral}
\end{figure}

We obtain $G_1$ from $G_0$ by deleting $v_{11}$, and we obtain $G_2$ from $G_1$ by deleting $v_{10}$.  Note that if the edge set $M$ is a subset of $\{v_1v_4, v_6v_9\}$, then $M$ is a claw-neutral matching in $G_2\cup M$.  We say that $G'$ is an {\em icosahedral thickening} if it is a proper thickening of $G_0$ or $G_1$, or is a thickening of $G_2\cup M$ under some $M\subseteq \{v_1v_4, v_6v_9\}$.  Any icosahedral thickening $G'$ has $\alpha(G')=3$ and $\chi(\overline{G'})=4$.

%%%%%%%%%%%%%%%%%%%%%%%%%%%%%%%%%%%%%%%%%%%%%%%%%%%%%%%%%%%%%%%%%%%%%%%%%%%%%%%%
%%%%%%%%%%%%%%%%%%%%%%%%%%%%%%%%%%%%%%%%%%%%%%%%%%%%%%%%%%%%%%%%%%%%%%%%%%%%%%%%
%%%%%%%%%%%%%%%%%%%%%%%%%%%%%%%%%%%%%%%%%%%%%%%%%%%%%%%%%%%%%%%%%%%%%%%%%%%%%%%%
%%%%%%%%%%%%%%%%%%%%%%%%%%%%%%%%%%%%%%%%%%%%%%%%%%%%%%%%%%%%%%%%%%%%%%%%%%%%%%%%
%%%%%%%%%%%%%%%%%%%%%%%%%%%%%%%%%%%%%%%%%%%%%%%%%%%%%%%%%%%%%%%%%%%%%%%%%%%%%%%%
\subsection{Colouring icosahedral thickenings}

We now prove that any icosahedral thickening satisfies $\chi(G)\leq \gamma_\ell(G)$.  To do so we remove triads from a supposed minimum counterexample, so first we need to consider induced subgraphs of icosahedral thickenings.

\begin{lemma}\label{lem:icossubgraph}
Let $G$ be an icosahedral thickening.  Then any skeletal induced subgraph $G'$ of $G$ is an icosahedral thickening or is three-cliqued or contains a clique cutset or admits a canonical linear interval 2-join.
\end{lemma}

The proof of this lemma is straightforward but technical, and we leave it to the end of this section.  This lemma allows us to prove the desired result:

\begin{theorem}\label{thm:icosahedral}
Suppose $G$ is an induced subgraph of an icosahedral thickening.  Then $\chi(G)\leq \gamma_\ell(G)$.
\end{theorem}

\begin{proof}
Let $G$ be a minimum counterexample to the theorem.  By Lemma \ref{lem:icossubgraph} we know $G$ is an icosahedral thickening or contains a clique cutset or is three-cliqued or admits a canonical interval 2-join.  But $G$ is vertex-critical so it cannot contain a clique cutset.  Lemma \ref{lem:quasilinemce2} (proved in \cite{kingthesis} and \cite{chudnovskykps12}) and Theorem \ref{thm:local} tell us that $G$ is in fact an icosahedral thickening.

First suppose that $G$ is a proper thickening of $G_0$, the icosahedron.  We remark that the icosahedron is 4-colourable, so we remove four stable sets with union denoted by $X$ containing exactly one vertex in $I(v_i)$ for every vertex $v_i$ of $G_0$.  When $X$ is removed, every remaining vertex $v$ in $G$ loses six neighbours (one of which is a twin), and since every maximal clique in $G$ corresponds to a triangle in $G_0$, $\omega(v)$ drops by three.  Thus $d(v)+\omega(v)$ drops by nine and it follows that $\gamma_\ell(G)$ drops by at least four, contradicting the minimality of $G$.

Now suppose that $G$ is a proper thickening of $G_1$ (see Figure \ref{fig:icosahedral}).  Again we remove one vertex from each $I(v_i)$, this time for $0\leq i\leq 10$, again using four stable sets.  When we remove the vertices, every remaining vertex loses at least five neighbours, one of which is a twin.  And as with $G_0$, every vertex $v$ of $G$ has $\omega(v)$ drop by three.  Thus $\gamma_\ell(G)$ drops by at least four, contradicting the minimality of $G$.

Finally suppose that $G$ is a thickening of $G_2\cup M$ under a matching $M$; we know that $M\subseteq \{v_1v_4, v_6v_9 \}$.  By minimality of $G$, $(I(v_1),I(v_4))$ and $(I(v_6),I(v_9))$ are skeletal homogeneous pairs of cliques.  We remove two stable sets with union $X$:  One intersects $I(v_1)$, $I(v_4)$, and $I(v_7)$ and intersects $\Omega(v_1v_4)$ if it is not empty.  The other intersects $I(v_3)$, $I(v_6)$, and $I(v_9)$ and intersects $\Omega(v_6v_9)$ if it is not empty.  These stable sets must exist because neither $I(v_1)\cup I(v_4)$ nor $I(v_6)\cup I(v_9)$ is a clique.

It is straightforward to confirm that $X$ intersects every maximal clique in $G$, so $\omega(v)$ drops by at least one for every $v\in G-X$, thus $\gamma_\ell(v)$ drops by at least two for any vertex with three neighbours in $X$.  Observe that any vertex in $G-X$ with only two neighbours in $X$ must be in $(I(v_1)\cup I(v_4))  \setminus \Omega(v_1v_4)$ or $(I(v_6)\cup I(v_9))\setminus \Omega(v_1v_4)$.  Furthermore, every such vertex has a twin in $X$.  Thus we can easily confirm that $\omega(v)$ drops by two for every such vertex.  So for any $v$ with only two neighbours in $X$, $\omega(v)$ drops by two.  Therefore $\gamma_\ell(G-X)\leq \gamma_\ell(G)-2$, contradicting the minimality of $G$.  This completes the proof.
\end{proof}

We now prove Lemma \ref{lem:icossubgraph}.

\begin{proof}[Proof of Lemma \ref{lem:icossubgraph}]
Suppose first that $G$ is a thickening of $G_2\cup M$ under $M\subseteq \{v_1v_4, v_6v_9\}$ (see Figure \ref{fig:icosahedral}).  If $G'$ has $I(v_i)$ nonempty for all $0\leq i\leq 9$ then clearly $G_2$ is an icosahedral thickening unless $I(v_1)\cup I(v_4)$ or $I(v_6)\cup I(v_9)$ becomes a clique, in which case we have a clique cutset.  If $I(v_i)$ is empty for some $i\in \{0,2,5,8\}$ then it is not hard to check that $G'$ is three-cliqued.  If $I(v_i)$ is empty for some $i\in \{1,4,6,9\}$ then $G'$ contains a clique cutset.  If none of these aforementioned sets $I(v_i)$ is empty but one of $I(v_3)$ and $I(v_7)$ is empty, then $G'$ admits a canonical interval 2-join.  For example, if $G'$ is reached from $G$ by deleting $I(v_3)$, then $((I(v_0)\cup I(v_9),I(v_5)\cup I(v_6)),(I(v_1),I(v_4)))$ is a canonical interval 2-join.

Now suppose that $G$ is a thickening of $G_1$.  Obviously $G'$ is an icosahedral thickening if $I(v_i)$ is nonempty for all $0\leq i\leq 10$.  If $I(v_i)$ is empty for any $i\in \{2,4,6,8,10 \}$ then the desired result follows from the previous paragraph.  If $I(v_0)$ is empty then $G'$ is a circular interval graph.  If $I(v_i)$ is empty for some $i\in \{1,3,5,7,9\}$ then it is easy to see from Figure \ref{fig:icosahedral} that $G'$ admits a canonical interval 2-join or a clique cutset.

Finally, suppose that $G$ is a thickening of $G_0$.  If $G'$ has any $I(v_i)$ empty for $0\leq i\leq 11$ then the desired result follows from the previous two paragraphs.  Otherwise $G'$ is clearly a thickening of $G_0$.  This completes the proof.
\end{proof}

%%%%%%%%%%%%%%%%%%%%%%%%%%%%%%%%%%%%%%%%%%%%%%%%%%%%%%%%%%%%%%%%%%%%%%%%%%%%%%%%
%%%%%%%%%%%%%%%%%%%%%%%%%%%%%%%%%%%%%%%%%%%%%%%%%%%%%%%%%%%%%%%%%%%%%%%%%%%%%%%%
%%%%%%%%%%%%%%%%%%%%%%%%%%%%%%%%%%%%%%%%%%%%%%%%%%%%%%%%%%%%%%%%%%%%%%%%%%%%%%%%
%%%%%%%%%%%%%%%%%%%%%%%%%%%%%%%%%%%%%%%%%%%%%%%%%%%%%%%%%%%%%%%%%%%%%%%%%%%%%%%%
%%%%%%%%%%%%%%%%%%%%%%%%%%%%%%%%%%%%%%%%%%%%%%%%%%%%%%%%%%%%%%%%%%%%%%%%%%%%%%%%
%%%%%%%%%%%%%%%%%%%%%%%%%%%%%%%%%%%%%%%%%%%%%%%%%%%%%%%%%%%%%%%%%%%%%%%%%%%%%%%%
%%%%%%%%%%%%%%%%%%%%%%%%%%%%%%%%%%%%%%%%%%%%%%%%%%%%%%%%%%%%%%%%%%%%%%%%%%%%%%%%
\subsection{Proving the main result}

%%%%%%%%%%%%%%%%%%%%%%%%%%%%%%%%%%%%%%%%%%%%%%%%%%%%%%%%%%%%%%%%%%%%%%%%%%%%%%%%
%%%%%%%%%%%%%%%%%%%%%%%%%%%%%%%%%%%%%%%%%%%%%%%%%%%%%%%%%%%%%%%%%%%%%%%%%%%%%%%%
%%%%%%%%%%%%%%%%%%%%%%%%%%%%%%%%%%%%%%%%%%%%%%%%%%%%%%%%%%%%%%%%%%%%%%%%%%%%%%%%
\subsubsection{A decomposition theorem}

To prove Theorem \ref{thm:main} we use a decomposition theorem for claw-free graphs; it is a weakening of Theorem 7.2 in \cite{clawfree5}:

\begin{theorem}\label{thm:structure}
Let $G$ be a skeletal claw-free graph containing no clique cutset.  Then one of the following is true:
\begin{enumerate*}
\item $G$ is quasi-line
\item $G$ is an antiprismatic thickening
\item $G$ is three-cliqued
\item $\chi(\overline{G}) \geq 4$ and $G$ admits a canonical interval 2-join, an antihat 2-join, a strange 2-join, a pseudo-line 2-join, or a gear 2-join
\item $\chi(\overline{G}) \geq 4$ and $G$ is an icosahedral thickening.
\end{enumerate*}
\end{theorem}

Getting from Chudnovsky and Seymour's structure theorem for claw-free trigraphs to Theorem \ref{thm:structure} is complicated but not difficult.  Still, we owe some explanation to the reader who is unfamiliar with trigraphs.  First note that the structure of a graph is precisely the same as the structure of a trigraph in which no two (distinct) vertices are semiadjacent -- only the terminology differs.  The class of claw-free graphs is precisely the class of claw-free trigraphs in which no two vertices are semiadjacent.  As a warm-up, one can easily check that if a claw-free graph $G$ is a thickening of a trigraph $G'$, and $G'$ is the union of three strong cliques, then $G$ is a three-cliqued claw-free graph.  Next, check that every graph which is a thickening of a member of $\fs_3$ is quasi-line.  Similarly, any graph which is a thickening of a member of $\fs_1$ or $\fs_7$ is an icosahedral thickening or an antiprismatic thickening, respectively.

This leaves {\em non-trivial strip structures}, discussed in Section 7 of \cite{clawfree5}.  Noting that in \cite{clawfree5}, $\fz_0=\fz_1 \cup \ldots \cup \fz_{15}$, observe that if a graph $G$ is a thickening of a trigraph $G'$ admitting a non-trivial strip structure involving a strip in $\fz_{6} \cup \ldots \cup \fz_{15}$, then $G$ admits a clique cutset.  Suppose now that $G$ is a thickening of a trigraph $G'$ admitting a non-trivial strip structure involving a strip in $\fz_{2} \cup \ldots \cup \fz_{5}$.  Then $G$ admits an antihat 2-join (arising from $\fz_2$), or a strange 2-join (arising from $\fz_3$), or a pseudo-line 2-join (arising from $\fz_4$), or a gear 2-join (arising from $\fz_5$).  It now suffices to confirm that if $G$ is a thickening of a trigraph $G'$ admitting a non-trivial strip structure in which all strips are in $\fz_2$ or are trivial (i.e.\ $(J,Z)$ where $|V(J)|=3$ and $|Z|=2$), then $G$ is quasi-line.

%%%%%%%%%%%%%%%%%%%%%%%%%%%%%%%%%%%%%%%%%%%%%%%%%%%%%%%%%%%%%%%%%%%%%%%%%%%%%%%%
%%%%%%%%%%%%%%%%%%%%%%%%%%%%%%%%%%%%%%%%%%%%%%%%%%%%%%%%%%%%%%%%%%%%%%%%%%%%%%%%
%%%%%%%%%%%%%%%%%%%%%%%%%%%%%%%%%%%%%%%%%%%%%%%%%%%%%%%%%%%%%%%%%%%%%%%%%%%%%%%%
\subsubsection{Proof of Theorem \ref{thm:main}}

We can now combine our results to prove the second main result of the paper.

\begin{proof}[Proof of Theorem \ref{thm:main}]
Let $G$ be a minimum counterexample to the theorem; clearly $G$ cannot contain a clique cutset.  Theorem \ref{thm:skelhp} tells us that $G$ is skeletal.  Theorem \ref{thm:quasiline} tells us that $G$ is not quasi-line,  Theorem \ref{thm:antiprismaticlocal} tells us that $G$ is not an antiprismatic thickening, and Theorem \ref{thm:local} tells us that $G$ is not three-cliqued.  Lemma \ref{lem:composition} tells us that $G$ does not admit a canonical interval 2-join, an antihat 2-join, a strange 2-join, a gear 2-join, or a pseudo-line 2-join.  Theorem \ref{thm:icosahedral} tells us that $G$ is not an icosahedral thickening.  Therefore $G$ cannot exist.
\end{proof}

%%%%%%%%%%%%%%%%%%%%%%%%%%%%%%%%%%%%%%%%%%%%%%%%%%%%%%%%%%%%%%%%%%%%%%%%%%%%%%%%
%%%%%%%%%%%%%%%%%%%%%%%%%%%%%%%%%%%%%%%%%%%%%%%%%%%%%%%%%%%%%%%%%%%%%%%%%%%%%%%%
%%%%%%%%%%%%%%%%%%%%%%%%%%%%%%%%%%%%%%%%%%%%%%%%%%%%%%%%%%%%%%%%%%%%%%%%%%%%%%%%
%%%%%%%%%%%%%%%%%%%%%%%%%%%%%%%%%%%%%%%%%%%%%%%%%%%%%%%%%%%%%%%%%%%%%%%%%%%%%%%%
%%%%%%%%%%%%%%%%%%%%%%%%%%%%%%%%%%%%%%%%%%%%%%%%%%%%%%%%%%%%%%%%%%%%%%%%%%%%%%%%
%%%%%%%%%%%%%%%%%%%%%%%%%%%%%%%%%%%%%%%%%%%%%%%%%%%%%%%%%%%%%%%%%%%%%%%%%%%%%%%%
%%%%%%%%%%%%%%%%%%%%%%%%%%%%%%%%%%%%%%%%%%%%%%%%%%%%%%%%%%%%%%%%%%%%%%%%%%%%%%%%
%%%%%%%%%%%%%%%%%%%%%%%%%%%%%%%%%%%%%%%%%%%%%%%%%%%%%%%%%%%%%%%%%%%%%%%%%%%%%%%%
%%%%%%%%%%%%%%%%%%%%%%%%%%%%%%%%%%%%%%%%%%%%%%%%%%%%%%%%%%%%%%%%%%%%%%%%%%%%%%%%
%%%%%%%%%%%%%%%%%%%%%%%%%%%%%%%%%%%%%%%%%%%%%%%%%%%%%%%%%%%%%%%%%%%%%%%%%%%%%%%%
%%%%%%%%%%%%%%%%%%%%%%%%%%%%%%%%%%%%%%%%%%%%%%%%%%%%%%%%%%%%%%%%%%%%%%%%%%%%%%%%
%%%%%%%%%%%%%%%%%%%%%%%%%%%%%%%%%%%%%%%%%%%%%%%%%%%%%%%%%%%%%%%%%%%%%%%%%%%%%%%%
%%%%%%%%%%%%%%%%%%%%%%%%%%%%%%%%%%%%%%%%%%%%%%%%%%%%%%%%%%%%%%%%%%%%%%%%%%%%%%%%
%%%%%%%%%%%%%%%%%%%%%%%%%%%%%%%%%%%%%%%%%%%%%%%%%%%%%%%%%%%%%%%%%%%%%%%%%%%%%%%%
%%%%%%%%%%%%%%%%%%%%%%%%%%%%%%%%%%%%%%%%%%%%%%%%%%%%%%%%%%%%%%%%%%%%%%%%%%%%%%%%
\section{Algorithmic considerations}\label{sec:algorithmic}

We now show that our proofs of Theorems \ref{thm:main} and \ref{thm:local} yield polynomial time algorithms for $\gamma(G)$- and $\gamma_\ell(G)$-colouring $G$, respectively.

It is well known that we can restrict our attention to graphs containing no clique cutset -- see e.g.\ \cite{kingthesis} \S 3.4.3 for an explanation.  By Theorem \ref{thm:skelhp} we can restrict our attention to skeletal graphs.  Furthermore we can identify maximal sets of twin vertices (i.e.\ equivalence classes of the ``twin'' equivalence relation) in $G$ in polynomial time \cite{cournierh94}.  This immediately implies that we can recognize skeletal icosahedral thickenings in polynomial time.  We can easily check whether or not a triad in a  graph is good in polynomial time, so in polynomial time we can determine whether or not a graph contains a good triad by checking all triples of vertices.

If $G$ is an icosahedral thickening, then observe that since $G$ is skeletal there are at most 14 equivalence classes of twin vertices.  Therefore there are at most $14^3$ different types of stable sets.  We can formulate the problem of colouring $G$ as an integer program in which each variable represents the number of stable sets of a given type we use in the colouring.  Each variable has size at most $n$, so we can exhaustively solve the problem in $O(n^{14^3})$ time to find an optimal colouring of $G$ (following the proof of Theorem \ref{thm:icosahedral} yields a much more efficient $\gamma_\ell(G)$-colouring algorithm).

We now consider the problem of colouring three-cliqued claw-free graphs and antiprismatic thickenings.

%%%%%%%%%%%%%%%%%%%%%%%%%%%%%%%%%%%%%%%%%%%%%%%%%%%%%%%%%%%%%%%%%%%%%%%%%%%%%%%%
%%%%%%%%%%%%%%%%%%%%%%%%%%%%%%%%%%%%%%%%%%%%%%%%%%%%%%%%%%%%%%%%%%%%%%%%%%%%%%%%
%%%%%%%%%%%%%%%%%%%%%%%%%%%%%%%%%%%%%%%%%%%%%%%%%%%%%%%%%%%%%%%%%%%%%%%%%%%%%%%%
%%%%%%%%%%%%%%%%%%%%%%%%%%%%%%%%%%%%%%%%%%%%%%%%%%%%%%%%%%%%%%%%%%%%%%%%%%%%%%%%
%%%%%%%%%%%%%%%%%%%%%%%%%%%%%%%%%%%%%%%%%%%%%%%%%%%%%%%%%%%%%%%%%%%%%%%%%%%%%%%%
%%%%%%%%%%%%%%%%%%%%%%%%%%%%%%%%%%%%%%%%%%%%%%%%%%%%%%%%%%%%%%%%%%%%%%%%%%%%%%%%
%%%%%%%%%%%%%%%%%%%%%%%%%%%%%%%%%%%%%%%%%%%%%%%%%%%%%%%%%%%%%%%%%%%%%%%%%%%%%%%%
%%%%%%%%%%%%%%%%%%%%%%%%%%%%%%%%%%%%%%%%%%%%%%%%%%%%%%%%%%%%%%%%%%%%%%%%%%%%%%%%
\subsection{Antiprismatic thickenings}

We already know that any skeletal antiprismatic thickening contains a good triad, but we have not proven that reducing a nonskeletal homogeneous pair of cliques in an antiprismatic thickening leaves another antiprismatic thickening.  It is enough to appeal to an easy result on {\em antiprismatic trigraphs}, which are defined in \cite{clawfree5}, Section 3.  The proof is trivial but in the language of trigraphs.

\begin{lemma}
If an antiprismatic trigraph $G$ is a thickening of a trigraph $H$, then $H$ is antiprismatic.
\end{lemma}
\begin{proof}
Assume for a contradiction that either $H$ contains a claw (in the trigraph sense) or that $H$ contains four vertices among which at most one pair is strongly adjacent.  In either case, the thickening from $H$ to $G$ provides us with a claw in $G$ or a set of four vertices of $G$, among which at most one pair is strongly adjacent, a contradiction.
\end{proof}

\begin{corollary}
If $(A,B)$ is a nonskeletal homogeneous pair of cliques in an antiprismatic graph $G$, and we obtain the graph $G'$ by contracting $A$ and $B$ down to adjacent vertices $a$ and $b$ respectively, then $G'$ is antiprismatic and $G$ is a thickening of $G'$ under a matching that contains $ab$.
\end{corollary}

As a consequence of this corollary, reducing a nonskeletal homogeneous pair of cliques in an antiprismatic thickening will leave us with an antiprismatic thickening.  We may therefore colour antiprismatic thickenings in the obvious way.

\begin{theorem}
Given an antiprismatic thickening $G$, we can find a $\gamma_\ell(G)$-colouring of $G$ in polynomial time.
\end{theorem}
\begin{proof}
Starting with $G$, we repeatedly apply Lemma \ref{lem:reduction}, removing edges to reach a subgraph $G'$ such that a $k$-colouring of $G'$ gives us a $k$-colouring of $G$ for any $k$.  As we just showed, $G'$ is an antiprismatic thickening, and therefore contains either no triad, in which case we can easily colour $G'$ and therefore $G$ in polynomial time, or contains a good triad $T$.  In the latter case, we remove $T$ and recursively $\gamma_\ell(G-T)$-colour $G-T$, noting that $G-T$ is again an antiprismatic thickening.

Since we can perform the recursion steps in polynomial time and there are $O(m)$ possible steps, we can $\gamma_\ell(G)$-colour $G$ in polynomial time.
\end{proof}

%%%%%%%%%%%%%%%%%%%%%%%%%%%%%%%%%%%%%%%%%%%%%%%%%%%%%%%%%%%%%%%%%%%%%%%%%%%%%%%%
%%%%%%%%%%%%%%%%%%%%%%%%%%%%%%%%%%%%%%%%%%%%%%%%%%%%%%%%%%%%%%%%%%%%%%%%%%%%%%%%
%%%%%%%%%%%%%%%%%%%%%%%%%%%%%%%%%%%%%%%%%%%%%%%%%%%%%%%%%%%%%%%%%%%%%%%%%%%%%%%%
%%%%%%%%%%%%%%%%%%%%%%%%%%%%%%%%%%%%%%%%%%%%%%%%%%%%%%%%%%%%%%%%%%%%%%%%%%%%%%%%
%%%%%%%%%%%%%%%%%%%%%%%%%%%%%%%%%%%%%%%%%%%%%%%%%%%%%%%%%%%%%%%%%%%%%%%%%%%%%%%%
%%%%%%%%%%%%%%%%%%%%%%%%%%%%%%%%%%%%%%%%%%%%%%%%%%%%%%%%%%%%%%%%%%%%%%%%%%%%%%%%
%%%%%%%%%%%%%%%%%%%%%%%%%%%%%%%%%%%%%%%%%%%%%%%%%%%%%%%%%%%%%%%%%%%%%%%%%%%%%%%%
%%%%%%%%%%%%%%%%%%%%%%%%%%%%%%%%%%%%%%%%%%%%%%%%%%%%%%%%%%%%%%%%%%%%%%%%%%%%%%%%
\subsection{Three-cliqued graphs}

Maffray and Preissmann proved that it is $\mathit{NP}$-complete to decide whether or not a triangle-free graph is three-colourable \cite{maffrayp96}.  Consequently it is $\mathit{NP}$-complete to decide whether or not a claw-free graph is three-cliqued.  This makes dealing with three-cliqued claw-free graphs a slightly delicate issue.  However, consider a claw-free graph $G$.  If $\alpha(G)\leq 2$ we know we can optimally colour it in polynomial time.  We will show that if $\alpha(G)=3$, then in polynomial time we can either $\gamma_\ell(G)$-colour $G$, or determine that $G$ is not three-cliqued.

\begin{lemma}\label{lem:3algo}
Let $G$ be a skeletal claw-free graph with $\alpha(G)=3$, and suppose $G$ contains no good triad.  Then in polynomial time we can $\gamma_\ell(G)$-colour $G$ or determine that $G$ is not three-cliqued.
\end{lemma}

\begin{proof}
We define the {\em triad graph} $t(G)$ of $G$.  We let $V(t(G))=V(G)$, and two vertices are adjacent in $t(G)$ precisely if some triad in $G$ contains both of them.  We can easily find the components of $t(G)$ in polynomial time; there is at least one which is not a singleton.

Suppose first that $G$ is three-cliqued.  Then it admits a hex-join into terms $(G_1,A_1,B_1,C_1)$ and (possibly empty) $(G_2,A_2,B_2,C_2)$ such that $G_1$ is minimal and contains a triad.  Since $G$ contains no good triad, it follows from the proofs of Lemmas \ref{lem:min32}, \ref{lem:min33}, \ref{lem:min35}, and \ref{lem:min36} that $(G_1,A_1,B_1,C_1)$ is in $\TTC_1$.  Furthermore the graph from which $G_1$ arises, i.e.\ $H$ such that $G_1=L(H)$, has more than three centres and hence more than six vertices, otherwise $G$ would contain a good triad.

Suppose first that $G$ is three-cliqued, and let $X$ be a non-singleton component of $t(G)$.  Then it admits a hex-join into terms $(G_1,A_1,B_1,C_1)$ and (possibly empty) $(G_2,A_2,B_2,C_2)$ such that $G_1$ is minimal and contains a triad.  Since $G$ contains no good triad, it follows from the proofs of Lemmas \ref{lem:min32}, \ref{lem:min33}, \ref{lem:min35}, and \ref{lem:min36} that $(G_1,A_1,B_1,C_1)$ is in $\TTC_1$.  Furthermore the graph from which $G_1$ arises, i.e.\ $H$ such that $G_1=L(H)$, has more than three centres and hence more than six vertices, otherwise $G$ would contain a good triad.

We claim that there is a component $X$ of $t(G)$ such that $X = V(G_1)$.  First note that any component of $t(G)$ is either contained in $V(G_1)$ or disjoint from $V(G_1)$, since no triad can span both sides of a hex-join.  Since every vertex of $G_1$ is in a triad, $V(G_1)$ is covered by non-singleton components of $t(G)$.  In the case that $G_1$ contains a simplicial vertex $v$, it is easy to show that $V(G_1)$ is a component of $t(G)$: since every vertex is in a triad, every vertex not in $N(v)$ (in $G_1$) is in a vertex with $v$, and every vertex in $N(v)$ is in a triad, which is necessarily not contained in $N(v)\cup \{v\}$.  So we may assume that $G_1$ contains no simplicial vertices.

Now it is sufficient to prove that every vertex in $A_1$ is in the same component of $t(G)$.  Bearing in mind the structure of the bipartite multigraph $H$ from which $G_1$ is constructed, the fact that $G_1$ has no simplicial vertex implies that the simple graph underlying $H$ is a complete bipartite graph minus a matching.  Therefore given two distinct edges of $H$ incident to $a$, there must be two triads in $G_1$ containing their corresponding vertices, such that the triads intersect in two vertices.  Therefore there is a component $X$ of $t(G)$ such that $X=V(G_1)$.

For every component $X$ of $t(G)$ we can test $G[X]$ for membership in $\TTC_1$ in polynomial time, because any graph in $\TTC_1$ is a proper thickening of a line graph of a specific bipartite graph $H$.  In particular we can find $(G_1,A_1,B_1,C_1)$ efficiently, because we can find $H$ efficiently and the definition of $\TTC_1$ implies that the choice of vertices $\{a,b,c\}$ of $H$ is unique.  Thus since $G_1$ is a term in a hex-join, we can determine $A_2$, $B_2$, and $C_2$ by taking a vertex in $G_2$ and looking at its neighbourhood in $G_1$, assuming that $G$ is three-cliqued.

We now proceed as in the proof of Lemma \ref{lem:min31}.  With our base graph $H$ of $(G_1,A_1,B_1,C_1)$ in hand, it is not hard to see that we can decide which action is necessary in polynomial time.  In each case we find a triad whose removal is guaranteed to lower $\gamma_\ell(G)$ or we remove edges from $G$ to reach a proper subgraph $G'$ such that $\chi(G')=\chi(G)$.  From the proof of Lemma \ref{lem:min31} it is clear that we can find $G'$ in polynomial time, and given a $k$-colouring of $G'$ we can find a $k$-colouring of $G$ in polynomial time.  We can recursively $\gamma_\ell(G)$-colour $G'$ in polynomial time, possibly appealing to the fact that we can find good triads efficiently.

Now suppose $G$ is not three-cliqued, which must be the case if no component of $t(G)$ induces a subgraph in $\TTC_1$.  If there is a component $X$ of $t(G)$ such that $G[X]$ is in $\TTC_1$, then again we have a unique choice of $\{a,b,c\}$ in $H$ and a unique expression of $G[X]$ as $(G_1,A_1,B_1,C_1)$.  Let $A_2$ be the set of vertices in $G-X$ which are complete to $A_1\cup B_1$; we define $B_2$ and $C_2$ accordingly.  Since $G$ is not three-cliqued, either $A_2$, $B_2$, and $C_2$ do not partition the vertices of $G-X$, or they are not all cliques.  Either way we can determine this in polynomial time.
\end{proof}

Using these two lemmas we can prove the desired result:

\begin{theorem}
Let $G$ be a claw-free graph with $\alpha(G)\geq 3$.  Then in polynomial time we can either $\gamma_\ell(G)$-colour $G$ or determine that $\chi(\gbar)\geq 4$.
\end{theorem}

\begin{proof}
By Theorem \ref{thm:skelhp} we can assume $G$ is skeletal.  If $G$ contains a good triad $T$, we can find $T$ in polynomial time and recursively $\gamma_\ell(G)-1$ colour $G-T$, or determine that $\chi(\gbar - T)\geq 4$.  If $G$ does not contain a good triad, then the result follows immediately from Lemma \ref{lem:3algo}.
\end{proof}

%%%%%%%%%%%%%%%%%%%%%%%%%%%%%%%%%%%%%%%%%%%%%%%%%%%%%%%%%%%%%%%%%%%%%%%%%%%%%%%%
%%%%%%%%%%%%%%%%%%%%%%%%%%%%%%%%%%%%%%%%%%%%%%%%%%%%%%%%%%%%%%%%%%%%%%%%%%%%%%%%
%%%%%%%%%%%%%%%%%%%%%%%%%%%%%%%%%%%%%%%%%%%%%%%%%%%%%%%%%%%%%%%%%%%%%%%%%%%%%%%%
%%%%%%%%%%%%%%%%%%%%%%%%%%%%%%%%%%%%%%%%%%%%%%%%%%%%%%%%%%%%%%%%%%%%%%%%%%%%%%%%
%%%%%%%%%%%%%%%%%%%%%%%%%%%%%%%%%%%%%%%%%%%%%%%%%%%%%%%%%%%%%%%%%%%%%%%%%%%%%%%%
%%%%%%%%%%%%%%%%%%%%%%%%%%%%%%%%%%%%%%%%%%%%%%%%%%%%%%%%%%%%%%%%%%%%%%%%%%%%%%%%
%%%%%%%%%%%%%%%%%%%%%%%%%%%%%%%%%%%%%%%%%%%%%%%%%%%%%%%%%%%%%%%%%%%%%%%%%%%%%%%%
%%%%%%%%%%%%%%%%%%%%%%%%%%%%%%%%%%%%%%%%%%%%%%%%%%%%%%%%%%%%%%%%%%%%%%%%%%%%%%%%
\subsection{Graphs that are not three-cliqued}

By Theorem \ref{thm:structure}, if $G$ is a skeletal claw-free graph that is not three-cliqued and does not contain a clique cutset, then one of the following applies:
\begin{enumerate*}
\item $G$ is an antiprismatic thickening
\item $G$ is an icosahedral thickening
\item $G$ is quasi-line
\item $G$ admits a canonical interval 2-join or an antihat 2-join or a pseudo-line 2-join or a strange 2-join or a gear 2-join.
\end{enumerate*}

We already know how to deal with the first three cases efficiently, either by colouring in polynomial time or reducing to a smaller colouring problem.  For each of the four latter types of generalized 2-join, of the form  $((X_1,Y_1),(X_2,Y_2))$, there is a $W_5$ in $G_2$ whose neighbourhood contains $G_2$.  Given the correct choice of a $W_5$ in $G$, it is straightforward to find an appropriate generalized 2-join separating $G_1$ from $G_2$ in polynomial time (see \cite{kingthesis} \S 8.2 for further details).  There are $O(n^6)$ 5-wheels in $G$, so we can find such a generalized 2-join in polynomial time.

Since $G$ is skeletal, we can easily check whether or not $G_2$ is a gear strip in polynomial time: a skeletal gear strip has at most twelve equivalence classes of twin vertices.  So assume that we have an antihat 2-join or a pseudo-line 2-join or a strange 2-join.  We can easily check for a strange 2-join similarly to checking for a gear 2-join.  Checking if we have an antihat 2-join is straightforward once we determine the adjacency between $X_2$ and $Y_2$.  Otherwise we have a pseudo-line 2-join.  In this case, $I(e_1)$ and $I(e_2)$ are precisely those vertices in $X_2$ and $Y_2$ respectively that are complete to $G_2-X_2-Y_2$.  Furthermore, adding all edges between $I(e_1)$ and $I(e_2)$ leaves us with a line graph, the structure of which we can easily determine.  Thus we can find these desired generalized 2-joins in polynomial time.

To reduce on these generalized 2-joins, we now consider the proof of Lemmas \ref{lem:compantihat}, \ref{lem:compstrange}, \ref{lem:compgear}, and \ref{lem:composition}.  We do one of two things: reduce the size of the graph and apply induction, or complete the $l$-colouring of $G$ in one step.  Just as with Lemma \ref{lem:quasilinemce2} in \cite{kingr08}, showing that we can do this in polynomial time is straightforward given the proof of the lemma.  Thus we get the desired algorithmic result:

\begin{theorem}
For any claw-free graph $G$, we can $\gamma(G)$-colour $G$ in polynomial time.
\end{theorem}

%%%%%%%%%%%%%%%%%%%%%%%%%%%%%%%%%%%%%%%%%%%%%%%%%%%%%%%%%%%%%%%%%%%%%%%%%%%%%%%%
%%%%%%%%%%%%%%%%%%%%%%%%%%%%%%%%%%%%%%%%%%%%%%%%%%%%%%%%%%%%%%%%%%%%%%%%%%%%%%%%
%%%%%%%%%%%%%%%%%%%%%%%%%%%%%%%%%%%%%%%%%%%%%%%%%%%%%%%%%%%%%%%%%%%%%%%%%%%%%%%%
%%%%%%%%%%%%%%%%%%%%%%%%%%%%%%%%%%%%%%%%%%%%%%%%%%%%%%%%%%%%%%%%%%%%%%%%%%%%%%%%
%%%%%%%%%%%%%%%%%%%%%%%%%%%%%%%%%%%%%%%%%%%%%%%%%%%%%%%%%%%%%%%%%%%%%%%%%%%%%%%%
%%%%%%%%%%%%%%%%%%%%%%%%%%%%%%%%%%%%%%%%%%%%%%%%%%%%%%%%%%%%%%%%%%%%%%%%%%%%%%%%
%%%%%%%%%%%%%%%%%%%%%%%%%%%%%%%%%%%%%%%%%%%%%%%%%%%%%%%%%%%%%%%%%%%%%%%%%%%%%%%%
%%%%%%%%%%%%%%%%%%%%%%%%%%%%%%%%%%%%%%%%%%%%%%%%%%%%%%%%%%%%%%%%%%%%%%%%%%%%%%%%
%%%%%%%%%%%%%%%%%%%%%%%%%%%%%%%%%%%%%%%%%%%%%%%%%%%%%%%%%%%%%%%%%%%%%%%%%%%%%%%%
%%%%%%%%%%%%%%%%%%%%%%%%%%%%%%%%%%%%%%%%%%%%%%%%%%%%%%%%%%%%%%%%%%%%%%%%%%%%%%%%
%%%%%%%%%%%%%%%%%%%%%%%%%%%%%%%%%%%%%%%%%%%%%%%%%%%%%%%%%%%%%%%%%%%%%%%%%%%%%%%%
%%%%%%%%%%%%%%%%%%%%%%%%%%%%%%%%%%%%%%%%%%%%%%%%%%%%%%%%%%%%%%%%%%%%%%%%%%%%%%%%
%%%%%%%%%%%%%%%%%%%%%%%%%%%%%%%%%%%%%%%%%%%%%%%%%%%%%%%%%%%%%%%%%%%%%%%%%%%%%%%%
\section{Proofs on homogeneous pairs of cliques}\label{sec:lemmas}

Finally, we give the postponed proofs of Lemmas \ref{lem:reduction} and \ref{lem:findhp}.

%%%%%%%%%%%%%%%%%%%%%%%%%%%%%%%%%%%%%%%%%%%%%%%%%%%%%%%%%%%%%%%%%%%%%%%%%%%%%%%%
%%%%%%%%%%%%%%%%%%%%%%%%%%%%%%%%%%%%%%%%%%%%%%%%%%%%%%%%%%%%%%%%%%%%%%%%%%%%%%%%
%%%%%%%%%%%%%%%%%%%%%%%%%%%%%%%%%%%%%%%%%%%%%%%%%%%%%%%%%%%%%%%%%%%%%%%%%%%%%%%%
%%%%%%%%%%%%%%%%%%%%%%%%%%%%%%%%%%%%%%%%%%%%%%%%%%%%%%%%%%%%%%%%%%%%%%%%%%%%%%%%
\subsection{Reducing on a nonskeletal homogeneous pair of cliques}

We now prove Lemma \ref{lem:reduction}, which is a straightforward generalization of Lemma 9 in \cite{kingr08}.  This tells us exactly how we reduce on a nonskeletal homogeneous pair of cliques $(A,B)$ and how we can manipulate colourings on $(A,B)$.

\begin{proof}[Proof of Lemma \ref{lem:reduction}]
Assume $|A| \geq |B|$.  We can find a maximum clique $X$ of $G[A\cup B]$ in $O(n^{5/2})$ time, choosing $X$ to be $A$ if $A$ is a maximum clique.  To construct $G'$ from $G$, we remove precisely the edges between $A$ and $B$ that are not in $X$.  Clearly $\omega(G'[A\cup B]) = \omega(G[A \cup B]) = |X|$, and $(A,B)$ is a skeletal homogeneous pair of cliques in $G'$, so (1) holds.  Since $(A,B)$ is not skeletal, $G'$ is a proper subgraph of $G$.  We can find $G'$ in $O(n^{5/2})$ time because we can find $X$ in $O(n^{5/2})$ time \cite{hopcroftk73}.

We must prove that $G'$ is claw-free.  Suppose there is a vertex $v$ seeing three mutually nonadjacent vertices $a, b, c$ in $G'$.  Then without loss of generality, $a \in A$, $b \in B$, and $c \notin A \cup B$ since $G$ is claw-free.  Since $c$ sees neither $a$ nor $b$ in $G'$, $c$ sees nothing in $A \cup B$ in $G$.  It follows that $v \notin A\cup B$, so $v$ sees all of $A \cup B$ in $G$.  Therefore since $A$ and $B$ are not complete to each other in $G$, $G$ contains a claw centred at $v$, a contradiction.  So $G'$ is claw-free.

Now suppose $G$ is quasi-line; we must show that $G'$ is quasi-line.  Suppose a vertex $v$ is not bisimplicial in $G'$ and let $(S,T)$ be a partitioning of $N_G(v)$ into two cliques.  If $v$ has a neighbour $w \in S \setminus (A\cup B)$ that sees $A$ but not $B$, then $B \subseteq T$ and thus $S \cup A$ and $T \setminus A$ are two cliques covering $N_{G'}(v)$ in $G'$.  By symmetry we can assume that if no such $w$ exists then all of $N_{G'}(v) \setminus (A \cup B)$ sees $A \cup B$, therefore $(S \cup A) \setminus B$ and $(T\cup B)\setminus A$ are two cliques covering $N_{G'}(v)$ in $G'$.  Therefore $G'$ is quasi-line if $G$ is quasi-line.  This proves (2).

Let $c_{G'}$ be a proper colouring of $G'$ using $k \geq \chi(G')$ colours.  Since $(A,B)$ is a homogeneous pair, to construct a $k$-colouring of $G$ it is enough to find a colouring of $G[A\cup B]$ that uses the same set of colours as $c_{G'}$ on $A$ and on $B$.  We can do this in $O(n^{5/2})$ time because the number of colours which appear on both $A$ and $B$ in the colouring of $G'$ is at most the maximum size of a matching in $\gbar'$, which is the same as the size of a maximum matching in $\gbar$, i.e.\ $|(A \cup B)-X|$.

Since $G[A\cup B]$ is perfect, this extends to fractional colourings.  Specifically, for any $l\geq \omega(G[A\cup B])$ there is a fractional $l$-colouring of $G[A\cup B]$.  Suppose we have a fractional $k$-colouring of $G'$.  This colouring uses weight $l\geq \omega(G[A\cup B])$ on $A\cup B$, so since $(A,B)$ is a homogeneous pair of cliques we can combine the colouring of $G'-(A\cup B)=G-(A\cup B)$ with a fractional $l$-colouring of $G[A\cup B]$ to find a fractional $k$-colouring of $G$.  This proves (3).

Suppose that $G$ is three-cliqued.  To prove (4), it suffices to prove that $\gbar$ has a 3-colouring in which no colour appears in both $A$ and $B$.  If colour $c_1$ appears in both $A$ and $B$ then since $G[A\cup B]$ is not a clique, a second colour $c_2$ must appear in $A\cup B$; assume $c_2$ appears in $A$.  In this case we can give all vertices of $A$ colour $c_2$ and give all colours in $B$ colour $c_1$ and since $(A,B)$ is a homogeneous pair of cliques in $G$, the result is a valid 3-colouring of $\overline{G'}$.  This proves (4).
\end{proof}

%%%%%%%%%%%%%%%%%%%%%%%%%%%%%%%%%%%%%%%%%%%%%%%%%%%%%%%%%%%%%%%%%%%%%%%%%%%%%%%%
%%%%%%%%%%%%%%%%%%%%%%%%%%%%%%%%%%%%%%%%%%%%%%%%%%%%%%%%%%%%%%%%%%%%%%%%%%%%%%%%
%%%%%%%%%%%%%%%%%%%%%%%%%%%%%%%%%%%%%%%%%%%%%%%%%%%%%%%%%%%%%%%%%%%%%%%%%%%%%%%%
%%%%%%%%%%%%%%%%%%%%%%%%%%%%%%%%%%%%%%%%%%%%%%%%%%%%%%%%%%%%%%%%%%%%%%%%%%%%%%%%
\subsection{Finding homogeneous pairs of cliques}

Everett, Klein, and Reed gave a $O(mn^3)$ algorithm for finding homogeneous pairs \cite{everettkr97}, but did not consider the restricted case of homogeneous pairs of cliques.

In \cite{kingr08} we gave an $O(n^2m)$-time algorithm for finding a nonlinear homogeneous pair of cliques; in \cite{chudnovskyk11} (Proposition 10) the same algorithm is shown to be implementable in $O(m^2)$ time, even in the setting of trigraphs.

\begin{lemma}\label{lem:nonlinearhpalgorithmic}
For any graph $G$ we can find a nonlinear homogeneous pair of cliques in $G$, or determine that none exists, in $O(m^2)$ time.
\end{lemma}

Now we need to find linear nonskeletal homogeneous pairs of cliques.  First we prove a structural result that renders the task almost trivial.

\begin{lemma}
Suppose a graph $G$ contains a nonskeletal linear homogeneous pair of cliques $(A,B)$.  Then $G$ contains three nonempty disjoint cliques $A_1$, $A_2$, $B_1$ such that
\begin{itemize*}
\item $|A_1|\geq |B_1|$.
\item Each of $A_1$, $A_2$, and $B_1$ is either a singleton or a homogeneous clique.
\item $A_1\cup A_2$ is a clique, $A_2\cup B_1$ is a clique, and there are no edges between $A_1$ and $B_1$.
\item $(A_1\cup A_2, B_1)$ is a nonskeletal linear homogeneous pair of cliques.
\end{itemize*}
\end{lemma}

\begin{proof}
Suppose the vertices of $G[A\cup B]$ are $a_1,\ldots, a_{|A|}, b_1,\ldots, b_{|B|}$ in linear order.

By swapping the names of $A$ and $B$, we can make an important assumption without loss of generality:  Either $A$ is a maximum clique in $G[A\cup B]$, or there is a maximum clique $X$ of $G[A\cup B]$ and some vertex in $B$ that sees some but not all of $X\setminus B$.  If we cannot assume this, then $\omega(G[A\cup B]) > \max\{|A|,|B|\}$ and there is a unique maximum clique $X$ in $(G[A\cup B])$.  Furthermore since $(G[A\cup B])$ is a linear interval graph, no vertex in $A\setminus X$ (resp.\ $B\setminus X$) has a neighbour in $B$ (resp.\ $A$), contradicting the assumption that $(A,B)$ is nonskeletal.

To construct $A_1$, $A_2$, and $B_1$ we first select two vertices $a_p$ and $a_q$ in $A$.  Let $p$ be minimum such that $a_p$ is in a maximum clique $X$ of $G[A\cup B]$; note that $p=1$ if $\omega(G[A\cup B]) = |A|$.  We claim that there is some minimum $q>p$ such that $\tilde N(a_p) \subset \tilde N(a_q)$, i.e.\ $a_p$ and $a_q$ are not twins.  If $q$ does not exist then by our above assumption either (i) $X=A$ and there are no edges between $A$ and $B$, a contradiction since $(A,B)$ is nonskeletal, or (ii) $|X|>|A|$ and no vertex in $B$ sees some but not all of $X\setminus B$, a contradiction since in this case $X$ must be the unique maximum clique of $G[A\cup B]$.

Let $A_1$ be $a_p$ along with its twins, and let $B_1$ be the set of vertices that see that see $a_p$ but not $a_q$.  Clearly $B_1 \subseteq B$, and observe that $|A_1|\geq |B_1|$, otherwise $a_p$ would not be in a maximum clique in $G[A\cup B]$, whereas $a_q$ would.  So let $A_2$ be $q$ along with its twins.  An example is shown in Figure \ref{fig:linearhpoc}.

\begin{figure}
\begin{center}
\includegraphics[scale=0.6]{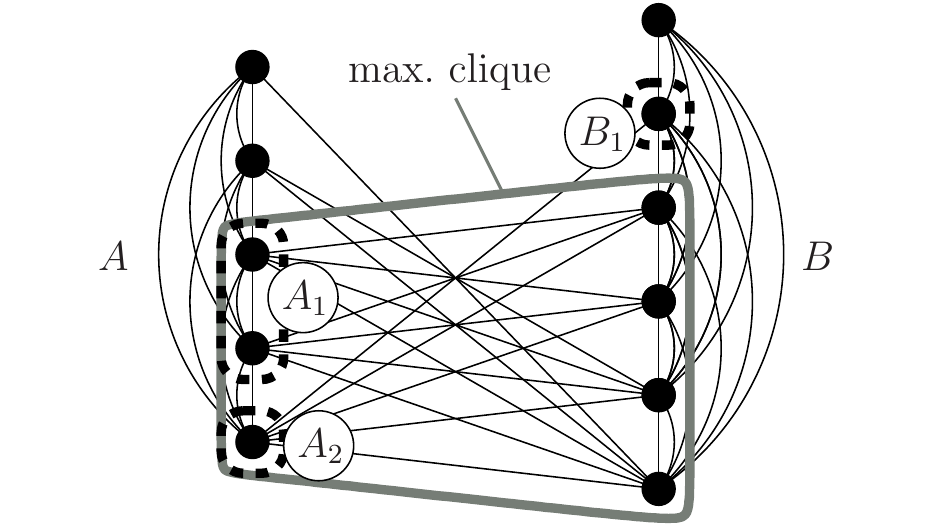}
\end{center}
\caption{If a linear homogeneous pair of cliques is not skeletal, we can find within it a homogeneous pair of cliques with a very specific structure.}
\label{fig:linearhpoc}
\end{figure}

To see that $(A_1\cup A_2, B_1)$ is a homogeneous pair of cliques, it is enough to show that $(\{a_p,a_q\},B_1)$ is a homogeneous pair of cliques.  By the structure of linear interval graphs, every vertex in $A\setminus (A_1\cup A_2)$ sees either all of $B_1$ or none of $B_1$, so $B_1$ is a singleton or a homogeneous clique.  Therefore $(\{a_p,a_q\},B_1)$ is a homogeneous pair of cliques, following from the fact that $(A,B)$ is a homogeneous pair of cliques.  Furthermore since $B_1$ is complete to $A_2$ and anticomplete to $A_1$, and $|A_1|\geq |B_1|$, it is easy to see that $(A_1\cup A_2, B_1)$ is a nonskeletal linear homogeneous pair of cliques (in particular, $A_1\cup A_2$ is a maximum clique in $G[A_1\cup A_2\cup B_1]$).
\end{proof}

Thus when searching for a linear nonskeletal homogeneous pair of cliques, we can focus on this specific structure.

\begin{lemma}\label{lem:nonskeletalhpalgorithmic}
Let $G$ be a graph containing no nonlinear homogeneous pair of cliques.  Then in $O(nm)$ time we can find some nonskeletal linear homogeneous pair of cliques $(A,B)$ in $G$, or determine that $G$ is skeletal.
\end{lemma}

Observe that Lemma \ref{lem:findhp} follows immediately from this lemma and Lemma \ref{lem:nonlinearhpalgorithmic}.

\begin{proof}
We find a nonskeletal homogeneous pair of cliques $(A,B)$ by finding the cliques $A_1$, $A_2$, and $B_1$ guaranteed by the previous lemma, as follows.  First we partition the vertices of $G$ into maximal homogeneous cliques in $O(m)$ time.  After that we just need to find three vertices $a_1$, $a_2$, and $b_1$ inducing a path such that $a_1$ has at least as many twins as $b_1$, no vertex sees $a_1$ but not $a_2$, and $b_1$ and its twins are the only vertices that see $a_2$ but not $a_1$.  We can easily do this in $O(nm)$ time by first guessing $b_1$, then deleting $b_1$ and checking for the appropriate resulting twins in $O(m)$ time.
\end{proof}

Finally, we remark that we can find a skeletal homogeneous pair of cliques in $O(m)$ time.  First we search for twins in time $O(m)$ -- twins immediately lead to a homogeneous pair of cliques if the graph has at least four vertices.  But the existence of a skeletal homogeneous pair $(A,B)$ implies the existence of twins:  Either $(A\cap \Omega(A,B), B\cap \Omega(A,B))$ is a homogeneous pair of cliques with all edges between them, or $(A,B)$ is a homogeneous pair of cliques with no edges between them.  Either case leads to twins.  With the results of this section, this implies the following:

\begin{theorem}
In $O(m^2)$ we can find a homogeneous pair of cliques in a graph or determine that none exists.
\end{theorem}

%%%%%%%%%%%%%%%%%%%%%%%%%%%%%%%%%%%%%%%%%%%%%%%%%%%%%%%%%%%%%%%%%%%%%%%%%%%%%%%%
%%%%%%%%%%%%%%%%%%%%%%%%%%%%%%%%%%%%%%%%%%%%%%%%%%%%%%%%%%%%%%%%%%%%%%%%%%%%%%%%
%%%%%%%%%%%%%%%%%%%%%%%%%%%%%%%%%%%%%%%%%%%%%%%%%%%%%%%%%%%%%%%%%%%%%%%%%%%%%%%%
%%%%%%%%%%%%%%%%%%%%%%%%%%%%%%%%%%%%%%%%%%%%%%%%%%%%%%%%%%%%%%%%%%%%%%%%%%%%%%%%
%%%%%%%%%%%%%%%%%%%%%%%%%%%%%%%%%%%%%%%%%%%%%%%%%%%%%%%%%%%%%%%%%%%%%%%%%%%%%%%%
%%%%%%%%%%%%%%%%%%%%%%%%%%%%%%%%%%%%%%%%%%%%%%%%%%%%%%%%%%%%%%%%%%%%%%%%%%%%%%%%
%%%%%%%%%%%%%%%%%%%%%%%%%%%%%%%%%%%%%%%%%%%%%%%%%%%%%%%%%%%%%%%%%%%%%%%%%%%%%%%%
%%%%%%%%%%%%%%%%%%%%%%%%%%%%%%%%%%%%%%%%%%%%%%%%%%%%%%%%%%%%%%%%%%%%%%%%%%%%%%%%
%%%%%%%%%%%%%%%%%%%%%%%%%%%%%%%%%%%%%%%%%%%%%%%%%%%%%%%%%%%%%%%%%%%%%%%%%%%%%%%%
%%%%%%%%%%%%%%%%%%%%%%%%%%%%%%%%%%%%%%%%%%%%%%%%%%%%%%%%%%%%%%%%%%%%%%%%%%%%%%%%
%%%%%%%%%%%%%%%%%%%%%%%%%%%%%%%%%%%%%%%%%%%%%%%%%%%%%%%%%%%%%%%%%%%%%%%%%%%%%%%%
%%%%%%%%%%%%%%%%%%%%%%%%%%%%%%%%%%%%%%%%%%%%%%%%%%%%%%%%%%%%%%%%%%%%%%%%%%%%%%%%
%%%%%%%%%%%%%%%%%%%%%%%%%%%%%%%%%%%%%%%%%%%%%%%%%%%%%%%%%%%%%%%%%%%%%%%%%%%%%%%%
\section{Conclusion}

The glaring open problem is Conjecture \ref{con:local} for claw-free graphs.  The only remaining case is that of compositions of pseudo-line strips, whose structure closely resembles that of line graphs.  It is possible that a refinement of the approach taken in \cite{chudnovskykps12} would do the trick.  For questions relating to more general local versions of the conjectures, we refer the reader to \cite{edwardsk12b}.

%%%%%%%%%%%%%%%%%%%%%%%%%%%%%%%%%%%%%%%%%%%%%%%%%%%%%%%%%%%%%%%%%%%%%%%%%%%%%%%%
%%%%%%%%%%%%%%%%%%%%%%%%%%%%%%%%%%%%%%%%%%%%%%%%%%%%%%%%%%%%%%%%%%%%%%%%%%%%%%%%
%%%%%%%%%%%%%%%%%%%%%%%%%%%%%%%%%%%%%%%%%%%%%%%%%%%%%%%%%%%%%%%%%%%%%%%%%%%%%%%%
%%%%%%%%%%%%%%%%%%%%%%%%%%%%%%%%%%%%%%%%%%%%%%%%%%%%%%%%%%%%%%%%%%%%%%%%%%%%%%%%
%%%%%%%%%%%%%%%%%%%%%%%%%%%%%%%%%%%%%%%%%%%%%%%%%%%%%%%%%%%%%%%%%%%%%%%%%%%%%%%%
%%%%%%%%%%%%%%%%%%%%%%%%%%%%%%%%%%%%%%%%%%%%%%%%%%%%%%%%%%%%%%%%%%%%%%%%%%%%%%%%
%%%%%%%%%%%%%%%%%%%%%%%%%%%%%%%%%%%%%%%%%%%%%%%%%%%%%%%%%%%%%%%%%%%%%%%%%%%%%%%%
%%%%%%%%%%%%%%%%%%%%%%%%%%%%%%%%%%%%%%%%%%%%%%%%%%%%%%%%%%%%%%%%%%%%%%%%%%%%%%%%
%%%%%%%%%%%%%%%%%%%%%%%%%%%%%%%%%%%%%%%%%%%%%%%%%%%%%%%%%%%%%%%%%%%%%%%%%%%%%%%%
%%%%%%%%%%%%%%%%%%%%%%%%%%%%%%%%%%%%%%%%%%%%%%%%%%%%%%%%%%%%%%%%%%%%%%%%%%%%%%%%
%%%%%%%%%%%%%%%%%%%%%%%%%%%%%%%%%%%%%%%%%%%%%%%%%%%%%%%%%%%%%%%%%%%%%%%%%%%%%%%%
%%%%%%%%%%%%%%%%%%%%%%%%%%%%%%%%%%%%%%%%%%%%%%%%%%%%%%%%%%%%%%%%%%%%%%%%%%%%%%%%
%%%%%%%%%%%%%%%%%%%%%%%%%%%%%%%%%%%%%%%%%%%%%%%%%%%%%%%%%%%%%%%%%%%%%%%%%%%%%%%%
\section{Acknowledgements}

The authors are very grateful to Maria Chudnovsky, Anna Galluccio, and Bruce Shepherd for their extremely helpful input on this work.

\bibliography{masterbib}
\end{document}